\documentclass[11pt]{amsart}

\usepackage[text={450pt,660pt},centering]{geometry}

\usepackage{color}
\usepackage{esint,amssymb}
\usepackage{graphicx}
\usepackage{MnSymbol}
\usepackage{amsmath, mathtools}
\usepackage[colorlinks=true, pdfstartview=FitV, linkcolor=blue, citecolor=blue, urlcolor=blue,pagebackref=false]{hyperref}
\usepackage{microtype}

\usepackage{bm}
\usepackage{scalerel} 
\usepackage{dsfont}
\usepackage{mathrsfs}
\usepackage[font={footnotesize}]{caption}
\usepackage{stmaryrd}
\usepackage{tikz}

\usepackage{enumitem}

\definecolor{darkgreen}{rgb}{0,0.5,0}
\definecolor{darkblue}{rgb}{0,0,0.7}
\definecolor{darkred}{rgb}{0.9,0.1,0.1}

\makeatletter
\newtheorem*{rep@theorem}{\rep@title}
\newcommand{\newreptheorem}[2]{%
\newenvironment{rep#1}[1]{%
 \def\rep@title{#2 \ref{##1}}%
 \begin{rep@theorem}}%
 {\end{rep@theorem}}}
\makeatother

\newtheorem{theorem}{Theorem}
\newreptheorem{theorem}{Theorem}
\newreptheorem{proposition}{Proposition}
\newreptheorem{lemma}{Lemma}
\newtheorem{proposition}{Proposition}
\newtheorem{lemma}[proposition]{Lemma}

\theoremstyle{remark}

\theoremstyle{definition}
\newtheorem{definition}[proposition]{Definition}
\newtheorem{remark}[proposition]{Remark}

\numberwithin{equation}{section}
\numberwithin{proposition}{section}

\newcommand{\Z}{\mathbb{Z}}
\newcommand{\N}{\mathbb{N}}
\newcommand{\R}{\mathbb{R}}

\newcommand{\E}{\mathbb{E}}
\renewcommand{\P}{\mathbb{P}}

\newcommand{\Zd}{\mathbb{Z}^d}
\newcommand{\Rd}{{\mathbb{R}^d}}

\newcommand{\di}{\mathrm{d}}
\newcommand{\inte}[1]{%
 {\kern0pt#1}^{\mathrm{o}}%
}

\newcommand{\ep}{\varepsilon}

\renewcommand{\a}{\mathbf{a}}

\renewcommand{\subset}{\subseteq}

\DeclareMathOperator{\var}{Var}

\DeclareMathOperator{\sign}{sign}

\DeclareMathOperator{\IV}{IV}

\DeclareMathOperator{\proj}{proj}

\renewcommand{\bar}{\overline}
\renewcommand{\tilde}{\widetilde}

\newcommand{\indc}{\mathds{1}}

\begin{document}

\title{Random-field random surfaces}

\author[P. Dario]{Paul Dario}
\address[P. Dario]{Universit\'e Paris-Est Cr\'eteil, Cr\'eteil, France}
\email{paul.dario@u-pec.fr}

\author[M. Harel]{Matan Harel}
\address[M. Harel]{Northeastern University, 360 Huntington Avenue, Boston, MA 02115}
\email{m.harel@northeastern.edu}

\author[R. Peled]{Ron Peled}
\address[R. Peled]{School of Mathematical Sciences, Tel Aviv University, Ramat Aviv, Tel Aviv 69978, Israel}
\email{peledron@tauex.tau.ac.il}

\maketitle


\begin{abstract}
We study how the typical gradient and typical height of a random surface are modified by the addition of quenched disorder in the form of a random independent external field. The results provide quantitative estimates, sharp up to multiplicative constants, in the following cases.

It is shown that for \emph{real-valued} random-field random surfaces of the $\nabla \phi$ type with a uniformly convex interaction potential: (i) The gradient of the surface delocalizes in dimensions $1\le d\le 2$ and localizes in dimensions $d\ge3$. (ii) The surface delocalizes in dimensions $1\le d\le 4$ and localizes in dimensions $d\ge 5$.

It is further shown that for the \emph{integer-valued} random-field Gaussian free field: (i) The gradient of the surface delocalizes in dimensions $d=1,2$ and localizes in dimensions $d\ge3$. (ii) The surface delocalizes in dimensions $d=1,2$. (iii) The surface localizes in dimensions $d\ge 3$ at low temperature and weak disorder strength. The behavior in dimensions $d\ge 3$ at high temperature or strong disorder is left open.

The proofs rely on several tools: Explicit identities satisfied by the expectation of the random surface, the Efron--Stein concentration inequality, a coupling argument for Langevin dynamics (originally due to Funaki and Spohn~\cite{FS}) and the Nash--Aronson estimate.
\end{abstract}

\section{Introduction}
\subsection{Main results} The seminal work of Imry and Ma~\cite{IM75} predicted that the addition of a quenched random external field eliminates the magnetization phase transition of low-dimensional spin systems. This was argued to be a generic phenomenon in two dimensions and to occur also in three and four dimensions in systems with continuous symmetry.
These predictions were confirmed for a broad class of spin systems in the celebrated work of Aizenman and Wehr~\cite{AW1989, AW89}. While the Imry-Ma phenomenon has mostly been studied in the spin system context, it has been recognized that related effects occur also for random surfaces of the $\nabla\phi$ type subjected to quenched disorder, in suitable ways~\cite{BK94, BK96, CK12, CK15, KO06, KO08, VK08} (see Section~\ref{sec.relresult}). In this work we study the way in which the fluctuations of random surfaces are enhanced by the addition of quenched randomness in the form of an independent external field, focusing on the localization and delocalization behavior of the gradient and heights of real- and integer-valued surfaces. Quantitative estimates are obtained in all cases studied, complementing other recent quantitative studies of the Imry--Ma phenomenon~\cite{AHP20, AP19,C18, DHP20++,DW2020, ding20192exponential}.

\medskip
{\bf Real-valued random-field random surfaces:} The first class of random surfaces that we consider are real-valued random surfaces of the $\nabla\phi$ type with a uniformly convex interaction potential, which are subjected to a quenched independent external field in the sense of~\eqref{eq:101071608} below.

We start with a few definitions. Let $\Z^d$ be the standard $d$-dimensional integer lattice, in which vertices are adjacent if they are equal in all but one coordinate and differ by one in that coordinate. For $\Lambda\subset\Z^d$, let $\partial\Lambda$ be the external vertex boundary of $\Lambda$, $E(\Lambda)$ be the set of edges of $\Z^d$ with both endpoints in $\Lambda$, and $\Lambda^+:=\Lambda\cup\partial\Lambda$.

Let $\Lambda\subset\Z^d$ be finite and let $\eta:\Lambda\to\R$. The Hamiltonian $H_\Lambda^\eta$ of the random surface on $\Lambda$ with external field $\eta$ associates to each $\phi:\Lambda^+\to\R$ the energy
\begin{equation} \label{eq:101071608}
    H_\Lambda^\eta \left( \phi \right) := \sum_{e\in E(\Lambda^+)} V\left( \nabla\phi(e) \right) - \lambda \sum_{x \in \Lambda} \eta(x) \phi(x),
\end{equation}
where $V:\R\to\R$ is a measurable function satisfying $V(x)=V(-x)$ for all $x$, $\lambda>0$ is the coupling strength of the external field $\eta$ and $V(\nabla\phi(e)):=V(\phi(x) - \phi(y))$ for an edge $e=\{x,y\}$ (noting that the orientation of $e$ is immaterial, as $V$ is an even function). We assume throughout that the potential $V$ is uniformly convex, i.e., that it is twice continuously differentiable and there exist $c_-,c_+$ satisfying
\begin{equation}\label{eq:V ellipticity}
  0 < c_- \leq V''(t) \leq c_+ < \infty.
\end{equation}
The probability distribution of the random surface, with zero boundary conditions, is then defined by
\begin{equation} \label{eq:defmuLeta}
    \mu_\Lambda^\eta (\di \phi) := \frac{1}{Z_\Lambda^\eta}\exp \left( - H_{\Lambda}^{\eta} (\phi) \right) \prod_{v \in \Lambda} d \phi(v)\prod_{v\in\partial\Lambda}\delta_{0} \left( d\phi(v) \right),
\end{equation}
where $dx$ indicates Lebesgue measure on $\R$, $\delta_0$ is the Dirac delta measure at $0$, and
\begin{equation}
Z_\Lambda^\eta := \int \exp \left( - H_{\Lambda}^{\eta} (\phi) \right) \prod_{v \in \Lambda} d \phi(v)\prod_{v\in\partial\Lambda}\delta_{0} \left( d\phi(v) \right).
\end{equation}
$Z_{\Lambda}^\eta$ is called the partition function, and normalizes $\mu_\Lambda^\eta$ to be a probability measure. We denote the expectation with respect to $\mu_\Lambda^\eta$ by $\left\langle \cdot \right\rangle_{\mu_\Lambda^\eta}$, and refer to it as the \emph{thermal expectation}.

A natural question pertaining to random surfaces is whether their fluctuations diverge on sequences of domains $\left( \mu_{\Lambda_n} \right)_{n\ge 1}$ which increase to $\Zd$. In the absence of an external field (i.e., when $\eta\equiv 0$) the following facts are known: In dimensions $d = 1,2$ the variance of the height at a fixed vertex diverges as $n$ tends to infinity~\cite{BLL75}; the random surface is \emph{delocalized} or rough. In dimensions $d \geq 3$, the Brascamp-Lieb concentration inequality~\cite{BL76, BL75} shows that the variance of the height remains bounded uniformly in $n$; the random surface is \emph{localized} or smooth. The Brascamp-Lieb inequality further implies that the fluctuations of the (discrete) gradient of the surface remain bounded in $n$ in every dimension $d \geq 1$.

In this paper, we study the effect that a quenched random independent field has on the localization and delocalization properties of the random surface and its gradient. Explicitly, we assume that the external field $\eta$ is random, with the random variables $(\eta(x))$ independent and satisfying various additional assumptions, and we study the fluctuations of $\mu_{\Lambda}^{\eta}$ for a typical realization of $\eta$. We shall denote the probability measure over the random field, its expectation, and its variance by $\P$, $\E$, $\var$, respectively.

A specific case of interest is the random-field Gaussian free field, i.e., the model~\eqref{eq:defmuLeta} with the quadratic potential $V(x) = \frac{1}{2}x^2$ (see also Section~\ref{sec:random-field Gaussian free field}). In this situation, the quenched random surface has a multivariate Gaussian distribution whose covariance structure can be explicitly calculated as a function of the realization of the random field $\eta$.
If $\eta$ is random, independent and each $\eta(x)$ has zero mean and unit variance, one can prove that, for almost every realization of the random field, the gradient of the random surface delocalizes if $d \leq 2$ and localizes if $d \geq 3$, and that the height of the surface delocalizes if $d \leq 4$ and localizes if $d \geq 5$. The result can be quantified and the typical height of the random surface and its gradient can be estimated in every dimension; we refer the reader to~\cite[Appendix A.1]{CK12}, where the qualitative delocalization of the random-field Gaussian free field is discussed in dimensions $d = 3,4$, and to~\cite[Section 1.2]{VK08} where the gradient fluctuations are quantified  (see also Section~\ref{secueta} for the calculations in the zero-temperature limit). Our analysis of the real-valued, random-field random surfaces extends these results to the class of potentials satisfying~\eqref{eq:V ellipticity}, for which the law of the random surface is not explicitly known.

Before stating the theorems, we introduce some notation. Write $\Lambda_L := \left\{ -L,  \ldots, L \right\}^d$ and let $\left| \Lambda_L \right|=(2L+1)^d$ be its cardinality. In the next two results we consider dimensions $d\ge 1$, integer $L\ge 2$, disorder strength $\lambda>0$, ellipticity parameters $0<c_-\le c_+<\infty$ and a twice-continuously differentiable $V:\R\to\R$ satisfying $V(x)=V(-x)$ for all $x$ and the uniform convexity assumption~\eqref{eq:V ellipticity}. We suppose $\eta:\Lambda_L\to\R$ are \emph{independent} random variables with moment assumptions as stated below, and that $\phi$ is sampled from the measure $\mu_{\Lambda_L}^\eta$ given by~\eqref{eq:defmuLeta}. 

Each of our theorems introduces its own positive constants $C,c$. Intuitively, $C$ stands for a generic large value while $c$ stands for a generic small value.

Our first theorem addresses the fluctuations of the gradient of the random surface.

\begin{theorem}[Gradient fluctuations, real-valued]  \label{prop3.100910}
Suppose $\E \left[ \eta(x) \right] =0$ and $\var \left[ \eta(x) \right] =1$ for all $x \in \Lambda_L$. There exist $C , c>0$ depending only on the dimension $d$ and the ratios $c_+/ c_-$ and $\lambda/ c_-$ such that the quantity
\begin{equation}\label{eq:L2 norm of gradient}
    \left\| \nabla \phi \right\|_{\underline{L}^2 \left(\Lambda_L^+ , \mu^\eta_{\Lambda_L} \right)}^2 := \frac{1}{\left| \Lambda_L^+ \right|} \sum_{e \in E \left( \Lambda_L^+ \right)} \left\langle (\nabla \phi(e))^2 \right\rangle_{\mu_{\Lambda_L}^\eta}
\end{equation}
satisfies
\begin{alignat}{3}
    &d=1:\qquad&&& c  L \leq~ &\E \left[ \left\| \nabla \phi \right\|_{\underline{L}^2 \left(\Lambda_L^+ , \mu^\eta_{\Lambda_L} \right)}^2 \right] \leq C  L + \frac{C}{c_-},\label{eq:TV0904}\\
    &d=2:&&& c \ln L \leq~ &\E \left[ \left\| \nabla \phi \right\|_{\underline{L}^2 \left(\Lambda_L^+ , \mu^\eta_{\Lambda_L} \right)}^2 \right] \leq C \ln L + \frac{C}{c_-},\label{eq:TV0904bis}\\
    &d\ge3:&&& c \leq~ &\E \left[ \left\| \nabla \phi \right\|_{\underline{L}^2 \left(\Lambda_L^+ , \mu^\eta_{\Lambda_L} \right)}^2 \right] \leq C \left( 1 + \frac{1}{c_-} \right).\label{eq:TV0904ter}
\end{alignat}
\end{theorem}
We remark that in this theorem, as well as in many of our subsequent results, the obtained bounds are with respect to the disorder-averaged (annealed) measure $\bar \mu (d \phi) := \int \mathbb{P}(d \eta) \mu^{\eta}_{\Lambda_L} (d \phi)$.

We further remark that our techniques for controlling the gradient fluctuations are applicable for general external fields $\eta$; see Theorem~\ref{eq:quenched gradient fluctuations}. In addition, our proof of Theorem~\ref{prop3.100910} applies under significant relaxations of the uniform convexity assumption~\eqref{eq:V ellipticity} (in particular, the proof applies to certain non-convex $V$); see Remark~\ref{rem:gradient fluctuations relaxed assumptions}.

Our second theorem concerns the fluctuations of individual heights in the random surface.

\begin{theorem}[Height fluctuations, real-valued]  \label{prop3.1009100bis}
Suppose $\var \left[ \eta(x) \right] =1$ for all $x \in \Lambda_L$. There exists $C>0$ depending only on the dimension $d$ and the ratios $c_+/ c_-$ and $\lambda/ c_-$ such that, for any $y \in \Lambda_L$,
\begin{alignat}{3}
    1\le d \leq 3&:\qquad&&&
    \var \left[  \left\langle \phi(y) \right\rangle_{\mu_{\Lambda_L}^{\eta}}  \right] &\leq C  L^{4-d},\label{eq:TV0904041}\\
    d=4&:&&&  \var \left[  \left\langle \phi(y) \right\rangle_{\mu_{\Lambda_L}^{\eta}}  \right] &\leq C \ln L, \label{eq:TV090404bis1}\\
    d\ge5&:&&& 
    \var \left[ \left\langle \phi(y) \right\rangle_{\mu_{\Lambda_L}^{\eta}}  \right] &\leq C. \label{eq:TV090404ter1}
\end{alignat}
Additionally, there exists $c>0$ depending only on the dimension $d$ and the ratios $c_+/ c_-$ and $\lambda/ c_-$ such that, for any $y \in \Lambda_{L/2}$,
\begin{alignat}{3}
    1\le d \leq 3&:\qquad&&&
     & \var \left[  \left\langle \phi(y) \right\rangle_{\mu_{\Lambda_L}^{\eta}}  \right] \geq c L^{4-d} ,\label{eq:TV0904041low}\\
    d=4&:&&&  & \var \left[  \left\langle \phi(y) \right\rangle_{\mu_{\Lambda_L}^{\eta}}  \right] \geq c \ln L, \label{eq:TV090404bis1low}\\
    d\ge5&:&&& &
    \var \left[ \left\langle \phi(y) \right\rangle_{\mu_{\Lambda_L}^{\eta}}  \right] \geq c. \label{eq:TV090404ter1low}
\end{alignat}
\end{theorem}

Let us make a few remarks about these results.

The measure $\mu_{\Lambda}^\eta$ does not have an explicit dependence on an inverse temperature parameter $\beta$; a more standard setup would have considered the probability measure $\mu_{\beta, \Lambda_L}^{\eta}$ defined by
\begin{equation}\label{eq:real-valued model with temperature}
    \mu_{\beta, \Lambda_L}^{\eta} := \frac{1}{Z_{\beta,\Lambda}^\eta}\exp \left( - \beta H_{\Lambda}^{\eta} (\phi) \right) \prod_{v \in \Lambda} d \phi(v)\prod_{v\in\partial\Lambda}\delta_{0} \left( d\phi(v) \right).
\end{equation}
However, the effect of $\beta$ can be mimicked in the model~\eqref{eq:defmuLeta} by multiplying $V$ and $\lambda$ by $\beta$ and thus the previous results are applicable also to the model~\eqref{eq:real-valued model with temperature}. Moreover, the lower bounds in Theorem~\ref{prop3.100910} and all bounds of Theorem~\ref{prop3.1009100bis} hold uniformly in $\beta$, since the constants $C,c$ in Theorem~\ref{prop3.100910} and Theorem~\ref{prop3.1009100bis} depend only on the ratio $\lambda / c_-$ and the ellipticity ratio $c_+/c_-$. The upper bound of Theorem~\ref{prop3.100910} also depends on $\frac{1}{c_-}$ and thus improves as $\beta$ increases. In particular, taking the limit $\beta \to \infty$ implies that the finite-volume ground configuration of the random-field $\nabla \phi$ model (see~\eqref{def.groundstate1706} below for its explicit formula when $\lambda=1$) satisfies the inequalities stated in Theorem~\ref{prop3.100910} and Theorem~\ref{prop3.1009100bis}.

The proof of Theorem~\ref{prop3.1009100bis} does not require $\eta$ to have mean zero (unlike the proof of Theorem~\ref{prop3.100910}). It is worth noting, however, that if $\eta$ is symmetric (i.e., $\eta$ has the same distribution as $-\eta$), then
\begin{equation*}
\E \left[  \left\langle \phi(y) \right\rangle_{\mu_{\Lambda_L}^{\eta}} \right] = 0 \hspace{3mm} \mbox{and consequently} \hspace{3mm} \var \left[  \left\langle \phi(y) \right\rangle_{\mu_{\Lambda_L}^{\eta}}  \right] = \E \left[ \left\langle \phi(y) \right\rangle_{\mu_{\Lambda_L}^{\eta}}^2 \right].
\end{equation*}
Thus, a symmetry assumption can be used to upgrade the conclusion of Theorem~\ref{prop3.1009100bis} from a variance bound to a bound on the $L^2$-norm (in the random field) of the thermal expectation $\left\langle \phi(0) \right\rangle_{\mu_{\Lambda_L}^\eta}$. It is plausible that such an $L^2$ bound also holds if the symmetry assumption is weakened to requiring that $\eta$ has mean zero, but this is not proven here; see also Section~\ref{sec:thermal expectation for grad phi model}.

Theorem~\ref{prop3.1009100bis} estimates the extent to which the thermal \emph{expectation} of the height fluctuates as the random field changes. It is also natural to consider the full fluctuations of the height, as a result of both thermal fluctuations and the randomness of the field. By the law of total variance, this can be decomposed as
\begin{equation}\label{eq:total variance formula}
   \E\left[\left\langle \phi(y)^2 \right\rangle_{\mu_{\Lambda_L}^{\eta}}\right] - \E\left[\left\langle \phi(y) \right\rangle_{\mu_{\Lambda_L}^{\eta}}\right]^2 = \E\bigg[\Big\langle \big(\phi(y)  - \left\langle \phi(y) \right\rangle_{\mu_{\Lambda_L}^{\eta}} \big)^2\Big\rangle_{\mu_{\Lambda_L}^{\eta}} \bigg] + \var\left[\left\langle \phi(y) \right\rangle_{\mu_{\Lambda_L}^{\eta}}\right]
\end{equation}
with the second term on the right-hand side estimated by Theorem~\ref{prop3.1009100bis}. The first term is estimated by the Brascamp-Lieb inequality~\cite{BL76, BL75}, which, in our setting, reads
\begin{equation} \label{eq:11201008}
    \left\langle \left(\phi(y) - \left\langle \phi (y) \right\rangle_{\mu_{\Lambda_L}^{\eta}} \right)^2 \right\rangle_{\mu_{\Lambda_L}^{\eta}} \leq \left\{ \begin{aligned}
    &C L &d=1, \\
    &C \ln L &d=2,\\
    &C  &d \geq 3.
    \end{aligned} \right.
\end{equation}
We thus see that the first term on the right-hand side of~\eqref{eq:total variance formula} is not larger than the second term in every dimension (up to a multiplicative factor). We further remark that the Brascamp-Lieb inequality can also be used to obtain Gaussian concentration estimates for the thermal fluctuations (see~\cite[Section 2.2.1]{DGI00} and~\cite[Theorem 4.9 and Remark 4.1]{F05}).

While the variance upper bounds of Theorem~\ref{prop3.1009100bis} hold pointwise for every vertex of $\Lambda_L$, matching lower bounds cannot be expected to hold for vertices arbitrarily close to the boundary and are thus stated for vertices which are at a certain distance of the boundary.

Theorem~\ref{prop3.100910} and Theorem~\ref{prop3.1009100bis} are related to the detailed investigations of Cotar and K\"ulske on the existence and uniqueness of translation-covariant gradient Gibbs measures for disordered random surfaces~\cite{CK12, CK15}. This is discussed further in Section~\ref{sec.relresult} but we already point out here that the upper bounds in Theorem~\ref{prop3.100910} in dimensions $d\ge 3$ may be deduced from~\cite[Proposition 2.1 and Lemma 3.5]{CK12}.

\medskip
{\bf Integer-valued random-field Gaussian free field:} In this section we study the fluctuations of the \emph{integer-valued} Gaussian free field when subjected to a quenched independent external field, as we now define.
Let $\Lambda\subset\Z^d$ be finite and let $\eta:\Lambda\to\R$. The Hamiltonian $H^{\IV , \eta}_\Lambda$ of the integer-valued Gaussian free field on $\Lambda$ with external field $\eta$ associates to each $\phi:\Lambda^+\to\Z$ the energy
\begin{equation}\label{eq:integer-valued Hamiltonian}
   H^{\IV , \eta}_\Lambda(\phi) := \frac 12 \sum_{e \in E \left( \Lambda^+ \right)} \nabla \phi(e)^2  -  \lambda \sum_{x \in \Lambda} \eta(x) \phi(x)
\end{equation}
where $\lambda>0$ is the coupling strength of the external field. The probability distribution for this surface, at inverse temperature $\beta > 0$ and with zero boundary conditions, assigns probability
\begin{equation} \label{eq:defGFFIV}
\mu^{\IV , \beta, \eta}_{\Lambda}( \phi) := \frac{1}{Z^{\IV, \beta, \eta}_{\Lambda} } \exp \left( - \beta H^{\IV , \eta}_\Lambda(\phi)  \right)
\end{equation}
to each $\phi : \Lambda^+ \to \Z$ satisfying $\phi \equiv 0$ on $\partial \Lambda$, where
\begin{equation}
Z_\Lambda^{\IV, \beta, \eta} := \sum_{\substack{\phi : \Lambda^+ \to \Z\\ \phi\equiv 0\text{ on }\partial\Lambda}} \exp \left( - \beta H^{\IV , \eta}_\Lambda(\phi) \right),
\end{equation}
the partition function, normalizes $\mu^{\IV , \beta, \eta}_{\Lambda}$ to be a probability measure. We denote the (thermal) expectation with respect to the measure $\mu^{\IV , \beta, \eta}_{\Lambda}$ by $\left\langle \cdot \right\rangle_{\mu^{\IV, \beta, \eta}_{\Lambda}}$ .

Our results show that the \emph{gradient} of the integer-valued random-field Gaussian free field shares the delocalization/localization properties of the real-valued surfaces discussed above in all dimensions, and that the two models share similar \emph{height} fluctuations in dimensions $d=1,2$. It is further shown that the integer-valued model localizes in all dimensions $d\ge 3$ at low temperature and weak disorder, thus exhibiting a different behavior from the real-valued surfaces in dimensions $d=3,4$. 

The behavior of the integer-valued model in dimensions $d\ge 3$ at high temperature or strong disorder is left open; See Section~\ref{sec:discussion and open questions}.

In the next theorems, we consider dimensions $d\ge 1$, integer $L\ge 2$, inverse temperature $\beta>0$ and disorder strength $\lambda>0$. We suppose $\eta:\Lambda_L\to\R$ are \emph{independent} random variables. $\phi$ is sampled from the measure $\mu_{\Lambda_L}^{\IV , \beta, \eta}$ given by~\eqref{eq:defGFFIV}.

\begin{theorem}[Gradient fluctuations, integer-valued]
\label{prop3.100910disc}
Suppose $\E \left[ \eta(x) \right] =0$ and $\var \left[ \eta(x) \right] =1$ for all $x \in \Lambda_L$. There exist $C , c>0$ depending only on the dimension such that the quantities: $c_\lambda := c \lambda^2$, $C_\lambda : = C \lambda^2$, $C_\beta := C \left( 1 + \beta^{-1} \right)$ and
\begin{equation*}
    \left\| \nabla \phi \right\|_{\underline{L}^2 \left(\Lambda_L^+ , \mu^{\IV, \beta, \eta}_{\Lambda_L} \right)}^2 := \frac{1}{\left| \Lambda_L^+ \right|} \sum_{e \in E \left( \Lambda_L^+ \right)} \left\langle (\nabla \phi(e))^2 \right\rangle_{\mu_{\Lambda_L}^{\IV, \beta, \eta}}
\end{equation*}
satisfy
\begin{alignat}{3}
    &d=1:\qquad&&& c_\lambda L - C \leq~ &\E \left[ \left\| \nabla \phi \right\|_{\underline{L}^2 \left(\Lambda_L^+ , \mu^{\IV , \beta, \eta}_{\Lambda_L} \right)}^2 \right] \leq C_\lambda L + C_\beta ,\label{eq:TV09042}\\
    &d=2:&&& c_\lambda\ln L - C  \leq~ &\E \left[ \left\| \nabla \phi \right\|_{\underline{L}^2 \left(\Lambda_L^+ , \mu^{\IV , \beta, \eta}_{\Lambda_L} \right)}^2 \right] \leq C_\lambda \ln L +  C_\beta,\label{eq:TV0904bis2}\\
    &d\ge3:&&& c_\lambda - C \leq~ & \E \left[ \left\| \nabla \phi \right\|_{\underline{L}^2 \left(\Lambda_L^+ , \mu^{\IV , \beta, \eta}_{\Lambda_L} \right)}^2 \right] \leq C_\lambda + C_\beta.\label{eq:TV0904ter2}
\end{alignat}
\end{theorem}

\begin{remark} \label{remark1.1}
The lower bound in~\eqref{eq:TV0904ter2} is trivial when $\lambda$ is small. This is improved in Remark~\ref{remark6.4}, where we establish that, for any dimension $d \geq 1$ and any disorder strength $\lambda > 0$, there exists a constant $c_1 >0$ depending on the dimension, the law of the random field and the disorder strength, such that
\begin{equation*}
    \liminf_{L \to \infty} \E \left[ \left\| \nabla \phi \right\|_{\underline{L}^2 \left(\Lambda_L^+ , \mu^{\IV , \beta, \eta}_{\Lambda_L} \right)}^2 \right] \geq c_1.
\end{equation*}
\end{remark}

\begin{theorem}[Height fluctuations, integer-valued, $d=1,2$]  \label{prop3.1009100bisdisc}
Suppose that the $(\eta(x))$ are identically distributed with $\E \left[ \eta(x) \right] =0$ and $\var \left[ \eta(x) \right] =1$ for all $x \in \Lambda_L$. There exist a constant $c>0$ depending on the common distribution of the $(\eta(x))$ and an absolute constant $C>0$ such that the quantities: $c_\lambda := c e^{- \frac{1}{c \lambda^2}}$, $C_{\lambda, \beta} := C \left(1 + \lambda^2 +\beta^{-1} \right)$ and
\begin{equation*}
\left\| \phi \right\|_{\underline{L}^2 \left(\Lambda_L , \mu^{\IV , \beta, \eta}_{\Lambda_L} \right)}^2 := \frac{1}{\left| \Lambda_L \right|} \sum_{x \in \Lambda_L} \left\langle  \phi(x)^2 \right\rangle_{\mu_{\Lambda_L}^{\IV, \beta, \eta}}
\end{equation*}
satisfy
\begin{alignat}{3}
    &d=1:\qquad&&& c_\lambda L^{3} \leq~ & \E \left[ \left\| \phi \right\|_{\underline{L}^2 \left(\Lambda_L , \mu^{\IV , \beta, \eta}_{\Lambda_L} \right)}^2 \right] \leq C_{\lambda, \beta}  L^3 ,\label{eq:TV0904222}\\
    &d=2:&&& c_\lambda L^{2} \leq ~ & \E \left[ \left\| \phi \right\|_{\underline{L}^2 \left(\Lambda_L , \mu^{\IV , \beta, \eta}_{\Lambda_L} \right)}^2 \right] \leq C_{\lambda, \beta} L^{2} \label{eq:TV0904bis2222}.
\end{alignat}
\end{theorem}

\begin{remark}
    The results of Theorem~\ref{prop3.100910disc} and Theorem~\ref{prop3.1009100bisdisc} apply for all temperatures $\beta\in [0 , \infty).$
\end{remark}

The theorems determine the order of magnitude of the norms of the gradient and the height of the surface as a function of $L$, and also estimate the dependence on the disorder strength $\lambda$. Comparing to Theorem~\ref{prop3.100910}, one notices that the dependence of the gradient norm on the disorder strength is the same as in the real-valued case. In contrast, compared with Theorem~\ref{prop3.1009100bis}, the lower bound at weak disorder for the height norm is significantly smaller (in its dependence on $\lambda$) in the integer-valued case than in the real-valued case. We expect that the two models indeed behave differently. For instance, the proof of Theorem~\ref{thm10162010} below shows that at zero temperature, in order for $\phi$ to be non-zero at the origin it is necessary that there exists a connected subset with connected complement of $\Lambda_L$ containing the origin in which the sum of $\lambda$ times the disorder $\eta$ exceeds a constant multiple of its boundary size. A recent result of Ding and Wirth~\cite[Proposition 2.2]{DW2020} shows that in two dimensions it is unlikely that there exist such subsets if $L\le\exp(\lambda^{-4/3+o(1)})$ (with the $o(1)$ term referring to the limit $\lambda\downarrow 0$).

We also remark that the assumption that the $(\eta(x))$ are identically distributed in Theorem~\ref{prop3.1009100bisdisc} is only required for the proof of the lower bounds on the height norm, and may be replaced by an assumption that a $(2+\delta)$-moment of the $(\eta(x))$ admits a uniform upper bound.

\begin{theorem}[Height fluctuations, integer-valued, $d\ge 3$, low temperature and weak disorder]\mbox{}\!\! \label{thm10162010}
  Suppose $d \geq 3$ and assume that $\eta(x)$ has the standard Gaussian distribution $N(0,1)$ for all~$x \in \Lambda_L$. There exist $\beta_0, \lambda_0, c>0$ such that for all $\beta \in (\beta_0,\infty)$, $\lambda \in (0, \lambda_0)$, integer $L\ge 2$, $v \in \Lambda_L$, and integer $t  >0$, 
  \begin{equation} \label{eq:10142010}
      \P \left( \mu^{\IV , \beta, \eta}_{\Lambda_L} \left( |\phi(v)| < t  \right) \geq 1-  e^{- c \beta t^{1/2}} \right) \geq 1 - e^{-\frac{c t^{1/6}}{\lambda^2}}.
  \end{equation}
In particular, for all $k >0$,
\begin{equation} \label{eq:10152010}
\sup_{L \in \N} \,  \mathbb{E} \left[\left\langle |\phi(v)|^k \right\rangle_{\mu^{\IV, \beta, \eta}_{\Lambda_L}}\right] < \infty.
\end{equation}
\end{theorem}
The exponents $1/2$ and $1/6$ in~\eqref{eq:10142010} are not sharp (the proof can be optimized to yield slight improvements). However, the $\lambda^2$ term appears to be necessary, at least for the case $t=1$, as it controls the local fluctuations of the field at $v$ when the disorder strength is weak.

As mentioned above, Theorem~\ref{thm10162010} proves that the behavior of the integer-valued random-field Gaussian free field differs from its real-valued counterpart in dimensions $d=3,4$.

We additionally note that it is possible to take the zero-temperature limit $\beta \to \infty$ in the results of Theorem~\ref{prop3.100910disc}, Theorem~\ref{prop3.1009100bisdisc} and Theorem~\ref{thm10162010} in order to obtain that the ground state of the integer-valued random-field Gaussian free field satisfies the estimates stated in theses results.

The discussion up until now has been centered around gradient models where the Hamiltonian depends on the discrete gradient of the surface. One can also consider the effects of the addition of a random field on ``higher-order" random surfaces. One such example is the random-field membrane model, whose Hamiltonian is given by the formula 
\begin{equation} \label{eq:12442102}
    H^{\eta, \Delta}_{\Lambda_L}(\phi) := \frac 12 \left\| \Delta \phi\right\|_{L^2 \left( \Lambda_L \right)}^2 - \sum_{x \in \Lambda_L} \eta(x) \phi(x),
\end{equation}
where $\Delta$ indicates the graph Laplacian on $\Lambda_L$. Explicit computations, which are available for this model, enable us to prove upper and lower bounds on the fluctuation of $\phi$; see Section~\ref{RFMM} for a more detailed discussion of the model and the results.

\bigskip

\subsection{Background} \label{sec.relresult} Brascamp, Lieb and Lebowitz~\cite{BLL75} initiated the first detailed investigation of the fluctuations of real-valued random surfaces of the form~\eqref{eq:defmuLeta} without an external field (sometimes called the $\nabla \phi$-model) and compared their behavior to the exactly-solvable Gaussian free field (the case $V(x) = \frac{1}{2}x^2$). Among other results, their work proved that such surfaces delocalize in two dimensions (the one-dimensional case is classical) and localize in three and higher dimensions under the assumption~\eqref{eq:V ellipticity}. The \emph{integer-valued} Gaussian free field (without an external field) exhibits similar behavior in dimensions $d=1$ and $d\ge 3$ but undergoes a \emph{roughening transition} in two dimensions as the temperature increases: Localization at low temperatures follows by a version of the classical Peierls argument. Delocalization at high temperatures was proved in the breakthrough work of Fr\"ohlich and Spencer~\cite{FrSp, frohlich1981kosterlitz} on the Berezinskii-Kosterlitz-Thouless transition (see also~\cite{kharash2017fr, L20}). Further details on random surfaces without external field can be found, e.g., in the works of Funaki~\cite{F05}, Velenik~\cite{V06} and Sheffield~\cite{Sh}. The theory of disordered random surfaces is less developed; we summarize some of the existing mathematical literature below.

\smallskip

{\bf Real-valued disordered random surfaces:} K\"{u}lske and Orlandi~\cite{KO06} studied the model~\eqref{eq:defmuLeta} in two dimensions under the assumption that $V$ is even, twice continuously differentiable and $\sup_t V''(t)<\infty$ (and, say, $V(t)/t^{1+\ep}\to\infty$ as $t\to\infty$ for some $\ep>0$). Using a Mermin-Wagner type argument they proved that the height fluctuations in $\Lambda_L$, for \emph{any} fixed external field $\eta$, are at least of order $\sqrt{\log L}$ (which is sharp when $V(t) = t^2$ and $\eta\equiv 0$) and established Gaussian lower bounds on the tail decay. In this sense, the external field can only enhance the height fluctuations. This result is extended in K\"{u}lske--Orlandi~\cite{KO08} where it is proved that even an arbitrarily strong $\delta$-pinning at height zero cannot localize a two-dimensional random surface with i.i.d.\ random external field $(\eta(x))$ whose common distribution is symmetric with finite non-zero variance. This is then complemented by an analysis of the effects of pinning in dimensions $d\ge 3$ proving, in particular, that sufficiently strong pinning can localize the surface (under the assumption that $\inf_t V''(t)>0$).

Van Enter and K\"{u}lske~\cite{VK08} considered the model~\eqref{eq:defmuLeta} when the potential $V$ is assumed to be even, continuously differentiable and to have super-linear growth ($V(t)/t^{1+\ep}\to\infty$ as $t\to\infty$ for some $\ep>0$), and when the $(\eta(x))$ are i.i.d., having a common symmetric distribution of finite non-zero variance. They proved that in two dimensions there are no translation-covariant \emph{gradient} Gibbs measures $\mu^\eta$ satisfying that $\E\left[|\left\langle V'(\nabla \phi(e)) \right\rangle_{\mu^\eta}|\right]<\infty$ for an edge $e$. They additionally proved that in three dimensions any gradient Gibbs measure $\mu^\eta$ has slow decay of correlations in the following sense: For any $\ep>0$, the correlations $C_{e,e'}:=\E\left[\left\langle V'(\nabla \phi(e)) \right\rangle_{\mu^\eta}\left\langle V'(\nabla \phi(e')) \right\rangle_{\mu^\eta}\right]$ satisfy $\lim_{r\to\infty} \sup_{e,e'} r^{1+\ep}|C_{e,e'}|=\infty$ where the supremum is over all pairs of edges at distance at least $r$ from each other.

Cotar and K\"{u}lske~\cite{CK12, CK15} considered two models of disordered random surfaces: (A) The model~\eqref{eq:defmuLeta} with even, twice continuously differentiable $V$ and i.i.d.\ $(\eta(x))$ whose common distribution has finite non-zero second moment. (B) A model in which a collection of i.i.d.\ random functions $\left( V_e \right)_{e \in E(\Zd)}$ is prescribed, and the formal Hamiltonian is given by
\begin{equation}
    H^{(V_e)}(\phi):=\sum_e V_e(\nabla \phi(e)).
\end{equation}
Their results show that model (B), subject to suitable assumptions on the random potentials $(V_e)$, behaves similarly to the non-disordered case. Thus, our discussion here pertains to their results on model (A), which are closer to our work.

The work~\cite{CK12} proves the following results on model (A) when $V$ grows quadratically at infinity and satisfies $\sup_x V''(x)<\infty$: (1) When $\E\left[\eta(x)\right]=0$ and $d\ge 3$, there \emph{exist} translation-covariant gradient Gibbs measures $\mu^\eta$ with any prescribed disorder-averaged tilt, which also satisfy $\E\left[\left\langle\nabla\phi(e)^2\right\rangle_{\mu^\eta}\right]<\infty$; (2) when $\E\left[ \eta(x) \right]=0$ and $d\ge 3$, the infinite-volume surface tension exists and is independent of $\eta$; (3) when $\E\left[\eta(x)\right]\neq 0$, there are no translation-covariant gradient Gibbs measures $\mu^\eta$ with $\E\left[|\left\langle V'(\nabla \phi(e)) \right\rangle_{\mu^\eta}|\right]<\infty$ (using the techniques of~\cite{VK08}); (4) the infinite-volume surface tension does not exist in dimensions $d=1,2$ and, when $\E\left[\eta(x)\right]\neq 0$, in dimensions $d\ge 3$.

The work~\cite{CK15} proves the following results on model (A) under the assumption~\eqref{eq:V ellipticity} and when $\E\left[\eta(x)\right]=0$ and, for some of the results, $\eta(x)$ has a symmetric distribution which satisfies a Poincar\'e inequality: (1) In dimensions $d\ge 3$, there is a \emph{unique} translation-covariant gradient Gibbs measure $\mu^\eta$ with prescribed disorder-averaged tilt, satisfying $\E\left[\left\langle\nabla\phi(e)^2\right\rangle_{\mu^\eta}\right]<\infty$, and whose disorder-average is ergodic; (2) for each translation-covariant gradient Gibbs measure $\mu^\eta$, $\text{Cov}\left(\left\langle\nabla\phi(e)\right\rangle_{\mu^\eta}, \left\langle\nabla\phi(e')\right\rangle_{\mu^\eta}\right)\le \frac{C}{\text{dist(e,e')}^{d-2}}$ (and more general bounds).

As mentioned above, our Theorem~\ref{prop3.100910} and Theorem~\ref{prop3.1009100bis} are related to the works of Cotar and K\"{u}lske~\cite{CK12, CK15}. It seems that the upper bounds on the gradient fluctuations in Theorem~\ref{prop3.100910} in dimensions $d\ge 3$ can be deduced from~\cite[Proposition 2.1 and Lemma 3.5]{CK12} and the proofs there may possibly extend to dimensions $d=1,2$. However, the proofs use different arguments and make different assumptions on the potential (neither set of assumptions implies the other): The results of~\cite{CK12} are established for potentials $V$ having quadratic growth at infinity ($V(t)\ge At^2 -B$ for $A>0$ and $B\in\R$) and uniformly upper bounded second derivative ($\sup_t V''(t)<\infty$), while our results apply in the class described in Remark~\ref{rem:gradient fluctuations relaxed assumptions}. We also note that while Cotar and K\"{u}lske discuss only gradient Gibbs measures, some of the tools which they use in~\cite{CK15} are related to our approach to the height fluctuations in Theorem~\ref{prop3.1009100bis}. Specifically, both proofs use a coupling of Langevin dynamics, as originally proposed by Funaki--Spohn~\cite{FS} (though~\cite[Section 4]{CK15} couples the dynamics with the same disorder while we couple them with different disorders). In addition, our argument relies on the Efron--Stein inequality while~\cite[Proposition 2.4]{CK15} uses a related covariance bound.

\smallskip

{\bf Integer-valued disordered random surfaces:} To our knowledge, the integer-valued random-field Gaussian free field has not been studied before. The following related model was treated mathematically by Bovier--K\"{u}lske~\cite{BK94, BK96} (with earlier treatments on a hierarchical lattice by Bovier, K\"{u}lske and Picco~\cite{bovier1991stability, bovier1992stability, bovier1993hierarchical}). The Hamiltonian of the model takes the form
\begin{equation*}
    H(\phi) := \sum_{x\sim y} |\phi(x) - \phi(y)| - \lambda \sum_{k \in \Z} \sum_{x} \eta_x (k) \indc_{\{ \phi(x) = k \}},
\end{equation*}
where the surface $\phi$ is integer-valued, the random variables $(\eta_x(k))_{x\in\Z^d, k\in\Z}$ form a $(d+1)$-dimensional environment and $\lambda>0$ is the disorder strength. The model provides an approximation to the domain walls in disordered ferromagnetic Ising models when the random variables $(\eta_x(k))$ are i.i.d.\ (random-bond case) or when the differences $(\eta_x(k) - \eta_x(k-1))$ are i.i.d.\ (random-field case) (see also~\cite[Section 5.1]{For91} for a physics discussion of related models in the continuum). The work~\cite{BK94} allows these distributions of $\eta$ (and more general settings), requiring that the i.i.d.\ ensembles have zero mean, unit variance and satisfy suitable Gaussian tail bounds, and establishes the existence of infinite-volume Gibbs measures in dimensions $d \geq 3$ at low temperatures and weak disorder (small $\lambda$) (building on the renormalization group approach of Bricmont--Kupiainen~\cite{BK88}). The work~\cite{BK96} shows that, when the $(\eta_x(k))$ are i.i.d.\ with zero mean and unit variance with common distribution having no isolated atoms or having compact support, the model does not admit translation-covariant Gibbs states in dimensions $d \leq 2$, at all positive temperatures and non-zero disorder strength $\lambda$ (adapting the arguments of Aizenman--Wehr~\cite{AW89}).

\subsection{Convention for constants} Throughout this article, the symbols $C$ and $c$ denote positive constants which may vary from line to line, with $C$ increasing and $c$ decreasing. Except where explicitly stated otherwise, these constants may depend only on the dimension $d$ and the ratios $c_+/c_-$ and $\lambda/c_-$.

\subsection{Strategy of the arguments} \label{sec.strat}

In this section, we present some of the main arguments developed in this article. 

\subsubsection{Gradient fluctuations} \label{section1.3.1}
The proof of Theorem~\ref{prop3.100910} relies on a quenched comparison principle. For any fixed realization of the disorder $\eta$, we are able to relate the $L^2$ norm~\eqref{eq:L2 norm of gradient} of the gradient of the random surface with a general uniformly elliptic potential $V$ to the analogous $L^2$ norm of the ground state of the Gaussian random surface (i.e. the case $V(x) = \frac{1}{2}x^2$). The result is precisely stated in~\eqref{eq:ComparisonOfL2Norms}. Theorem~\ref{prop3.100910} follows from this comparison principle and the fact that the $\eta$-average of the $L^2$ norm can easily be estimated via explicit formulas in the Gaussian case, as presented in Proposition~\ref{prop.propueta}.

The proof of the comparison principle~\eqref{eq:ComparisonOfL2Norms} relies on computations based on the two identities
\begin{align}
  &-\langle \sum_{e \ni x} V'(\nabla\phi(e))\rangle_{\mu_{\Lambda_L}^\eta} = \, \eta(x), \label{identity1intro}\\
  &-\langle \phi(x)\sum_{e \ni x} V'(\nabla\phi(e))\rangle_{\mu_{\Lambda_L}^\eta} = 1+\, \eta(x)\langle \phi(x) \rangle_{\mu_{\Lambda_L}^\eta}, \label{identity2intro}
\end{align}
where we used the notation $\sum_{e \ni x}$ to sum over the directed edges containing $x$ as an endpoint (see~\eqref{eq:discrete gradient} and~\eqref{eq:11121509}). Both identities follow from a simple integration by parts; see Section~\ref{sec:IBP}. We remark that a version of the identity~\eqref{identity1intro} was previously used by van Enter and K\"{u}lske~\cite[Proposition 2.2]{VK08} (the ``divergence equation'') in their work on gradient Gibbs states in dimensions $d=2,3$.

\subsubsection{Height fluctuations} \label{section1.3.2}

To highlight the main ideas of the argument of the proof of Theorem~\ref{prop3.1009100bis}, we describe the strategy for the upper bounds in the case of the ground state, instead of the thermal expectation of the field under the Gibbs measure $\mu_{\Lambda_L}^\eta$. Additionally, for notational simplicity, we set the field strength $\lambda$ to $1$ and consider the variance of the height at the origin. To be more precise, the ground state is defined as the minimizer of the energy
\begin{equation*}
    H_{\Lambda_L}^{\eta}(v) := \sum_{e \in E(\Lambda_L)} V \left( \nabla v(e) \right) - \sum_{x \in \Lambda_L} \eta(x) v(x)
\end{equation*}
among all the mappings $v : \Lambda_L^+ \to \R$ whose values are set to $0$ on the boundary $\partial \Lambda_L$. We denote the minimizer of $H_{\Lambda_L}^{\eta}$ by $v_{L,\eta} : \Lambda_L^+ \to \R$ ; equivalently, it can be characterized as the unique solution of the discrete non-linear elliptic equation
\begin{equation} \label{def.groundstate1706}
    \left\{ \begin{aligned}
    - \sum_{e \ni y} V' \left(  \nabla v_{L , \eta}(e) \right) &= \eta(y) &~\mbox{for}~y \in \Lambda_L, \\
    v_{L,\eta}(y) &= 0 &~\mbox{for}~y \in \partial \Lambda_L.
    \end{aligned}
    \right.
\end{equation}
We wish to estimate the variance (over the random field $\eta$) of the random variable $v_{L , \eta}(0)$, and prove that it satisfies the bounds required in Theorem~\ref{prop3.1009100bis}

The proof of the upper bounds relies on the Efron--Stein concentration inequality (stated in Proposition~\ref{prop.Efronstein}): If we consider two independent copies of the random field, which we denote by $\eta$ and $\tilde \eta$, and let $\eta^{x}$ be the field satisfying $\eta^{x}(y) = \eta(y)$ if $y \neq x$ and $\eta^{x}(x) = \tilde \eta(x)$, then we have the variance estimate
\begin{equation} \label{eq:12531309}
    \var \left[ v_{L,\eta}(0) \right] \leq \frac 12 \sum_{x \in \Lambda_L} \E \left[ \left( v_{L , \eta}(0) - v_{L , \eta^x}(0) \right)^2 \right].
\end{equation}
Consequently, it is sufficient, in order to obtain the desired upper bounds, to prove the inequalities, for any point $x \in \Lambda_L$,
\begin{equation} \label{eq:14021308}
    \E \left[ \left( v_{L , \eta}(0) - v_{L , \eta^x}(0) \right)^2 \right] \leq
    \left\{ \begin{aligned}
    &CL^2 &d =1,\\
    &C \left(\ln  \frac{L}{1 \vee |x|} \right)^2 &d = 2, \\
    &\frac{C}{ 1 \vee |x|^{2d-4}}&d \geq 3
    \end{aligned}
    \right.
\end{equation}
(note that, although the two-dimensional bound is logarithmic, substituting it into \eqref{eq:12531309} will bound the sum by a constant multiple of $L^2$, as stated in Theorem~\ref{prop3.1009100bis}.)

The proof of the upper bounds~\eqref{eq:14021308} is based on the observation that the difference $w^x :=  v_{L , \eta} - v_{L , \eta^x}$ solves a discrete linear elliptic equation of the form
\begin{equation} \label{eq:092627077bis}
    \left\{ \begin{aligned}
     - \nabla \cdot \a \nabla w^x & = \left(\eta(x) - \tilde \eta(x) \right) \delta_x &&~\mbox{in}~\Lambda_L, \\
    w^x &= 0 &&~\mbox{on}~\partial \Lambda_L,
    \end{aligned} \right.
\end{equation}
with $\delta_y$ the Kronecker delta function and with the elliptic operator $-\nabla \cdot \a \nabla$ defined in~\eqref{eq:11330408}. Here, the
environment $\a$ is an explicit function of the ground state $v_{\eta, L}$ and the potential $V$ (in particular, $\a$ is random) which satisfies the pointwise uniform ellipticity estimates $c_- \leq \a \leq c_+$ almost surely. Using the linearity of the equation~\eqref{eq:092627077bis}, the mapping $w^{x}$ can be rewritten as
\begin{equation} \label{eq:13531708}
    w^{x}( 0) = \left(\eta(x) - \tilde \eta(x) \right) G_\a \left( 0, x \right)
\end{equation}
where $G_\a : \Lambda_L^+ \to (0 , \infty)$ is the Green's function associated with the environment $\a$ and satisfying Dirichlet boundary condition on the boundary of the box $\Lambda_L$. The famed Nash-Aronson bounds (Proposition~\ref{prop.NashAronson}) provides upper and lower bounds on the Green's function $G_\a$ whenever $\a$ is uniformly elliptic. They establish that the map $G_\a$ is comparable to the standard random walk Green's function on $\mathbb{Z}^d$, whose square has the same order of magnitude as the right-hand side of \eqref{eq:14021308}. Using that the random variables $\eta(x)$ and $\tilde \eta^x(x)$ have the same expectation and are of unit variance allows to deduce the upper bound~\eqref{eq:14021308}.

The proof of the lower bounds relies on a similar, but more involved, strategy that is outlined below. Building upon the techniques used in the proof of the Efron-Stein inequality (see for instance~\cite[Theorem 3.1]{BGM13}), we consider an enumeration $x_1, \ldots, x_{(2L+1)^d}$ of the vertices of the box $\Lambda_L$, and let $\mathcal{F}_n$ be the $\sigma$-algebra generated by the random variables $\eta(x_1), \ldots , \eta(x_n)$. We then introduce the martingale $X_n = \E \left[ v_{L,\eta} \, | \,  \mathcal{F}_n\right]$ and observe that
\begin{equation*}
    \var \left[ v_{L,\eta}(0) \right] = \sum_{x \in \Lambda_L} \E \left[ (X_n - X_{n-1})^2 \right].
\end{equation*}
We are then able to lower bound each of the terms in the right-hand sides using a strategy similar to the one used for the upper bound, relying on the Nash--Aronson estimates (see~\eqref{Nash.infprop3.3} of Proposition~\ref{prop.NashAronson}) to provide lower bounds on the Green's function $G_\a$.

The extension of the result from the ground state to thermal expectation over the Gibbs measure $\mu^\eta_{\Lambda_L}$, as stated in Theorem~\ref{prop3.1009100bis}, is done by appealing to the Langevin dynamics associated with the models (see Section~\ref{sectionLangevin}), and extending the argument presented above from the setting of elliptic equations to the one of parabolic equations. The details are developed in Section~\ref{section4}.

\subsubsection{Integer-valued random-field Gaussian free field} \label{section1.3.3}
For the proof of Theorem~\ref{prop3.100910disc}, we begin by observing that the Hamiltonian~\eqref{eq:integer-valued Hamiltonian} of the integer-valued random-field Gaussian free field may be expressed, by completing a square, as
\begin{equation}
H^{\IV , \eta}_\Lambda(\phi) := \frac 12 \sum_{e \in E \left( \Lambda^+ \right)} \left(\nabla \phi(e) -  \lambda \nabla u_{\Lambda,\eta}(e)\right)^2 + c(\eta)
\end{equation}
where $c(\eta)$ depends only on $\eta$ and $u_{\Lambda,\eta}$ is the ground state of the \emph{real-valued} random-field Gaussian free field (i.e., $-\Delta u_{\Lambda,\eta} = \eta$; see Section~\ref{secueta}). This motivates the following quenched upper bound (Lemma~\ref{Thm7DGFF})
\begin{equation}\label{eq:closeness of integer and real gradients}
    \left\langle \exp \left(  \frac{\beta}{8} \left\| \nabla \phi - \lambda \nabla u_{\Lambda, \eta} \right\|_{L^2  \left( \Lambda^+ \right)}^2 \right)\right\rangle_{\mu_\Lambda^{\IV, \beta, \eta}} \leq \exp \left( C (1+\beta) \left| \Lambda^+ \right| \right)
\end{equation}
(see Section~\ref{Fctandevts} for the $L^2$ notations).
The bound implies, using Jensen's inequality, that  
\begin{equation} \label{eq:11253005}
    \left\| \nabla \phi - \lambda \nabla u_{\Lambda, \eta} \right\|_{\underline{L}^2  \left( \Lambda^+ , \mu_\Lambda^{\IV, \beta, \eta}\right)}^2  \leq C  \left( 1 + \beta^{-1} \right)
  \end{equation}
(this bound is especially simple to see at zero temperature, since the left-hand side is essentially the Hamiltonian $H^{\IV , \eta}_\Lambda$ and one option for $\phi$ is the integer part of the function $\lambda u_{\Lambda,\eta}$). The upper bounds on the gradient fluctuations of $\phi$ (the upper bounds in Theorem~\ref{prop3.100910disc}) are then a direct consequence of~\eqref{eq:11253005} and the upper bounds on $\nabla u_{\Lambda_L, \eta}$ stated in Proposition~\ref{prop.propueta}.
which shows that $\phi$ cannot fluctuate significantly more than $u_{\Lambda, \eta}$ in dimensions $d=1,2$.

It is possible to also rely on~\eqref{eq:11253005} to obtain lower bounds on the gradient and height fluctuations of $\phi$, but the quality of these lower bounds will deteriorate as $\beta$ tends to zero. To remove the $\beta$-dependence we prove that the \emph{thermal average} of $\nabla \phi$ is close to $\lambda\nabla u_{\Lambda, \eta}$ uniformly in the temperature (Lemma~\ref{lemma6.1}),
\begin{equation} \label{eq:strategy uniform closeness}
    \left\| \left\langle \nabla \phi \right\rangle_{\mu^{\IV , \beta, \eta}_{\Lambda}} - \lambda \nabla u_{\Lambda, \eta} \right\|_{\underline{L}^2(\Lambda^+)}^2 \leq C.
\end{equation}
The lower bounds on the gradient fluctuations in Theorem~\ref{prop3.100910disc} are an immediate consequence by the corresponding estimates on $u_{\Lambda, \eta}$.

For Theorem~\ref{prop3.1009100bisdisc}, we use distinct strategies for the upper and lower bounds. The upper bound on the height fluctuations uses the inequality~\eqref{eq:11253005} and the Poincar\'e inequality to deduce that
\begin{equation} \label{eq:20451802}
    \left\| \phi - \lambda u_{\Lambda, \eta} \right\|_{\underline{L}^2  \left( \Lambda_L, \mu_{\Lambda_L}^{\IV, \beta, \eta} \right)}^2 \leq C L^2\left\| \nabla \phi - \lambda \nabla u_{\Lambda, \eta} \right\|_{\underline{L}^2  \left( \Lambda^+_L, \mu_{\Lambda_L}^{\IV, \beta, \eta} \right)}^2 \leq C  \left( 1 + \beta^{-1} \right) L^2.
\end{equation}
This inequality is then combined with the estimates of Proposition~\ref{prop.propueta} to show that the field $\phi$ cannot fluctuate significantly more than $u_{\Lambda, \eta}$ in dimensions $d=1,2$, since the fluctuations of $u_{\Lambda, \eta}$ are larger than (the square root of) the right-hand side of~\eqref{eq:20451802}.
The lower bounds on the height fluctuations are more involved. We first reduce to a suitable estimate for $u_{\Lambda_L,\eta}$ by noting that
\begin{align*}
\left\| \phi \right\|_{\underline{L}^2 \left(\Lambda_L^+ ,  \mu^{\IV, \beta ,\eta}_{\Lambda_L}  \right)}  & \geq  \left\| \left\langle \phi \right\rangle_{\mu^{\IV, \beta ,\eta}_{\Lambda_L}} \right\|_{\underline{L}^2 \left(\Lambda_L^+  \right)}
\\ & \geq \| \lambda u_{\Lambda_L,\eta} \|_{\underline{L}^2(\Lambda_L)}  -\left\| \left\langle \phi \right\rangle_{\mu^{\IV, \beta ,\eta}_{\Lambda_L}} - \lambda u_{\Lambda_L,\eta} \right\|_{\underline{L}^2 \left(\Lambda_L^+\right)}
\\ & \geq \| \lambda u_{\Lambda_L,\eta} \|_{\underline{L}^2(\Lambda_L)} - CL
\end{align*}
where we used Jensen's, the triangle and Poincar\'e inequalities and applied the estimate~\eqref{eq:strategy uniform closeness} (this computation can be found in~\eqref{eq:1941} and~\eqref{eq:1942}). Thus, 
\begin{equation*}
\mathbb{P}\left[\left\| \phi \right\|_{\underline{L}^2 \left(\Lambda_L^+ ,  \mu^{\IV, \beta ,\eta}_{\Lambda_L}  \right)} > c L^{2 - d/2}\right] \geq \mathbb{P}\left[\| \lambda u_{\Lambda_L,\eta} \|_{\underline{L}^2(\Lambda_L)} > c L^{2-d/2} + CL\right].
\end{equation*}
As the height fluctuations of $u_{\Lambda_L,\eta}$ are of order $L^{2-d/2}$ (see Proposition~\ref{prop.propueta}), it is natural to expect the right-hand side of the last inequality to be uniformly positive in $L$ in dimensions $d=1,2$ (with a prefactor depending on $\lambda$ which tends to zero as $\lambda$ decreases). To establish this rigorously, we prove a Central Limit Theorem for the projection of $u_{\Lambda_L,\eta}$ in a suitable direction (see~\eqref{eq:CLT}).

We outline the proof of Theorem~\ref{thm10162010} for the ground state of the model and in the case $t=1$ (to simplify the presentation of the argument). Let $w_{\Lambda_L , \eta} : \Lambda_L \to \Z$ be the ground state of the integer-valued random-field Gaussian free field. In this scenario, the conclusion of Theorem~\ref{thm10162010} can be stated as follows: For any vertex $v \in \Lambda_L$,
\begin{equation} \label{eq:151506}
      \P \left(  w_{\Lambda_L , \eta}(v) \neq 0 \right) \geq 1 - e^{-\frac{c}{\lambda^2}}.
\end{equation}
The proof of~\eqref{eq:151506} relies on the following observation: If $w_{\Lambda_L , \eta}(v) \neq 0$, then there exists a finite connected with connected complement $D \subseteq \Zd$ containing the vertex $v$ such that
\begin{equation} \label{eq:10572406}
    \lambda \left| \sum_{x \in D} \eta(x) \right| \geq \left| \partial D \right|.
\end{equation}
The set $D$ can be constructed as follows. Let us assume, without loss of generality, that $w_{\Lambda_L , \eta}(v)>0$. Let $D_0$ be the connected component of the vertex $v$ in the set
\begin{equation*}
\{ x \in \Lambda_L \, : \, w_{\Lambda_L , \eta}(x)> 0 \},
\end{equation*}
and set $D$ to be $D_0$ union all the connected components of $D_0^c$ except the one containing $\partial\Lambda_L$ (this ensures that $D^c$ is connected). 
The ground state $w_{\Lambda_L , \eta}$ necessarily has lower energy than the function $w_{\Lambda_L , \eta} - \indc_{D}$, which implies~\eqref{eq:10572406} upon rearranging the terms in the Hamiltonian.

We next appeal to a result of Fisher--Fr\"{o}hlich--Spencer~\cite{fisher1984ising}, which states that, in dimensions $d \geq 3$, the following holds: for all sufficiently small $\lambda$ and any fixed vertex $v \in \Lambda_L$, the probability (over the disorder $\eta$) that there exists a finite connected set $D \subseteq \Zd$, containing $v$ and having connected complement, such that~\eqref{eq:10572406} holds is at most $e^{-\frac{c}{\lambda^2}}$. The result implies the inequality~\eqref{eq:151506}. The argument can be extended from the ground state to the low-temperature case, as stated in Theorem~\ref{thm10162010}, through a Peierls-type argument.

\subsection{Organisation of the article} 
The rest of the article is organized as follows. Section~\ref{section2} introduces additional notation. Section~\ref{section2.55} collects the tools and preliminary results used in the proofs. Sections~\ref{section4RVS} and~\ref{section4} treat the case of real-valued random surfaces and are devoted to the proofs of Theorems~\ref{prop3.100910} and~\ref{prop3.1009100bis}, respectively. Section~\ref{section4} is devoted to the integer-valued random-field Gaussian free field, and contains the proofs of Theorems~\ref{prop3.100910disc},~\ref{prop3.1009100bisdisc} and~\ref{thm10162010}. Appendix~\ref{appA} provides a proof of the Nash--Aronson estimate used in Section~\ref{section4RVS}, for the heat kernel in a time-dependent uniformly elliptic environment in a box with Dirichlet boundary condition.

\subsubsection*{Acknowledgements} We are indebted to David Huse for a discussion on the possible behavior of the integer-valued random-field Gaussian free field in dimensions $d\ge 3$, and to Michael Aizenman, Charles M. Newman, Thomas Spencer and Daniel L. Stein for encouragement and helpful conversations on the topics of this work. We are grateful to Florian Schweiger for clarifying to us the different choices of boundary conditions in the membrane model and for pointing out a mistake in a previous draft of Section~\ref{sec:discussion and open questions}. We thank two anonymous referees for useful and constructive comments which helped us improve the manuscript. We are also grateful to Antonio Auffinger, Wei-Kuo Chen, Izabella Stuhl and Yuri Suhov for the opportunity to present these results in online talks and for useful discussions.
The research of the authors was supported in part by Israel Science Foundation grants 861/15  and  1971/19  and  by  the  European  Research  Council starting  grant 678520 (LocalOrder).

\section{Notation} \label{section2}

\subsection{General}
Given a (simple) graph $G = (V(G), E(G))$ we let $\vec{E}(G)$ be the set of directed edges of $G$ (each edge in $E(G)$ appears in $\vec{E}(G)$ with both orientations). We write $x\sim y$ to denote that $\{x,y\}\in E(G)$. We often identify subsets $\Lambda\subset G$ with the induced subgraph of $G$ on $\Lambda$. In particular, we write $G$ for $V(G)$ and write $E(\Lambda)$ and $\vec{E}(\Lambda)$ for the edges of $E(G)$ and $\vec{E}(G)$ having both endpoints in $\Lambda$, respectively. We let $\partial \Lambda = \partial_G \Lambda$ be the \emph{external vertex boundary} of $\Lambda$ in $G$,
\begin{equation*}
    \partial \Lambda := \left\{ x \in G\setminus \Lambda \, \colon \, \exists y \in \Lambda, y \sim x  \right\},
\end{equation*}
$\Lambda^+:=\Lambda\cup\partial\Lambda$, and $\left| \Lambda \right|$ be the cardinality of $\Lambda$, sometimes referred to as \emph{discrete volume}.

Let $\Z^d$ be the standard $d$-dimensional lattice, and let $\left| \, \cdot \, \right|$ be the $\ell^\infty$-norm on $\Zd$. Two vertices $v,w\in\Z^d$ are adjacent if they are equal in all but one coordinate and differ by one in that coordinate. Given $M \in \N$ with $M \geq 1$, we denote by $M \Zd$ the set of points of $\Zd$ whose coordinates are divisible by $M$.

We say that a set $\Lambda$ is \emph{connected}, if, for any pair of distinct vertices $x , y \in \Lambda$, there exists a finite collection of points $x_1 , \ldots, x_N \in \Lambda$ such that $x = x_1$, $y = x_N$ and for any $i \in \{ 1 , \ldots, N - 1 \}$, $x_{i}$ and $x_{i+1}$ are adjacent. We say that a set $\Lambda$ is $\star-$\emph{connected}, if, for any pair of distinct vertices $x , y \in \Lambda$, there exists a finite collection of points $x_1 , \ldots, x_N \in \Lambda$ such that $x = x_1$, $y = x_N$ and $\left|x_i - x_{i+1}\right|=1$ for every $1\le i \le N - 1$.

Write $\Lambda_L := \{-L, \ldots, L\}^d\subset\Z^d$ for any integer $L \ge 0$. This is extended to all $L \in [0 , \infty)$ by setting $\Lambda_L:=\Lambda_{\lfloor L\rfloor}$ where $\lfloor L\rfloor$ denotes the floor of $L$. For $x \in \Zd$, we denote by $x + \Lambda_L$ the box $\Lambda_L$ translated by the vector $x$.

Let $a \wedge b$ be the minimum and $a \vee b$ be the maximum of $a , b \in \R$, and by $\lceil a \rceil$ the ceiling of $a \in \R$.

\subsection{$L^2$-Norms} \label{Fctandevts}
Let $G$ be a finite graph. For a function $\phi : G \to \R$, define the $L^2$ and normalized $L^2$-norms of $\phi$ by the formulae
\begin{equation} \label{def.avL2norm}
    \left\| \phi \right\|_{L^2\left( G \right)} := \left( \sum_{x \in G} \left| \phi(x)\right|^2 \right)^\frac12\quad\text{and}\quad \left\| \phi \right\|_{\underline{L}^2 \left( G\right)} := \left( \frac{1}{\left| G\right|}\sum_{x \in G} \left| \phi(x)\right|^2 \right)^\frac12.
\end{equation}
Define the discrete gradient
\begin{equation}\label{eq:discrete gradient}
\nabla \phi(e) := \phi(y) - \phi(x)\quad\text{for directed edges $e = (x , y) \in \vec{E}(G)$}.
\end{equation}
In expressions which do not depend on the orientation of the edge, such as $|\nabla\phi(e)|^2$ or $V(\nabla\phi(e))$, we allow the edge $e$ to be undirected. For a function $v:\vec{E}(G)\to\R$ satisfying $v((x,y)) = -v((y,x))$ (such as the function $\nabla\phi$) define the $L^2$ and normalized $L^2$-norms of $v$ by the formulae
\begin{equation} \label{eq:11151509}
    \left\| v \right\|_{L^2\left( G \right)} := \left( \sum_{e \in E \left( G\right)} \left| v(e)\right|^2 \right)^\frac12 ~\mbox{and}~ \left\| v \right\|_{\underline{L}^2 \left( G\right)} := \left( \frac{1}{\left| G\right|}\sum_{e \in E \left( G\right)} \left| v(e)\right|^2 \right)^\frac12.
\end{equation}

\subsection{Environments and operators}
Let $G$ be a graph and introduce the notation, for each vertex $x \in G$,
\begin{equation} \label{eq:11121509}
    \sum_{e \ni x} := \sum_{\{e \in \vec{E}(G)\,\colon\, \exists y \in G, e=(x , y)\}}.
\end{equation}
A map $\a : E(G) \to \R$ is called an \emph{environment}. Its definition is extended to directed edges by setting $\a((x,y)):=\a(\{x,y\})$ for $(x,y)\in\vec{E}(\Lambda)$. The operator $-\nabla \cdot \a \nabla$ is defined by the formula
\begin{equation} \label{eq:11330408}
    \nabla \cdot \a \nabla \phi(x) = \sum_{e \ni x} \a (e) \nabla \phi(e)
\end{equation}
for a function $\phi : G \to \R$ and $x\in G$. Unravelling the definitions shows that, for any pair of functions functions $\phi , \psi : G \to \R$, the following discrete integration by parts identity holds:
\begin{equation}\label{eq:discrete integration by parts}
    -\sum_{x \in G} \left( \nabla \cdot \a \nabla \phi(x) \right) \psi(x) = \sum_{e \in E \left(G \right)} \a (e) \nabla \phi(e) \nabla \psi(e)
\end{equation}
in the sense that if one side converges absolutely then the other converges absolutely to the same value. Note that the terms inside the sum on the right-hand side are well defined for undirected edges. We say that the environment $\a$ is uniformly elliptic if there exist $c_-,c_+$ such that $c_- \leq \a \leq c_+$, pointwise.

The above definitions naturally extend to \emph{time-dependent environments} $\a : I \times E(G) \to \R$, where $I\subset\R$ is a (time) interval: The operator $-\nabla\cdot\a\nabla$ acts on time-dependent functions $\phi : I \times G \to \R$ with the same definition~\eqref{eq:11330408} applied at each fixed time. The identity~\eqref{eq:discrete integration by parts} then holds at each fixed time for time-dependent functions $\phi,\psi$.

The discrete Laplacian $\Delta$ is the operator $\nabla \cdot \a \nabla$ with $\a \equiv 1$.

\section{Tools} \label{section2.55}

In this section, we collect tools pertaining to random surfaces, concentration inequalities and estimates on the solution of parabolic equations which are used in the proofs of Theorem~\ref{prop3.100910} and Theorem~\ref{prop3.1009100bis}.

\subsection{Langevin dynamics} \label{sectionLangevin}

The Gibbs measure $\mu_{\Lambda_L}^{\eta}$ (defined in~\eqref{eq:defmuLeta}) is naturally associated with the following dynamics.

\begin{definition}[Langevin dynamics]
Given an integer $L \ge 0$, random field strength $\lambda$, external field $\eta:\Lambda_L\to\R$ and a collection of independent standard Brownian motions $\left\{ B_t(x) \, : \, x \in \Lambda_L\right\}$, define the Langevin dynamics $\left\{ \phi_t(x) \, : \, x \in \Lambda_L\right\}$ to be the solution of the system of stochastic differential equations
\begin{equation} \label{eq:09262707bis}
    \left\{ \begin{aligned}
    \di \phi_t (y) &=  \sum_{e \ni y} V' \left( \nabla \phi_t(e) \right) \di t + \lambda \eta(y)\di t + \sqrt{2}\,\di B_t(y) &&(t , y) \in  (0 , \infty) \times \Lambda_L, \\
    \phi_t(y) &= 0 &&(t , y) \in  (0 , \infty) \times \partial \Lambda_L, \\
    \phi_0(y) &= 0 &&y \in \Lambda_L.
    \end{aligned} \right.
\end{equation}
\end{definition}

The Langevin dynamics~\eqref{eq:09262707bis} is stationary, reversible and ergodic with respect to the Gibbs measure $\mu_{\Lambda_L}^\eta$; in particular, one has the convergence
\begin{equation} \label{eq:ergLangevin}
    \left\langle \phi_t (0) \right\rangle \underset{t \to \infty}{\longrightarrow}  \left\langle \phi (0) \right\rangle_{\mu_{\Lambda_L}^\eta},
\end{equation}
where, in a slight abuse of notation, we use the symbol $\left\langle \cdot \right\rangle$ in the left-hand side to refer to the expectation with respect to the Brownian motions $\left\{ B_t(x) \, : \, x \in \Lambda_L \right\}$ . 

\subsection{The Efron--Stein inequality} \label{section2.5}

We record below the Efron--Stein inequality which will be used in the proof of Theorem~\ref{prop3.1009100bis}. A proof can be found in~\cite[Theorem 3.1]{BGM13}.

\begin{proposition}[Efron--Stein inequality] \label{prop.Efronstein}
Let $\Lambda$ be a finite set, let $\eta,\tilde \eta:\Lambda\to\R$ be independent and identically distributed random vectors with independent coordinates, and let $f : \R^{\Lambda} \to \R$ be a measurable map satisfying $\E \left[ f(\eta)^2 \right] < \infty$. For each $x \in \Lambda$, set $\eta^{x}$ to be the field defined by the formula $\eta^x (y) = \eta(y)$ if $y \neq x$ and $\eta^x (x) = \tilde \eta(x)$. Then
\begin{equation*}
    \var \left[ f \right] \leq \frac 12 \sum_{x \in \Lambda} \E \left[ \left( f(\eta) - f\left(\eta^{x}\right) \right)^2 \right].
\end{equation*}
\end{proposition}

\subsection{Heat kernel bounds} \label{sectionheatkernel}
Let $L\ge 0$ be an integer, let $0<c_-\le c_+<\infty$, let $s_0\in\R$ and let $y\in\Lambda_L$. Let $\a : [s_0 , \infty) \times E(\Lambda_L) \to [c_- , c_+]$ be a \emph{continuous} time-dependent (uniformly elliptic) environment. For each initial time $s \geq s_0$, denote by $P_\a= P_\a(\cdot,\cdot;s,y):[s,\infty)\times\Lambda_L^+\to[0,1]$ the heat kernel associated with Dirichlet boundary conditions in the box $\Lambda_L$, i.e., the solution of the parabolic equation
\begin{equation} \label{eq:16041308}
    \left\{ \begin{aligned}
    \partial_t P_\a (t , x ; s  ,y) - \nabla \cdot \a \nabla P_\a (t , x ; s  ,y) &= 0 &&(t , x) \in (s , \infty) \times \Lambda_L, \\
    P_\a (t , x ; s  ,y) &= 0 &&(t , x) \in (s , \infty) \times \partial \Lambda_L, \\
    P_\a(s , x ; s  ,y) &= \indc_{\{ x = y \}} &&x \in \Lambda_L^+.
    \end{aligned} \right.
\end{equation}
The maximum principle ensures that $P_\a$ is non-negative. Summing $\nabla\cdot \a \nabla P_\a$ over $\Lambda_L$ and integrating by parts implies that this sum is non-positive, which, in turn, shows that $\sum_{x\in\Lambda_L} P_\a(t,x;s,y)$ is a non-increasing function of $t$ on $[s,\infty)$. 
In particular, since the value $\sum_{x\in\Lambda_L} P_\a(s,x;s,y)$ is equal to $1$, one has the estimate, for any $t \geq s$,
\begin{equation} \label{eq:15150110}
    \sum_{x\in\Lambda_L} P_\a(t,x;s,y) \leq 1.
\end{equation}

Upper and lower bounds on heat kernels are usually referred to as Nash--Aronson estimates. They were first established by Aronson in~\cite{Ar}, in the continuous setting and in infinite volume for parabolic equations with time-dependent and uniformly elliptic environment. In the discrete setting, we refer to the article of Delmotte~\cite{De99} and the references therein for a collection of heat kernel estimates on general graphs in static environment. The case of discrete parabolic equations with dynamic and uniformly elliptic environment is treated by Giacomin--Olla--Spohn in~\cite[Appendix B]{GOS}. The proposition stated below is a finite-volume version of their result.

\begin{proposition}[Nash--Aronson estimates] \label{prop.NashAronson}
In the above setup with $L\ge 1$ there exist positive constants $C_0 , c_0$ depending only on the dimension $d$ and the ratio of ellipticity  $c_+ / c_-$ such that the following holds. For all $t\ge s$ and $x\in \Lambda_L$,
\begin{equation} \label{Nash.supprop3.3}
     P_\a \left(t , x ; s  ,y \right) \leq  \frac{C_0}{1 \vee \left( c_-(t-s) \right)^{\frac d2}} \exp \left( - \frac{c_0|x - y|}{1 \vee \left( c_- (t-s) \right)^{\frac 12}} \right) \exp \left( - \frac{c_0c_- (t-s)}{L^2} \right).
\end{equation}
In addition, there exists a constant $c_1 > 0$ depending only on $d$ and $c_+ / c_-$ such that for any $t\ge s$ and any $(x,y)\in\Lambda_L \times \Lambda_{L/2}$ satisfying $|x - y| \leq \sqrt{c_- (t - s)} \leq c_1 L$,
\begin{equation} \label{Nash.infprop3.3}
     P_\a \left(t , x ; s  ,y \right) \geq \frac{c_0}{1 \vee \left( c_-(t-s) \right)^{\frac d2}}.
\end{equation}
\end{proposition}

The proof of this result is the subject of Appendix~\ref{appA}.

\begin{remark}
The $L$-dependent term in the right-hand side of~\eqref{Nash.supprop3.3} is a consequence of the Dirichlet boundary condition in the definition~\eqref{eq:16041308} of the heat kernel $P_\a$. It shows that the map $P_\a$ decays exponentially fast after a time of order $L^2$. It is obtained by analytic arguments in Appendix~\ref{appA} and has a natural probabilistic interpretation: If we denote by $(X_t)_{t \geq s}$ the continuous time random-walk evolving in the time dependent environment $\a$ (see~\cite[Section 3.2]{GOS} for a formal definition of the process) started at time $s$ from the vertex $y$, then one has the identity
\begin{equation} \label{eq:15071902}
    P_\a \left(t , x ; s  ,y \right) = \P \left[ X_t = x, \, \tau_{\Lambda_L } \geq t \right],
\end{equation}
where $\tau_{\Lambda_L}$ is the random walk hitting time of the boundary $\partial \Lambda_L$. Under the assumption that the environment $\a$ is uniformly elliptic, the random walk $X_t$ will typically hit the boundary $\partial \Lambda_L$ in a time of order $L^2$, and the probability that it remains in the box $\Lambda_L$ for a time $t-s$ will exhibit an exponential decay of the form $\exp \left( - c (t-s)/L^2 \right)$. This observation implies the exponential decay of the heat kernel $P_\a$ for large $t-s \geq L^2$ by the identity~\eqref{eq:15071902}. The inequality~\eqref{Nash.supprop3.3}, proved in Appendix~\ref{appA} using equivalent, analytic techniques, can be deduced from a refinement of the previous argument and gives an accurate description of the heat kernel $P_\a$ over the entire time line $t \in (s , \infty)$.
\end{remark}

\subsection{Probability density identities}\label{sec:IBP}

Suppose $f:\R^n\to[0,\infty)$ is a continuously differentiable probability density such that $|y|f(y)$ tends to zero at infinity and satisfies that $y \to (1 + |y|)\nabla f(y)$ is integrable. Integration by parts implies that, for each index $1\le j\le n$,
\begin{align}
  \int_{\R^n} \frac{d f(y)}{dy_j} dy &= 0,\label{eq:translation identity}\\
  \int_{\R^n} y_j\frac{d f(y)}{dy_j} dy &= -1.\label{eq:multiplication identity}
\end{align}

Let $\Lambda\subset\Z^d$ be finite and $\eta:\Lambda\to\R$. Applying the above identities to the probability density of $\mu_\Lambda^\eta$ (see~\eqref{eq:defmuLeta}) under the assumption~\eqref{eq:V ellipticity} shows that, for each $x\in\Lambda$,
\begin{align}
  &-\langle \sum_{e \ni x} V'(\nabla\phi(e))\rangle_{\mu_\Lambda^\eta} = \lambda\, \eta(x),\label{eq:first quenched identity}\\
  &-\langle \phi(x)\sum_{e \ni x} V'(\nabla\phi(e))\rangle_{\mu_\Lambda^\eta} = 1+\lambda\, \eta(x)\langle \phi(x) \rangle_{\mu_\Lambda^\eta}.\label{eq:second quenched identity}
\end{align}

\section{Gradient fluctuations in the real-valued case} \label{section4RVS}

The objective of this section is to prove the delocalization of the gradient of the real-valued random-field random surfaces in dimensions $d \leq 2$, and its localization in dimensions $d \geq 3$, proving Theorem~\ref{prop3.100910}.

\subsection{The ground state of the random-field Gaussian free field} \label{secueta}
Given a finite $\Lambda\subset\Z^d$ and a function $\eta:\Lambda\to\R$, let $u_{\Lambda,\eta} : \Lambda^+ \to \R$ be the solution of the Dirichlet problem
\begin{equation} \label{def.ueta}
    \left\{ \begin{aligned}
    -\Delta u_{\Lambda,\eta} &= \eta &&\mbox{in}~\Lambda, \\
    u_{\Lambda,\eta} &= 0 &&\mbox{on}~\partial\Lambda.
    \end{aligned}
    \right.
\end{equation}
One readily checks that $u_{\Lambda,\eta}$ is the ground state of the random-field Gaussian free field in $\Lambda$ with zero boundary conditions and unit disorder strength. That is, $u_{\Lambda,\eta}$ minimizes the Hamiltonian $H_\Lambda^\eta$ given by~\eqref{eq:101071608}, with $V(x) = \frac{1}{2}x^2$ and $\lambda=1$, among all functions $\phi$ which equal zero on $\partial\Lambda$. The function $u_{\Lambda,\eta}$ will be instrumental in analyzing the gradient fluctuations of the real-valued random-field random surfaces and will also play a role in our analysis of the integer-valued random-field Gaussian free field in Section~\ref{section66}.

The next proposition studies the order of magnitude of $u_{\Lambda,\eta}$ and its gradient when $\Lambda = \Lambda_L$ and $\eta$ is random with independent and normalized values.
\begin{proposition}\label{prop.propueta}
Suppose $(\eta(x))_{x\in\Lambda}$ are independent with zero mean and unit variance. Then
\begin{enumerate}[label=(\roman*)]
  \item\label{item:mean zero} The random variables $(u_{\Lambda,\eta}(x))_{x\in\Lambda}$ and $(\nabla u_{\Lambda,\eta} (e))_{e\in\vec{E}(\Lambda)}$ have zero mean.
\end{enumerate}
Now let $\Lambda=\Lambda_L$ for an integer $L\ge 2$. There exist constants $C, c > 0$ depending only on the dimension $d$ such that
\begin{enumerate}[label=(\roman*), resume*]
    \item The following upper bounds hold for all $x\in\Lambda_L$,
    \begin{equation} \label{eq:10321908}
    \E \left[u_{\Lambda_L,\eta}(x)^2 \right] \leq \left\{ \begin{aligned}
        &C L^{4-d} &&1\le d\leq 3,\\
        &C \ln L  &&d =4, \\
        &C &&d \geq 5;
        \end{aligned} \right.
    \end{equation}
    \item The following lower bounds hold for all $x \in\Lambda_{c L}$,
    \begin{equation} \label{eq:10331908}
    \E \left[ u_{\Lambda_L,\eta}(x)^2 \right] \geq \left\{ \begin{aligned}
        &c L^{4-d} &&1\le d\leq 3,\\
        &c \ln L  &&d =4,\\
        &c &&d \geq 5;
        \end{aligned} \right.
    \end{equation}
    \item\label{item:L 2 bounds for gradient of u eta} The averaged $L^2$-norm of the gradient of $u_{\Lambda_L,\eta}$ satisfies
    \begin{equation} \label{eq:11421908}
    \E \left[ \left\| \nabla u_{\Lambda_L,\eta} \right\|_{\underline{L}^2 \left( \Lambda_L^+ \right)}^2 \right] \approx \left\{ \begin{aligned}
        &L &&d=1,\\
        &\ln L  &&d =2, \\
        &1&&d \geq 3,
        \end{aligned} \right.
    \end{equation}
    where $a\approx b$ is used here in the sense $c\cdot a\le b\le C\cdot a$.
\end{enumerate}
\end{proposition}

\begin{proof}
Let $G_\Lambda:\Lambda^+\times\Lambda\to[0,\infty)$ be the elliptic Green's function with Dirichlet boundary condition, defined by requiring that for all $y \in \Lambda$,
\begin{equation} \label{eq:100214bis}
    \left\{ \begin{aligned}
    -\Delta G_\Lambda(\cdot, y) &= \delta_y &&\mbox{in}~\Lambda, \\
     G_\Lambda(\cdot, y) &= 0 &&\mbox{on}~ \partial \Lambda
    \end{aligned}
    \right.
\end{equation}
with $\delta_y$ the Kronecker delta function.
The maximum principle shows that $G_\Lambda$ is indeed non-negative. The linearity of~\eqref{def.ueta} and~\eqref{eq:100214bis} implies that
\begin{equation} \label{eq:16051008}
    u_{\Lambda,\eta}(x) = \sum_{y \in \Lambda} G_\Lambda (x , y) \eta(y) \hspace{3mm} \mbox{and} \hspace{3mm} \nabla u_{\Lambda,\eta} (e) = \sum_{y \in \Lambda} \nabla G_\Lambda (e , y) \eta(y)
\end{equation}
for $x \in \Lambda^+$ and $e \in \vec{E} \left( \Lambda_L^+\right)$. Property~\ref{item:mean zero} is an immediate consequence. In addition, our assumptions on $\eta$ imply that
\begin{equation} \label{eq:10121908}
    \E \left[ u_{\Lambda,\eta}(x)^2 \right] = \sum_{y \in \Lambda} G_\Lambda (x , y)^2  \hspace{3mm} \mbox{and} \hspace{3mm} \E \left[ \left(\nabla u_{\Lambda,\eta} (e) \right)^2 \right] = \sum_{y \in \Lambda} \left(\nabla G_\Lambda (e , y)\right)^2.
\end{equation}

Now specialize to the case $\Lambda = \Lambda_L$ with $L\ge 2$.
Denote by $P_{\Lambda_L}$ the solution of the parabolic equation~\eqref{eq:16041308} in the specific case of the heat equation (i.e., in the case $\a \equiv 1$). The elliptic Green's function $G_{\Lambda_L}$ is related to the heat kernel $P_{\Lambda_L}$ through the identity: For any $x\in\Lambda_L^+$ and $y \in \Lambda_L$,
\begin{equation} \label{eq:02031902}
    G_{\Lambda_L} \left( x , y \right) = \int_0^\infty P_{\Lambda_L} \left( t , x ;0 , y \right) \, dt.
\end{equation}
The identity~\eqref{eq:02031902} can be checked by a direct computation: Proposition~\ref{prop.NashAronson} implies that $P_{\Lambda_L} \left( t , x ;0 , y \right)$ tends to $0$ as $t$ tends to infinity for any $x , y \in \Lambda_L$ and that the right-hand side of~\eqref{eq:02031902} is well-defined. Taking the discrete Laplacian of the map $x \mapsto \int_0^\infty P_{\Lambda_L} \left( t , x ;0 , y \right) \, dt$ and using the definition of the heat kernel $P_{\Lambda_L}$, we obtain
\begin{align}
    -\Delta \int_0^\infty P_{\Lambda_L} \left( t , x ;0 , y \right) \, dt = \int_0^\infty -\Delta P_{\Lambda_L} \left( t , x ;0 , y \right) \, dt = \int_0^\infty -\partial_t P_{\Lambda_L} \left( t , x ;0 , y \right) \, dt & = P(0 , x ; 0 , y) \\
    & = \indc_{\{ x = y\}}. \notag
\end{align}
Consequently the map $x \mapsto \int_0^\infty P_{\Lambda_L} \left( t , x ;0 , y \right) \, dt$ solves the equation~\eqref{eq:100214bis} for each $y\in\Lambda_L$. Since this equation has a unique solution, it implies the identity~\eqref{eq:02031902}.
Proposition~\ref{prop.NashAronson} additionally shows the upper bounds
\begin{equation} \label{eq:10281908}
     G_{\Lambda_L} \left( x , y \right) \leq \left\{
     \begin{aligned}
     &C L &d=1 , \\
     &C \ln \left( \frac{C L}{ 1 \vee |x - y|} \right) &d=2, \\
       &\frac C{1 \vee |x - y|^{d-2}} &d\geq 3
     \end{aligned}
     \right.
\end{equation}
valid for all $x, y \in \Lambda_L$, and the lower bounds
\begin{equation} \label{eq:10291908}
     G_{\Lambda_L} \left( x , y \right) \geq \left\{
     \begin{aligned}
     &c L &d=1 , \\
     &c \ln \left( \frac{L}{ 1 \vee |x - y|} \right) &d=2, \\
       &\frac c{1 \vee |x - y|^{d-2}} &d\geq 3,
     \end{aligned}
     \right.
\end{equation}
valid for $y \in \Lambda_{L/2}$ and $x \in \Lambda_L$ such that $|x - y| \leq c L$, with $C,c>0$ depending only on $d$.
Combining the identity~\eqref{eq:10121908} with the upper and lower bounds~\eqref{eq:10281908} and~\eqref{eq:10291908} shows the upper and lower bounds on the value of $\E \left[ u_{\Lambda_L,\eta}(x)^2 \right]$ stated in~\eqref{eq:10321908} and~\eqref{eq:10331908}.

We proceed to obtain the estimates on the averaged $L^2$-norm of $\nabla u_{\Lambda_L,\eta}$. The identity~\eqref{eq:10121908} implies that
\begin{equation}\label{eq:averaged L_2 norm of u_eta initial}
    \E \left[ \left\| \nabla u_{\Lambda_L,\eta} \right\|_{\underline{L}^2 \left( \Lambda_L^+ \right)}^2 \right] =  \frac1{\left| \Lambda_L^+ \right|} \sum_{e \in E \left( \Lambda_L^+ \right), \, y \in \Lambda_L} \left(\nabla G_{\Lambda_L} (e , y)\right)^2.
\end{equation}
Applying~\eqref{eq:discrete integration by parts} (with $G = \Lambda_L^+$, $a\equiv 1$ and $\phi=\psi=G_{\Lambda_L}(\cdot,y)$) and using~\eqref{eq:100214bis} we see that
\begin{equation} \label{eq:11411908}
   \sum_{e \in E \left( \Lambda_L^+ \right)} \left(\nabla G_{\Lambda_L} (e , y)\right)^2 = G_{\Lambda_L} (y , y)
\end{equation}
for each $y \in \Lambda_L$.
Substituting in~\eqref{eq:averaged L_2 norm of u_eta initial}, we obtain
\begin{equation*}
    \E \left[ \left\| \nabla u_{\Lambda_L,\eta} \right\|_{\underline{L}^2 \left( \Lambda_L^+ \right)}^2 \right] = \frac1{\left| \Lambda_L^+ \right|} \sum_{y \in \Lambda_L} G_{\Lambda_L}(y, y).
\end{equation*}
Applying the upper and lower bounds~\eqref{eq:10281908} and~\eqref{eq:10291908} with the choice $x = y$ and using the non-negativity of the Green's function $G_{\Lambda_L}$ we obtain the inequalities in~\eqref{eq:11421908}.
\end{proof}

\subsection{Gradient fluctuations} \label{subsec3.1}

In this section, we prove Theorem~\ref{prop3.100910} for the uniformly convex $\nabla \phi$-model, as a consequence of the following result which holds for any choice of $\eta$. Theorem~\ref{prop3.100910} follows from the result upon setting $\Lambda = \Lambda_L$ and letting $\eta$ be random as in the theorem and then taking an expectation over $\eta$ and applying the bounds in item~\ref{item:L 2 bounds for gradient of u eta} of Proposition~\ref{prop.propueta}.
\begin{theorem}\label{eq:quenched gradient fluctuations}
  Let $\Lambda\subset\Z^d$ be finite, $\lambda>0$ and $0<c_-\le c_+<\infty$. Let $\eta:\Lambda\to\R$ and let $u_{\Lambda, \eta}:\Lambda^+\to\R$ be the solution of~\eqref{def.ueta}. Suppose $\phi$ is sampled from the measure $\mu^\eta_\Lambda$ of~\eqref{eq:defmuLeta}, where the potential $V$ satisfies~\eqref{eq:V ellipticity}. Then the quantity
  \begin{equation}\label{eq:gradient norm in two spaces}
    \left\| \nabla \phi \right\|_{L^2  \left( \Lambda^+, \mu_\Lambda^\eta\right)} := \left(\sum_{e \in E \left( \Lambda^+\right)} \langle\left( \nabla\phi(e)\right)^2\rangle_{\mu_\Lambda^\eta} \right)^\frac12
  \end{equation}
  satisfies
  \begin{equation}\label{eq:ComparisonOfL2Norms}
    \frac{\lambda}{c_+}\left\| \nabla u_{\Lambda,\eta} \right\|_{L^2 \left( \Lambda^+\right)} \le \left\| \nabla \phi \right\|_{L^2  \left( \Lambda^+, \mu_\Lambda^\eta\right)} \le \frac{2\lambda}{c_-}\left\| \nabla u_{\Lambda,\eta} \right\|_{L^2 \left( \Lambda^+\right)} + \sqrt{\frac{2|\Lambda|}{c_-}}.
  \end{equation}
\end{theorem}
\begin{proof}
  We start with the proof of the lower bound. Multiply the identity~\eqref{eq:first quenched identity} by $u_{\Lambda,\eta}(x)$ and sum over all $x\in \Lambda$ to obtain
  \begin{equation}\label{eq:gradient lower bound initial equality}
    -\sum_{x\in\Lambda} u_{\Lambda, \eta}(x)\langle \sum_{e \ni x} V'(\nabla\phi(e))\rangle_{\mu_\Lambda^\eta} = \lambda\sum_{x\in\Lambda} u_{\Lambda, \eta}(x)\eta(x).
  \end{equation}
  The left-hand side is developed by observing that each edge appears in the sum with both orientations and that $V'$ is an odd function,
  \begin{equation*}
    -\sum_{x\in\Lambda} u_{\Lambda, \eta}(x)\langle \sum_{e \ni x} V'(\nabla\phi(e))\rangle_{\mu_\Lambda^\eta} = \sum_{e\in E(\Lambda^+)} \nabla u_{\Lambda,\eta}(e) \langle V'(\nabla\phi(e))\rangle_{\mu_\Lambda^\eta}
  \end{equation*}
  (where the boundary terms may be added by noting that $u_{\Lambda,\eta}\equiv 0$ on $\partial\Lambda$).
  The right-hand side of~\eqref{eq:gradient lower bound initial equality} is developed as
  \begin{equation*}
    \sum_{x\in\Lambda} u_{\Lambda, \eta}(x)\eta(x) = -\sum_{x\in\Lambda} u_{\Lambda, \eta}(x)\Delta u_{\Lambda,\eta}(x) = \sum_{e\in E(\Lambda^+)} (\nabla u_{\Lambda, \eta}(e))^2 = \left\| \nabla u_{\Lambda, \eta} \right\|_{L^2\left( \Lambda^+ \right)}^2
  \end{equation*}
  by using the equation~\eqref{def.ueta} and the equality~\eqref{eq:discrete integration by parts}. Substituting the last two equalities in~\eqref{eq:gradient lower bound initial equality} shows that
  \begin{equation}
    \lambda \left\| \nabla u_{\Lambda, \eta} \right\|_{L^2\left( \Lambda^+ \right)}^2 = \sum_{e\in E(\Lambda^+)} \nabla u_{\Lambda,\eta}(e) \langle V'(\nabla\phi(e))\rangle_{\mu_\Lambda^\eta}.
  \end{equation}
  An application of the Cauchy-Schwarz inequality to the right-hand side then gives
  \begin{equation}\label{eq:gradient lower bound middle inequality}
    \lambda \left\| \nabla u_{\Lambda, \eta} \right\|_{L^2\left( \Lambda^+ \right)}\le \big\| \langle V'(\nabla\phi)\rangle_{\mu_\Lambda^\eta} \big\|_{L^2\left( \Lambda^+ \right)}.
  \end{equation}
  Finally, as $V$ is an even function satisfying~\eqref{eq:V ellipticity}, it follows that
  \begin{equation}\label{eq:derivative of V estimate}
    c_- x^2 \le V'(x) x\le c_+ x^2
  \end{equation}
  for all $x\in\R$. In particular, $|V'(x)|\le c_+ |x|$. Substituting in~\eqref{eq:gradient lower bound middle inequality} shows that
  \begin{equation*}
    \lambda \left\| \nabla u_{\Lambda, \eta} \right\|_{L^2\left( \Lambda^+ \right)} \le c_+ \big\| \langle|\nabla\phi|\rangle_{\mu_\Lambda^\eta} \big\|_{L^2\left( \Lambda^+ \right)}.
  \end{equation*}
  The Cauchy-Schwarz inequality shows that $\langle|\nabla\phi(e)|\rangle_{\mu_\Lambda^\eta}^2\le \langle(\nabla\phi(e))^2\rangle_{\mu_\Lambda^\eta}$ and the lower bound of the theorem follows.

  We proceed to prove the upper bound. Sum the identity~\eqref{eq:second quenched identity} over all $x\in \Lambda$ to obtain
  \begin{equation}\label{eq:gradient upper bound initial equality}
    -\sum_{x\in\Lambda} \langle \phi(x)\sum_{e \ni x} V'(\nabla\phi(e))\rangle_{\mu_\Lambda^\eta} = |\Lambda| + \lambda \sum_{x\in\Lambda^+}\eta(x)\langle \phi(x) \rangle_{\mu_\Lambda^\eta}
  \end{equation}
  (where the terms with $x\in\partial\Lambda$ are added by noting that $\phi\equiv 0$ on $\partial\Lambda$). As in the analysis that lead to the lower bound, the left-hand side is developed by summing over the two orientations of each edge, and thus
  \begin{equation*}
    -\sum_{x\in\Lambda} \langle \phi(x)\sum_{e \ni x} V'(\nabla\phi(e))\rangle_{\mu_\Lambda^\eta} = \sum_{e\in E(\Lambda^+)} \langle\nabla \phi(e) V'(\nabla\phi(e))\rangle_{\mu_\Lambda^\eta}.
  \end{equation*}
  The right-hand side of~\eqref{eq:gradient upper bound initial equality} is developed as
  \begin{equation*}
    \sum_{x\in\Lambda} \eta(x)\langle \phi(x) \rangle_{\mu_\Lambda^\eta} = -\sum_{x\in\Lambda^+} \Delta u_{\Lambda,\eta}(x)\langle \phi(x) \rangle_{\mu_\Lambda^\eta} = \sum_{e\in E(\Lambda^+)} \nabla u_{\Lambda, \eta}(e) \langle \nabla \phi(e) \rangle_{\mu_\Lambda^\eta}
  \end{equation*}
  by using the equation~\eqref{def.ueta} and the equality~\eqref{eq:discrete integration by parts}. Substituting in~\eqref{eq:gradient upper bound initial equality} shows that
  \begin{equation}
    \sum_{e\in E(\Lambda^+)} \langle\nabla \phi(e) V'(\nabla\phi(e))\rangle_{\mu_\Lambda^\eta} = |\Lambda| + \lambda\sum_{e\in E(\Lambda^+)} \nabla u_{\Lambda, \eta}(e) \langle \nabla \phi(e) \rangle_{\mu_\Lambda^\eta}.
  \end{equation}
  The left-hand side of the equality is developed using the lower bound in~\eqref{eq:derivative of V estimate} while the right-hand side is developed using the Cauchy-Schwarz inequality, yielding
  \begin{equation}
    c_-\sum_{e\in E(\Lambda^+)} \langle(\nabla \phi(e))^2 \rangle_{\mu_\Lambda^\eta}\le |\Lambda| + \lambda\big\| \nabla u_{\Lambda,\eta} \big\|_{L^2\left( \Lambda^+ \right)} \big\| \langle\nabla\phi\rangle_{\mu_\Lambda^\eta} \big\|_{L^2\left( \Lambda^+ \right)}.
  \end{equation}
  Recalling the definition of $\left\| \nabla \phi \right\|_{L^2  \left( \Lambda^+, \mu_\Lambda^\eta\right)}$ from~\eqref{eq:gradient norm in two spaces} and using that $a+b\le 2\max(a,b)$ for $a,b>0$ and $\langle\nabla\phi(e)\rangle_{\mu_\Lambda^\eta}^2\le \langle\nabla\phi(e)^2\rangle_{\mu_\Lambda^\eta}$ (which follows from Jensen's inequality), we conclude that
  \begin{equation}
    c_-\left\| \nabla \phi \right\|_{L^2  \left( \Lambda^+, \mu_\Lambda^\eta\right)}^2\le 2\max\left\{|\Lambda|, \lambda\big\| \nabla u_{\Lambda,\eta} \big\|_{L^2\left( \Lambda^+ \right)} \left\| \nabla \phi \right\|_{L^2  \left( \Lambda^+, \mu_\Lambda^\eta\right)}\right\}.
  \end{equation}
  Rearranging we get
  \begin{equation}
  \left\| \nabla \phi \right\|_{L^2  \left( \Lambda^+, \mu_\Lambda^\eta\right)}\le \max\left\{\sqrt{\frac{2|\Lambda|}{c_-}}, \frac{2\lambda}{c_-}\left\| \nabla u_{\Lambda,\eta} \right\|_{L^2 \left( \Lambda^+\right)}\right\}
  \end{equation}
  which implies the upper bound of the theorem.
\end{proof}

\begin{remark}\label{rem:gradient fluctuations relaxed assumptions}
An inspection of the proof of Theorem~\ref{eq:quenched gradient fluctuations} shows that the assumption that $V$ satisfies~\eqref{eq:V ellipticity} may be relaxed. Specifically, the proof of the lower bound in Theorem~\ref{eq:quenched gradient fluctuations} requires only that $V$ is a sufficiently smooth even function, having sufficient growth at infinity for the identity~\eqref{eq:first quenched identity} to hold, and that $|V'(x)|\le c_+ |x|$ for all $x\in\R$. Similarly, the proof of the upper bound in Theorem~\ref{eq:quenched gradient fluctuations} requires only that $V$ is a sufficiently smooth even function, having sufficient growth at infinity for the identity~\eqref{eq:second quenched identity} to hold, and that $V'(x)\ge c_- x$ for all $x \geq 0$. Note that $V$ need not be convex for these relaxed assumptions to hold.

Consequently, the lower and upper bounds in Theorem~\ref{prop3.100910} hold under the same relaxed assumptions.

Section~\ref{sec:general potentials discussion} contains additional discussion on general potentials.
\end{remark}

\section{Height fluctuations in the real-valued case}
\label{section4}

In this section, we study the fluctuations of the height of real-valued random-field random surfaces. We prove that the surface delocalizes in dimensions $1\le d\le 4$ and that it localizes in dimensions $d\ge 5$, proving Theorem~\ref{prop3.1009100bis}. Quantitative upper and lower bounds for the variance of the thermal expectation of the height are obtained. The section is organized as follows. In Subsection~\ref{sec.quant.thm.arbit.ext.field}, we establish a quantitative theorem which estimates the difference of the thermal expectations of the height of the random surface at the center of a box with two different external fields. Subsection~\ref{section.section4.2} is devoted to the proof of the upper bounds~\eqref{eq:TV0904041},~\eqref{eq:TV090404bis1} and~\eqref{eq:TV090404ter1} by combining the results of Subsection~\ref{sec.quant.thm.arbit.ext.field} with the Efron--Stein inequality following the outline presented in Section~\ref{sec.strat}. In Section~\ref{subsection4.3}, we prove the lower bounds~\eqref{eq:TV0904041low},~\eqref{eq:TV090404bis1low} and~\eqref{eq:TV090404ter1low}, and thus complete the proof of Theorem~\ref{prop3.1009100bis}.

\subsection{A quantitative estimate for the thermal expectation of the height of the random surface.} \label{sec.quant.thm.arbit.ext.field}

This section is devoted to the proof of an upper bound and a lower bound on the difference of the thermal expectations of the height at the center of a box with two different (and arbitrary) external fields. The argument relies on a coupling argument for the Langevin dynamics associated with the random-field $\nabla \phi$-model and on the Nash--Aronson estimate for parabolic equation with uniformly convex and time-dependent environment (Proposition~\ref{prop.NashAronson}).

\begin{theorem} \label{Thm6}
Let $L\geq 2$, $\lambda > 0$ and $\eta , {\bar \eta} : \Lambda_L \to \R$ be two external fields. Then there exists a constant $C \in (0 , \infty)$ depending only on $d$ and the ratios $c_+/c_-$ and $\lambda/c_-$ such that, for any $y \in \Lambda_L$,
  \begin{equation} \label{eq:15131709}
      \left| \left\langle \phi(y) \right\rangle_{\mu_{\Lambda_L}^\eta} - \left\langle \phi(y)  \right\rangle_{\mu_{\Lambda_L}^{ {\bar \eta}}} \right| \leq
   \left\{ \begin{aligned}
    & C L  \sum_{x \in \Lambda_L} \left| \eta(x) - {\bar \eta}(x) \right| &d = 1, \\
    &C \sum_{x \in \Lambda_L} \ln \left( \frac{C L}{1 \vee |x - y|} \right) \left| \eta(x) - {\bar \eta}(x) \right| &d = 2, \\
    & C \sum_{x \in \Lambda_L} \frac{\left| \eta(x) - {\bar \eta}(x) \right|}{1 \vee |x - y|^{d- 2}} &d \geq 3.
    \end{aligned} \right.
  \end{equation}
Moreover, there exist two constants $c , \tilde c  \in (0 , \infty)$ depending on $d$, $c_+/c_-$ and $\lambda/c_-$ such that, for any $y \in \Lambda_{L / 2}$, if the fields $\eta$ and $\bar \eta$ are such that there exists a vertex $x \in y + \Lambda_{\tilde c L}$ such that $\eta = {\bar \eta}$ on $\Lambda_L \setminus \{ x \}$ and $\eta(x) \neq \bar \eta(x)$, then
   \begin{equation} \label{eq:1514170999}
      \frac{\left\langle \phi(y) \right\rangle_{\mu_{\Lambda_L}^\eta} - \left\langle \phi(y)  \right\rangle_{\mu_{\Lambda_L}^{ {\bar \eta}}}}{\eta(x) - {\bar \eta}(x)} \geq
   \left\{ \begin{aligned}
    &c L &d = 1, \\
    &c  \ln \left( \frac{L}{1 \vee |x-y|} \right)  &d = 2, \\
    & \frac{c}{1 \vee |x-y|^{d- 2}} &d \geq 3.
    \end{aligned} \right.
  \end{equation}
\end{theorem}

\begin{remark}
We highlight that the left-hand side of~\eqref{eq:1514170999} does not involve an absolute value. This way of stating the result is strictly stronger than with absolute values as it implies that the sign of the difference $\left\langle \phi(y) \right\rangle_{\mu_{\Lambda_L}^\eta} - \left\langle \phi(y)  \right\rangle_{\mu_{\Lambda_L}^{ {\bar \eta}}}$ is the same as the one of $\eta(x) - \bar \eta(x)$, i.e., that the value of the thermal expectation $\left\langle \phi(y) \right\rangle_{\mu_{\Lambda_L}^\eta}$ increases with the value of $\eta(x)$.

It will be clear from the proof that this monotonicity is a direct consequence of the non-negativity of the heat kernel $P_\a$ (defined under the environment~\eqref{eq:15141308} below) and is an important ingredient in the proof of the lower bounds of Theorem~\ref{prop3.1009100bis} in Section~\ref{subsection4.3}.
\end{remark}

The proof of this theorem relies on coupling two different Langevin dynamics using the same Brownian motion. This argument was originally used by Funaki and Spohn in~\cite{FS}.

\begin{proof}
Let us fix an integer $L\geq 2$, two external fields $\eta , {\bar \eta} : \Lambda_L \to \R$, consider a collection of independent Brownian motions $\left\{ B_t(x) \, : \, x \in \Lambda_L \right\}$, and run two Langevin dynamics (with the same Brownian motions). Explicitly, we set
\begin{equation} \label{eq:11060408}
    \left\{ \begin{aligned}
    \di \phi_t(y) &= \sum_{e \ni y} V' \left( \nabla \phi_t ( e) \right) \di t + \lambda \eta(y) \di t + \sqrt{2} \di B_t(y) &&(t , y) \in  (0 , \infty) \times \Lambda_L, \\
    \phi_0(y)  &= 0 &&y \in \Lambda_L,\\
    \phi_t( y) &= 0 &&(t , y) \in (0 , \infty) \times \partial \Lambda_L,
    \end{aligned} \right.
\end{equation}
and
\begin{equation} \label{eq:11070408bis}
    \left\{ \begin{aligned}
    \di \bar \phi_t(y) &=  \sum_{e \ni y} V' \left( \nabla \bar \phi_t (e) \right) \di t + \lambda {\bar \eta}(y) \di t + \sqrt{2} \di B_t(y) &&(t , y) \in  (0 , \infty) \times \Lambda_L, \\
    \bar \phi_0(y)  &= 0 &&y \in \Lambda_L,\\
    \bar \phi_t ( y) &= 0 &&(t , y) \in (0 , \infty) \times \partial \Lambda_L.
    \end{aligned} \right.
\end{equation}
By ergodicity of the Langevin dynamics stated in Subsection~\ref{sectionLangevin}, the two dynamics convergence to the appropriate thermal averages
\begin{equation} \label{eq:13420408}
    \left\langle \phi_t(0) \right\rangle \underset{t \to \infty}{\longrightarrow} \left\langle \phi(0) \right\rangle_{\mu_{\Lambda_L}^\eta}  \hspace{5mm} \mbox{and} \hspace{5mm} \left\langle \bar \phi_t(0) \right\rangle \underset{t \to \infty}{\longrightarrow} \left\langle \phi(0) \right\rangle_{\mu_{\Lambda_L}^{{\bar \eta}}}.
\end{equation}
Taking the difference between the two stochastic differential equations~\eqref{eq:11060408} and~\eqref{eq:11070408bis}, and using that the two driving Brownian motions are the same, we obtain
\begin{equation} \label{eq:11500408}
    \left\{ \begin{aligned}
    \partial_t \left( \phi_t(y) - \bar \phi_t(y) \right)   &= \sum_{e \ni y} \left( V' \left( \nabla \phi_t (e) \right)  - V' \left( \nabla \bar \phi_t (e) \right) \right) - \lambda \left(  \eta(y) - \bar \eta (y)\right) &&(t , y) \in  (0 , \infty) \times \Lambda_L, \\
    \phi_0(y) - \bar \phi_0(y)  &= 0 &&y \in \Lambda_L,\\
    \phi_t(y) - \bar \phi_t(y) &= 0 &&(t , y) \in (0 , \infty) \times \partial \Lambda_L.
    \end{aligned} \right.
\end{equation}
The strategy is then to rewrite the equation~\eqref{eq:11500408} as a discrete linear parabolic equation. To this end, we note that, since the potential $V$ is assumed to be twice continuously differentiable, the following identity is valid, for any edge $e \in E \left( \Lambda_L\right)$ and any time $t \geq 0$,
\begin{equation*}
    V' \left( \nabla \phi_t (e) \right)  - V' \left( \nabla \bar \phi_t (e) \right) = \left(\int_0^1 V'' \left( s \nabla \phi_t(e) + (1 - s) \nabla  \bar \phi_t(e) \right) \, \di s \right) \times \left( \nabla \phi_t (e) - \nabla \bar \phi_t (e) \right).
\end{equation*}
Thus, if we let $\a : [0 , \infty) \times E\left( \Lambda_L^+ \right) \to [0, \infty)$ be the time-dependent environment defined by the formula
\begin{equation} \label{eq:15141308}
    \a (t  , e ) := \int_0^1 V'' \left( s \nabla \phi_t(e) + (1 - s) \nabla {\bar \phi_t}(e) \right) \, \di s,
\end{equation}
then we have the identity, for any point $y \in \Lambda_L$,
\begin{equation*}
    \sum_{e \ni y} \left( V' \left( \nabla \phi_t (e) \right)  - V' \left( \nabla {\bar \phi_t} (e) \right) \right) = \nabla \cdot \a \nabla \left( \phi_t - {\bar \phi_t} \right)(y),
\end{equation*}
where we used the notation~\eqref{eq:11330408} for discrete elliptic operators. The inequalities $c_- \leq V'' \leq c_+$ imply that the environment $\a$ is uniformly elliptic and satisfies $c_- \leq \a(t,y)\leq c_+$ for any pair $(t,e) \in [0, \infty) \times E \left( \Lambda_L \right)$.
Denoting by $w_t (z) := \phi_t(z) - {\bar \phi_t}(z)$, we obtain that the map $w$ solves the discrete parabolic equation
\begin{equation} \label{eq:11100508}
    \left\{ \begin{aligned}
    \partial_t w_t ( z) &= \nabla \cdot \a \nabla w_t(z) + \left( \eta(z) - \bar  \eta(z) \right) &&(t , z) \in  (0 , \infty) \times \Lambda_L, \\
    w_t(z) &= 0 &&(t , z) \in (0 , \infty) \times \partial \Lambda_L, \\
    w_0(z) &= 0 &&z \in  \Lambda_L.
    \end{aligned} \right.
\end{equation}
Using the notations introduced in Section~\ref{sectionheatkernel}, we denote by $P_\a$ the heat kernel associated with the environment $\a$. We next claim that one has the identity
\begin{equation} \label{eq:11180508}
    w_t (y) = \lambda \sum_{x \in \Lambda_L} \left( \eta(x) - \bar \eta(x) \right) \int_{0}^t P_\a \left( t , y; s , x \right) \, ds.
\end{equation}
The identity~\eqref{eq:11180508} is a consequence of Duhamel's formula (see~\cite[Theorem 2 page 50]{evans2010partial} where the formula is proved for the continuous heat equation in $\Rd$) applied to the equation~\eqref{eq:11100508}. In the present setting, it can be verified directly by computing the time derivative of the function $(t , y) \mapsto \lambda \sum_{x \in \Lambda_L} \left( \eta(x) - \bar \eta(x) \right) \int_{0}^t P_\a \left( t , y; s , x \right) \, ds$ and observing that this map is a solution of the equation~\eqref{eq:11100508}.

By the Nash--Aronson estimate on the heat-kernel $P_\a$ stated in Proposition~\ref{prop.NashAronson}, we obtain the inequality, for any time $t \geq 0$,
\begin{align*}
\int_{0}^t P_\a \left( t , y; s , x \right) \, ds & \leq \int_0^{t}  \frac{C_0}{1 \vee \left( c_-s \right)^{\frac d2}} \exp \left( - \frac{c_0|x - y|}{1 \vee \left( c_- s \right)^{\frac 12}} \right) \exp \left( - \frac{c_0c_- s}{L^2} \right) \, ds \\
& \leq  \left\{ \begin{aligned}
    &\frac{CL}{c_-}  &d = 1, \\
    &\frac{C}{c_-} \ln \left( \frac{C L}{1 \vee |x-y|} \right) &d = 2, \\
    &\frac{C}{c_-}\frac{1}{1 \vee |x -y|^{d- 2}} &d \geq 3,
    \end{aligned} \right.
\end{align*}
where the constant $C$ depends on $d$ and the ratio of ellipticity $c_+ / c_-$.
Taking the expectation with respect to the collection of Brownian motions $\left\{  B_t(y) \, : \, y \in \Lambda_L \right\}$ and using the identities~\eqref{eq:13420408}, we obtain
\begin{align} \label{eq:13500408}
    \left| \left\langle \phi(y) \right\rangle_{\mu_{\Lambda_L}^\eta} -  \left\langle \phi(y) \right\rangle_{\mu_{\Lambda_L}^{{\bar \eta}}} \right| & \leq \limsup_{t \to \infty} \left| \left\langle \phi_t(y) \right\rangle - \left\langle {\bar \phi_t}(y) \right\rangle \right| \\
    & \leq \limsup_{t \to \infty} \left| \left\langle  w(t , y) \right\rangle \right| \notag \\
    & \notag \leq
    \left\{ \begin{aligned}
    &C L \sum_{x \in \Lambda_L}\left| \eta(x) - {\bar \eta}(x) \right| &d = 1, \\
    &C \sum_{x \in \Lambda_L} \ln \left( \frac{C L}{1 \vee |x-y|} \right) \left| \eta(x) - {\bar \eta}(x) \right| &d = 2, \\
    &C  \sum_{x \in \Lambda_L} \frac{\left| \eta(x) - {\bar \eta}(x) \right|}{1 \vee |x-y|^{d- 2}} &d \geq 3,
    \end{aligned} \right.
\end{align}
where the constant $C$ depends on $d$ and the ratios $c_+ / c_-$ and $\lambda / c_-$.
This is the inequality~\eqref{eq:15131709}. It thus suffices to prove~\eqref{eq:1514170999}. Since the two external fields $\eta$ and ${\bar \eta}$ are equal everywhere except at the vertex $x$, the identity~\eqref{eq:11180508} becomes
\begin{equation*}
    w_t (y) := \lambda \left( \eta(x) - \bar \eta(x) \right) \int_{0}^t P_\a \left( t , y; s , x \right) \, ds.
\end{equation*}
Let us denote by $c_1$ the constant which appears in the statement of Proposition~\ref{prop.NashAronson} and set $\tilde c := c_1/2$. Assuming that the vertex $x$ belongs to the box $y+ \Lambda_{\tilde c L}$, we may apply the lower bound for the map $P_\a$ stated in~\eqref{Nash.infprop3.3} and use that the heat kernel is non-negative to obtain, for any time $t \geq c_1^2 L^2/c_-$,
    \begin{align*}
       \int_{0}^t P_\a \left( t , y; s , x \right) \, ds & \geq \int_{t  - c_1^2 L^2/c_-}^{t - |x-y|^2/c_-} P_\a \left( t , y; s , x \right) \, ds \\
        & \geq \int_{t  - c_1^2 L^2/c_-}^{t - |x-y|^2/c_-} \frac{c}{\left(c_-(t-s)\right)^{\frac d2} \vee 1} \, ds \\
        & \geq \frac{1}{c_-} \int_{|x-y|^2}^{c_1^2 L^2} \frac{c}{s^{\frac d2} \vee 1} \, ds.
    \end{align*}
The term in the right-hand side can be explicitly computed, and we obtain the lower bound, for any time $t \geq c_1^2 L^2/c_-$,
\begin{equation} \label{eq:13470508}
     \int_{0}^t P_\a \left( t , y; s , x \right) \, ds \geq
      \left\{ \begin{aligned}
    &\frac{c}{c_-} L &d = 1, \\
    &\frac{c}{c_-} \ln \left( \frac{L}{1 \vee |x-y|} \right) &d = 2, \\
    &\frac{c}{c_-}\frac{1}{1 \vee |x-y|^{d- 2}} &d \geq 3.
    \end{aligned} \right.
\end{equation}
A similar computation as in~\eqref{eq:13500408} completes the proof of the lower bound~\eqref{eq:1514170999}.
\end{proof}

\subsection{Real-valued random-field random surface models: Upper bounds} \label{section.section4.2}
In this section, we establish the upper bounds of the inequalities~\eqref{eq:TV0904041},~\eqref{eq:TV090404bis1} and~\eqref{eq:TV090404ter1} by combining the result of Theorem~\ref{Thm6} with the Efron--Stein concentration inequality stated in Proposition~\ref{prop.Efronstein}.

\begin{proof}[Proof of Theorem~\ref{prop3.1009100bis}: Upper bounds~\eqref{eq:TV0904041},~\eqref{eq:TV090404bis1} and~\eqref{eq:TV090404ter1}]
Let us fix an integer $L \ge 2$ and a vertex $y \in \Lambda_L$. The strategy is to apply the Efron--Stein inequality with the map $f : \eta \mapsto \left\langle \phi(y)\right\rangle_{\mu_{\Lambda_L}^{\eta}}$. Let us note that, by the inequality~\eqref{eq:15131709} applied with $\bar \eta = 0$ and the observation $\left\langle \phi(y)  \right\rangle_{\mu_{\Lambda_L}^{0}} = 0$, the $L^2$-norm of the mapping $\eta \mapsto \left\langle \phi(y)\right\rangle_{\mu_{\Lambda_L}^{\eta}}$ is finite. The assumption of Proposition~\ref{prop.Efronstein} is thus satisfied.

\smallskip

We let $\eta$ and $\tilde \eta$ be two independent copies of the random field. For each point $x \in \Lambda_L$, we denote by $\eta^x$ the resampled random field defined by the formula $\eta^x(z) = \eta(z)$ if $z \neq x$ and $\eta^x(x) = \tilde \eta (x)$. By Proposition~\ref{prop.Efronstein}, one has the estimate, for any $y \in \Lambda_L$,
\begin{equation} \label{eq:10380408}
    \var \left[ \left\langle \phi(y)\right\rangle_{\mu_{\Lambda_L}^{\eta}}  \right] \leq \frac 12 \sum_{x \in \Lambda_L} \E \left[ \left( \left\langle \phi(y)\right\rangle_{\mu_{\Lambda_L}^{\eta}} - \left\langle \phi(y)\right\rangle_{\mu_{\Lambda_L}^{\eta^x}}  \right)^2 \right].
\end{equation}
Applying the quenched upper bound~\eqref{eq:15131709} obtained in Theorem~\ref{Thm6} with the external fields $\eta := \eta$ and ${\bar \eta} := \eta^x$, we obtain the upper bound
\begin{equation*}
    \left( \left\langle \phi(y)\right\rangle_{\mu_{\Lambda_L}^{\eta}} - \left\langle \phi(y)\right\rangle_{\mu_{\Lambda_L}^{\eta^x}}  \right)^2 \leq \left\{ \begin{aligned}
    &C L^2 \left( \eta(x) - \tilde \eta(x) \right)^2 &d = 1, \\
    &C \left(\ln  \frac{L}{1 \vee |x-y|} \right)^2  \left( \eta(x) - \tilde \eta(x) \right)^2  &d = 2, \\
    &C \frac{  \left( \eta(x) - \tilde \eta(x) \right)^2}{1 \vee |x - y|^{2d- 4}} &d \geq 3.
    \end{aligned} \right.
\end{equation*}
Taking the expectation in the previous display (with respect to the random field $\eta$) and using that the random variables $\eta(x)$ and $\tilde \eta(x)$ have expectation $0$ and variance $1$ shows the inequality, for any $x \in \Lambda_L$,
\begin{equation} \label{eq:10302009}
    \E \left[ \left( \left\langle \phi(y)\right\rangle_{\mu_{\Lambda_L}^{\eta}} - \left\langle \phi(y)\right\rangle_{\mu_{\Lambda_L}^{\eta^x}}  \right)^2 \right] \leq \left\{ \begin{aligned}
    &C L^2 &d = 1, \\
    &C  \left(\ln \frac{CL}{1 \vee |x-y|} \right)^2  &d = 2, \\
    &C\frac{1}{ 1 \vee |x-y|^{2d- 4}} &d \geq 3.
    \end{aligned} \right.
\end{equation}
Summing the inequality~\eqref{eq:10302009} over all the points $x \in \Lambda_L$ and using the estimate~\eqref{eq:10380408} completes the proof of the upper bounds of~\eqref{eq:TV0904041},~\eqref{eq:TV090404bis1} and~\eqref{eq:TV090404ter1}.
\end{proof}

\subsection{Real-valued random-field random surface models: Lower bounds} \label{subsection4.3}
The objective of this section is to prove the lower bounds for the variance (in the random field) of the thermal average of the heights.

\begin{proof}[Proof of Theorem~\ref{prop3.1009100bis}: Lower bounds~\eqref{eq:TV0904041low},~\eqref{eq:TV090404bis1low} and~\eqref{eq:TV090404ter1low}]
We fix an integer $L \ge 2$, a vertex $y \in \Lambda_{L/2}$ and let $\tilde c$ be the constant which appears in the statement of Theorem~\ref{Thm6}. We may assume without loss of generality that $\tilde c \leq 1/2$. We let $x_1 , \ldots x_{(2L+1)^d}$ be an enumeration of the points of the box $\Lambda_{L}$ such that the collection $x_0 , \ldots, x_{\left(2 \lfloor \tilde c L \rfloor +1\right)^d}$ is an enumeration of the vertices of the box $y + \Lambda_{ \tilde c L}$. We let $\eta$ and $\tilde \eta$ be two independent copies of the random field. For each integer $k \in \{ 1 , \ldots, (2L+1)^d \}$, we introduce the notations $\eta_k := \eta(x_k)$, $\tilde \eta_k := \tilde \eta(x_k)$, denote by $\nu_k$ the law of $\eta_k$, and let $\mathcal{F}_k$ be the sigma-algebra generated by the random variables $\eta_1 , \ldots, \eta_k$. We additionally denote by $\eta^k$ the random field satisfying $\eta^k(x) = \eta(x)$ if $x \neq x_k$, and $\eta^k(x_k) = \tilde \eta(x_k)$.

For each integer $k \in \{1 , \ldots, (2L+1)^d \}$, we introduce two random variables. We let $X_k$ be the conditional expectation of $\left\langle \phi(y)\right\rangle_{\mu_{\Lambda_L}^{\eta}}$ given the $\eta_1,\dots \eta_k$ --- that is,
\[
X_k = \E \left[ \left\langle \phi(y)\right\rangle_{\mu_{\Lambda_L}^{\eta}}  | \mathcal{F}_{k} \right].
\]
The standard Pythagorean identity for martingales tells us that
\begin{equation}\label{eq:26051624}
\var \left[ \left\langle \phi(y)\right\rangle_{\mu_{\Lambda_L}^{\eta}}  \right] = \sum_{k = 1}^{(2L+1)^d}  \E \left[ \left( X_{k} - X_{k-1} \right)^2 \right].
\end{equation}
Next, we define $\tilde{X}_k$ as the conditional expectation of $\left\langle \phi(y)\right\rangle_{\mu_{\Lambda_L}^{\eta}}$ given $\eta_1,\dots \eta_{k-1}$ \emph{and} where the external field at $x_k$ is set to $\tilde{\eta}_k$. Since $\tilde{\eta}_k$ is independent of $\eta_k$, we find that
\[
X_{k-1} = \E\left[\tilde{X}_k | \mathcal{F}_k\right],
\]
and therefore 
\begin{align} \label{eq:16542009}
    \E \left[ \left( X_k - \tilde X_k \right)^2 \right] & = \E \left[X_k^2 + \tilde X_k^2 \right] - 2 \E \left[ X_k \tilde X_k \right] \\
    & = 2 \E \left[X_k^2 \right] - 2\E \left[ X_k  X_{k-1} \right] \notag \\
    & = 2 \E \left[\left( X_k - X_{k-1} \right)^2 \right]. \notag
\end{align}
Combining this with \eqref{eq:26051624} allows us to conclude that 
\begin{equation} \label{eq:122121}
    \var \left[ \left\langle \phi(y)\right\rangle_{\mu_{\Lambda_L}^{\eta}}  \right]  = \frac 12 \sum_{k = 1}^{(2L+1)^d} \E \left[ \left( X_{k} - \tilde X_k \right)^2 \right].
\end{equation}
We next prove the lower bound, for any integer $k \in \left\{ 1 ,  \ldots, \left(2 \lfloor \tilde c L \rfloor + 1 \right)^d \right\}$,
\begin{equation} \label{eq:10032107}
    \E \left[ \left( X_k - \tilde X_k \right)^2 \right] \geq \left\{ \begin{aligned}
     & cL^2 &d = 1, \\
     & c\left(\ln \frac{L}{1 \vee |x_k - y|} \right)^2 &d = 2, \\
     &\frac{ c}{ 1 \vee |x_k - y|^{2d-4}} &d \geq 3.
    \end{aligned} \right.
\end{equation}
Using that all the terms in the right side of~\eqref{eq:122121} are non-negative, that the collection of points $x_1 , \ldots x_{2^d \lfloor \tilde c L \rfloor^d \vee 1}$ is an enumeration of the points of the box $y +\Lambda_{\tilde c L}$ and the estimate~\eqref{eq:10032107} completes the proof of the lower bounds~\eqref{eq:TV0904041},~\eqref{eq:TV090404bis1} and~\eqref{eq:TV090404ter1}.

There remains to prove~\eqref{eq:10032107}. Let us note that the two fields $\eta$ and $\eta^k$ are equal on the set $\Lambda_L \setminus \{ x_k \}$. We can thus apply Theorem~\ref{Thm6} to obtain, on the event $\left\{ \eta_k \neq \tilde \eta_k \right\}$,
\begin{equation} \label{eq:11020508}
      \frac{\left\langle \phi(y)\right\rangle_{\mu_{\Lambda_L}^{\eta}} - \left\langle \phi(y)\right\rangle_{\mu_{\Lambda_L}^{\eta^{k}}}}{\eta_k - \tilde \eta_k}  \geq \left\{ \begin{aligned}
       &cL &d = 1, \\
       &c\ln \left( \frac{L}{1 \vee |x_k-y|} \right) &d = 2, \\
       &\frac{
       c}{\left( 1 \vee |x_k-y|^{d-2}\right)} &d \geq 3.
      \end{aligned}
      \right.
\end{equation}
Integrating the inequality~\eqref{eq:11020508} over the field variables $\eta_{ k+1}, \ldots, \eta_{L^d}$, we obtain, on the event $\{ \eta_k \neq \tilde \eta_k \}$,
\begin{equation} \label{eq:23020508}
    \frac{X_k - \tilde X_k}{\eta_k - \tilde \eta_k} = \int  \frac{\left\langle \phi(y)\right\rangle_{\mu_{\Lambda_L}^{\eta}} - \left\langle \phi(y)\right\rangle_{\mu_{\Lambda_L}^{\eta^{k}}}}{\eta_k - \tilde \eta_k}  \prod_{j > k} \nu_j (d \eta_j)  \geq \left\{ \begin{aligned}
       &cL  &d = 1, \\
       &c \ln \left( \frac{L}{1 \vee |x_k-y|} \right) &d = 2, \\
       &\frac{c}{1 \vee |x_k-y|^{d-2}} &d \geq 3.
      \end{aligned}
      \right.
\end{equation}
The lower bound~\eqref{eq:23020508} shows, for any realization of the random fields $\eta$ and $\tilde \eta$,
\begin{equation*}
    \left( X_k - \tilde X_k \right)^2 \geq   \left\{ \begin{aligned}
    &cL^2\left( \eta_k - \tilde \eta_k \right)^2 &d =1, \\
    & c \left( \eta_k - \tilde \eta_k \right)^2 \left( \ln  \frac{L}{ 1 \vee |x_k-y|} \right)^2 &d = 2, \\
    & \frac{c\left( \eta_k - \tilde \eta_k \right)^2}{1 \vee |x_k-y|^{2d-4}}  &d \geq 3.
    \end{aligned} \right.
\end{equation*}
Taking the expectation with respect to the external field, using that the two random variables $\eta_k$ and $\tilde \eta_k$ are independent and that their variance is equal to $1$ completes the proof of~\eqref{eq:10032107}.
\end{proof}

\begin{remark}\label{rem:anti-concentration}
Theorem~\ref{prop3.1009100bis} is proved for independent $(\eta(x))$ with the sole assumption that $\var(\eta(x))=1$ for all $x\in\Lambda_L$. We mention that, when further assumptions are imposed on $\eta$, a different approach to the variance lower bounds of the heights of the random surface is available. We present below a simple statement of this kind in the case that each $\eta(x)$ has a standard Gaussian distribution which yields, in addition to the variance lower bound, the following estimate: There exists $\ep > 0$ (depending on $d$, $c_+/c_-$ and $\lambda/c_-$) such that for each $y\in\Lambda_{L/2}$ and each interval $I \subseteq \R$ whose side length is equal to $L^{2-d/2}$ in dimensions $d = 1 , 2 , 3$, $\sqrt{\ln L}$ in dimension $d = 4$ and $1$ in dimensions $d \geq 5$, one has the estimate
\begin{equation} \label{eq:18271902}
    \P \left( \left\langle \phi(y) \right\rangle_{\mu_{\Lambda_L}^\eta} \in \ep I \right) \leq \frac{2}{3}.
\end{equation}
Estimates of this form are often referred to as anti-concentration bounds, and we refer to~\cite[Section 1.1]{MP15} for related arguments which can yield more detailed anti-concentration bounds.

Let $X$ be a standard Gaussian vector in $\R^n$. We first recall a useful bound on the total variation distance,
\begin{equation}\label{eq:total variation estimate}
  d_{\text{TV}}(\eta, \eta + t) \le C\|t\|_2
\end{equation}
for each deterministic vector $t\in\R^n$, where we write $d_{\text{TV}}(X,Y)$ for the total variation distance between the distributions of the random vectors $X$ and $Y$.

Now suppose $f:\R^n\to\R$ satisfies that for each $1\le x\le n$ there exists $\ell(x)>0$ so that
\begin{equation}\label{eq:lower Lipschitz condition}
  f(\alpha + r \delta_x) - f(\alpha) \ge \ell(x) r\quad\text{for all $\alpha\in\R^n$ and $r>0$},
\end{equation}
where $\delta_x\in\R^n$ is the vector satisfying $\delta_x(y) = 0$ for $y\ne x$ and $\delta_x(x)=1$. Such a function necessarily satisfies the following: There exists an $\varepsilon>0$ such that for each interval $I\subset\R$ satisfying $|I|\le \varepsilon\|\ell\|_2$ it holds that
\begin{equation}\label{eq:anti-concentration bound}
  \P(f(\eta)\in I)\le \frac{2}{3}.
\end{equation}
This bound clearly implies that $\var(f(\eta))\ge c\|\ell\|_2^2$. To see~\eqref{eq:anti-concentration bound} set $t:=\varepsilon\frac{\ell}{\|\ell\|_2}$ and observe that, by~\eqref{eq:lower Lipschitz condition}, it is impossible that for some $\alpha$, $f(\alpha)\in I$ and $f(\alpha+t)\in I$ simultaneously. Consequently, by~\eqref{eq:total variation estimate},
\begin{equation}
\begin{split}
  \P(f(\eta)\in I) &\le \frac{1}{2}(\P(f(\eta)\in I) + \P(f(\eta + t)\in I) + d_{\text{TV}}(\eta, \eta + t))\\
  &= \frac{1}{2}\P(\{f(\eta)\in I\}\cup\{f(\eta + t)\in I\}) + \frac{1}{2}d_{\text{TV}}(\eta, \eta + t)\\
  &\le \frac{1}{2} + \frac{1}{2}C\varepsilon
\end{split}
\end{equation}
which implies~\eqref{eq:anti-concentration bound} by choosing $\varepsilon$ sufficiently small.

Since for each $y\in\Lambda_{L/2}$ the function $f(\eta) = \left\langle \phi(y) \right\rangle_{\mu_{\Lambda_L}^\eta}$ satisfies bounds of the form~\eqref{eq:lower Lipschitz condition} by~\eqref{eq:1514170999}, the variance lower bounds of Theorem~\ref{prop3.1009100bis} follow from the above argument.
\end{remark}

\section{Integer-valued random-field Gaussian free field} \label{section66}

In this section, we study the integer-valued random-field Gaussian free field. We emphasize that, unlike our discussion of real-valued random surfaces with random fields, we assume here that the potential is $V(x) = \tfrac12 x^2$. We further assume that $\eta$ is independent and satisfies $\E[\eta(x)]=0$ and $\var[\eta(x)] =1$.

\subsection{Upper bounds on gradient and height fluctuations} \label{subsec3.2}

The upper bounds of Theorem~\ref{prop3.100910disc} and Theorem~\ref{prop3.1009100bisdisc} will follow from the following lemma which compares the gradient of the integer-valued random-field Gaussian free field with the gradient of $u_{\Lambda,\eta}$, the ground state of the real-valued random-field Gaussian free field, at any fixed external field $\eta$.

\begin{lemma} \label{Thm7DGFF}
  Let $d\ge 1$, $\Lambda\subset\Z^d$ be a finite subset of $\Zd$, $\lambda>0$ be the strength of the random field, $\beta > 0$ be an inverse temperature, $\eta:\Lambda\to\R$ be an external field and $u_{\Lambda, \eta}:\Lambda^+\to\R$ be the solution of~\eqref{def.ueta}. Suppose $\phi$ is sampled from the probability distribution $\mu^{\IV , \beta, \eta}_{\Lambda}$. Then there exists a constant $C$ depending only on the dimension $d$ such that
  \begin{equation}
  \label{eq:Thm7DGFF}
    \left\langle \exp \left(  \frac{\beta}{8} \left\| \nabla \phi - \lambda \nabla u_{\Lambda, \eta} \right\|_{L^2  \left( \Lambda^+ \right)}^2 \right)\right\rangle_{\mu_\Lambda^{\IV, \beta, \eta}} \leq \exp \left( C (1+\beta) \left| \Lambda^+ \right| \right).
  \end{equation}
\end{lemma}

\begin{proof}[Proof of Lemma~\ref{Thm7DGFF}]
We first treat the specific case $\eta = 0$ and prove the following slightly stronger version of the inequality~\eqref{eq:Thm7DGFF}: There exists a constant $C$ depending on $d$ such that
\begin{equation} \label{eq:12012309}
    \left\langle \exp \left(  \frac{\beta}{4} \left\| \nabla \phi \right\|_{L^2  \left( \Lambda^+ \right)}^2 \right)\right\rangle_{\mu_\Lambda^{\IV, \beta, 0}} \leq \exp \left( C (1 + \beta) \left| \Lambda^+ \right| \right).
\end{equation}
To prove the estimate~\eqref{eq:12012309}, we first note that, for each $\phi : \Lambda^+ \to \Z$ normalized to be $0$ on the boundary $\partial \Lambda$, the cardinality of the set $\left\{ \psi : \Lambda^+ \to \Z \, : \, \psi \equiv 0 ~ \mbox{on} ~ \partial \Lambda ~ \mbox{and} ~ \lfloor \psi/2 \rfloor = \phi \right\}$ is smaller than $2^{\left| \Lambda \right|}$. From this observation, one deduces the inequality
\begin{align} \label{eq:09031309}
    \sum_{\psi : \Lambda^+ \to \Z} \exp \left( -  \frac{\beta}{2} \left\| \nabla \left\lfloor \frac{\psi}{2} \right\rfloor \right\|_{L^2 \left( \Lambda^+ \right)}^2 \right) &=
    \sum_{\phi : \Lambda^+ \to \Z} \sum_{\psi : \, \left\lfloor \frac{\psi}{2} \right\rfloor = \phi} \exp \left( - \frac{\beta}{2} \left\| \nabla \phi \right\|_{L^2 \left( \Lambda^+ \right)}^2 \right) \\
    & \leq 2^{\left| \Lambda \right|}  \sum_{\phi : \Lambda^+ \to \Z} \exp \left( - \frac{\beta}{2} \left\| \nabla \phi \right\|_{L^2 \left( \Lambda^+ \right)}^2 \right). \notag
\end{align}
Using that the difference between a real number and its floor is smaller than $1$, and the inequality $(a + b)^2 \leq 2a^2 + 2b^2$, we obtain
\begin{equation} \label{eq:09041309}
     \left\| \nabla \left\lfloor \frac{\psi}{2} \right\rfloor \right\|_{L^2 \left( \Lambda^+ \right)}^2 = \left\|  \frac{\nabla \psi}{2} + \left( \nabla \left\lfloor \frac{\psi}{2} \right\rfloor - \frac{\nabla\psi}{2} \right) \right\|_{L^2 \left( \Lambda^+ \right)}^2 \leq \frac12 \left\| \nabla \psi \right\|_{L^2 \left( \Lambda^+ \right)}^2 + C\left| \Lambda^+ \right|.
\end{equation}
Bringing a factor of $1/2$ `into' the $L^2$ norm and then applying inequalities~\eqref{eq:09031309} and~\eqref{eq:09041309}, we can derive the bound 
\begin{align} \label{eq:09091309}
    \sum_{\psi : \Lambda^+ \to \Z} \exp \left( -  \frac{\beta}{4} \left\| \nabla \psi \right\|_{L^2 \left( \Lambda^+ \right)}^2 \right) & \leq e^{C \beta \left| \Lambda_+\right|}\sum_{\psi : \Lambda^+ \to \Z} \exp \left( -  \frac{\beta}{2} \left\| \nabla \left\lfloor \frac{\psi}{2} \right\rfloor \right\|_{L^2 \left( \Lambda^+ \right)}^2 \right) \\
    & \leq  2^{\left| \Lambda\right|} \times  e^{C \beta \left|\Lambda^+ \right|} \sum_{\phi : \Lambda^+ \to \Z} \exp \left( - \frac{\beta}{2} \left\| \nabla \phi \right\|_{L^2 \left( \Lambda^+ \right)}^2 \right). \notag
\end{align}
We then expand the thermal expectation as a normalized sum, writing
\begin{align} \label{eq:09331309}
    \left\langle \exp \left( \frac \beta4  \left\| \nabla \phi \right\|_{L^2 \left( \Lambda^+ \right)}^2 \right) \right\rangle_{\mu^{\IV ,\beta, \eta}_L} & = \frac{1}{Z^{\IV, \beta, 0}_{\Lambda}} \sum_{\phi : \Lambda^+ \to \Z} \exp \left( \frac \beta4  \left\| \nabla \phi \right\|_{L^2 \left( \Lambda^+ \right)}^2 \right) \times \exp \left( -
    \frac \beta2 \left\| \nabla \phi \right\|_{L^2 \left( \Lambda^+ \right)}^2 \right)\\
    & = \frac{1}{Z^{\IV, \beta,0}_{\Lambda}} \sum_{\phi :  \Lambda^+ \to \Z} \exp \left( - \frac \beta4 \left\| \nabla \phi \right\|_{L^2 \left( \Lambda^+ \right)}^2 \right) \notag.
\end{align}
Using the explicit formula for the partition function $Z^{\IV, \beta, 0}_{\Lambda} := \sum_{\phi : \Lambda^+ \to \Z} \exp \left( - \frac{\beta}{2} \left\| \nabla \phi \right\|_{L^2 \left( \Lambda^+ \right)}^2 \right)$ and combining the inequality~\eqref{eq:09091309} with the identity ~\eqref{eq:09331309} and the inequality $\left| \Lambda \right| \leq \left| \Lambda^+ \right|$ completes the proof of~\eqref{eq:12012309}.

We now prove the inequality~\eqref{eq:Thm7DGFF} for a general external field $\eta : \Lambda \to \R$. Applying the estimate~\eqref{eq:12012309} at inverse temperature $\frac 54 \beta$, we see that it is sufficient to prove that there exists a constant $C > 0$ depending on $d$ such that
\begin{equation} \label{eq:14132309}
    \left\langle \exp \left(  \frac{\beta}{8} \left\| \nabla \phi - \lambda \nabla u_{\Lambda, \eta} \right\|_{L^2  \left( \Lambda^+ \right)}^2 \right)\right\rangle_{\mu_\Lambda^{\IV, \beta, \eta}} \leq \exp \left( C \beta \left| \Lambda^+ \right| \right) \left\langle \exp \left(  \frac{5\beta}{16} \left\| \nabla \phi \right\|_{L^2  \left( \Lambda^+ \right)}^2 \right)\right\rangle_{\mu_\Lambda^{\IV, \frac {5 \beta}4, 0}}.
\end{equation}
The rest of the argument is thus devoted to the proof of the inequality~\eqref{eq:14132309}. Using the identity $-\Delta u_{\Lambda,\eta} = \eta$ and performing a discrete integration by parts, one obtains, for any integer-valued surface $\phi : \Lambda^+ \to \Z$ normalized to be $0$ on $\partial \Lambda$,
\begin{align} \label{eq:10162309}
    \frac 12 \left\| \nabla \phi \right\|^2_{L^2 \left( \Lambda^+ \right)} - \lambda \sum_{x \in \Lambda} \phi(x) \eta(x)
    & = \frac 12 \left\| \nabla \phi \right\|^2_{L^2 \left( \Lambda^+ \right)} - \lambda \sum_{e \in E \left(\Lambda^+\right)} \nabla \phi(e) \nabla u_{\Lambda, \eta}(e) \\
    & = \frac 12 \left\| \nabla \phi - \lambda \nabla u_{\Lambda, \eta} \right\|^2_{L^2 \left( \Lambda^+ \right)}  - \frac 12 \left\| \lambda \nabla u_{\Lambda, \eta} \right\|^2_{L^2 \left( \Lambda^+ \right)}. \notag
\end{align}
Using the previous computation, we see that
\begin{align} \label{eq:10042409}
   \lefteqn{\left\langle \exp \left(  \frac{\beta}{8} \left\| \nabla \phi - \lambda \nabla u_{\Lambda, \eta} \right\|_{L^2  \left( \Lambda^+ \right)}^2 \right)\right\rangle_{\mu_\Lambda^{\IV, \beta,  \eta}}} \qquad & \\ & =  \frac1{{Z^{\IV, \beta, \eta}_{\Lambda}}}\sum_{\phi : \Lambda^+ \to \Z} \exp \left( \frac \beta8 \left\| \nabla \phi - \lambda \nabla u_{ \Lambda, \eta} \right\|_{L^2 \left( \Lambda^+ \right)}^2\right) \exp \left( - \frac \beta2 \left\| \nabla \phi - \lambda \nabla u_{ \Lambda, \eta} \right\|_{L^2 \left( \Lambda^+ \right)}^2  +  \frac \beta2  \left\| \lambda \nabla u_{ \Lambda, \eta} \right\|_{L^2 \left( \Lambda^+ \right)}^2\right) \notag \\
   & = \frac{1}{Z^{\IV, \beta, \eta}_{\Lambda} \times \exp \left( - \frac \beta2  \left\| \lambda \nabla u_{ \Lambda, \eta} \right\|_{L^2 \left( \Lambda_L \right)}^2 \right)} \sum_{\phi : \Lambda^+ \to \Z} \exp \left( - \frac{3\beta}{8} \left\| \nabla \phi - \lambda \nabla u_{ \Lambda, \eta} \right\|_{L^2 \left( \Lambda^+ \right)}^2 \right). \notag
\end{align}
Additionally, from the definition of the partition function $Z^{\IV, \beta, \eta}_{\Lambda_L}$ and the computation~\eqref{eq:10162309}, we obtain the identity
\begin{align} \label{eq:10082409}
    Z^{\IV, \beta, \eta}_{\Lambda} \times \exp \left( - \frac\beta2  \left\|\lambda \nabla u_\eta \right\|_{L^2 \left( \Lambda \right)}^2 \right) & = \sum_{\phi : \Lambda \to \Z} \exp \left( - \frac{\beta}2 \left\| \nabla \phi \right\|^2_{L^2 \left( \Lambda^+ \right)} + \beta \lambda \sum_{x \in \Lambda} \phi(x) \eta(x) -  \frac\beta2  \left\| \lambda \nabla u_{\Lambda,\eta} \right\|_{L^2 \left( \Lambda \right)}^2 \right) \\
    & = \sum_{\phi : \Lambda^+ \to \Z} \exp \left( - \frac\beta2 \left\| \nabla \phi -\lambda \nabla u_{ \Lambda, \eta} \right\|_{L^2 \left( \Lambda^+ \right)}^2 \right).  \notag
\end{align}
Let us then introduce the mapping $w_{\Lambda, \eta} := \lambda u_{\Lambda, \eta} - \lfloor \lambda u_{\Lambda, \eta} \rfloor$. Since the function $\lfloor \lambda u_{\Lambda, \eta} \rfloor$ is integer-valued, we may perform the discrete change of variable $\phi \mapsto \phi - \lfloor \lambda u_{\Lambda, \eta} \rfloor$, and write
\begin{equation} \label{eq:15072309}
    \sum_{\phi : \Lambda^+ \to \Z} \exp \left( - \frac{3\beta}{8} \left\| \nabla \phi -\lambda \nabla u_{ \Lambda, \eta} \right\|_{L^2 \left( \Lambda^+ \right)}^2 \right) = \sum_{\phi : \Lambda^+ \to \Z} \exp \left( - \frac{3\beta}{8} \left\| \nabla \phi -\nabla w_{ \Lambda, \eta} \right\|_{L^2 \left( \Lambda^+ \right)}^2 \right).
\end{equation}
By definition, the mapping $w_{ \Lambda, \eta}$ is bounded by $1$. Applying the inequality $(a + b)^2 \geq \frac 56 a^2 - 5 b^2$, we obtain
\begin{equation} \label{eq:15082309}
    \left\| \nabla \phi -\nabla w_{ \Lambda, \eta} \right\|_{L^2 \left( \Lambda^+ \right)}^2 \geq \frac56\left\| \nabla \phi \right\|_{L^2 \left( \Lambda^+ \right)}^2 - C \left|\Lambda^+ \right|.
\end{equation}
A combination of~\eqref{eq:15072309} and~\eqref{eq:15082309} implies
\begin{equation} \label{eq:10092409}
    \sum_{\phi : \Lambda^+ \to \Z} \exp \left( - \frac{3\beta}{8} \left\| \nabla \phi - \lambda \nabla u_{ \Lambda, \eta} \right\|_{L^2 \left( \Lambda^+ \right)}^2 \right) \leq \exp \left( C \beta \left| \Lambda^+\right| \right) \sum_{\phi : \Lambda^+ \to \Z} \exp \left( - \frac{5\beta}{16} \left\| \nabla \phi \right\|_{L^2 \left( \Lambda^+ \right)}^2 \right).
\end{equation}
A similar computation, but using this time the inequality $\left\| \nabla \phi -\nabla w_{ \Lambda, \eta} \right\|_{L^2 \left( \Lambda^+ \right)}^2 \leq \frac54\left\| \nabla \phi \right\|_{L^2 \left( \Lambda^+ \right)}^2 + C \left|\Lambda^+ \right|$, yields the lower bound
\begin{align} \label{eq:10102409}
    \sum_{\phi : \Lambda^+ \to \Z} \exp \left( - \frac\beta2 \left\| \nabla \phi - \lambda \nabla u_{ \Lambda, \eta} \right\|_{L^2 \left( \Lambda^+ \right)}^2 \right) & = \sum_{\phi : \Lambda^+ \to \Z} \exp \left( - \frac\beta2 \left\| \nabla \phi -\nabla w_{ \Lambda, \eta} \right\|_{L^2 \left( \Lambda^+ \right)}^2 \right) \\
    & \geq \exp \left( - C \beta \left|\Lambda^+ \right| \right) \sum_{\phi : \Lambda^+ \to \Z} \exp \left( - \frac{5\beta}8 \left\| \nabla \phi \right\|_{L^2 \left( \Lambda^+ \right)}^2 \right). \notag
\end{align}
A combination of the identities~\eqref{eq:10042409},~\eqref{eq:10082409} with the inequalities~\eqref{eq:10092409} and~\eqref{eq:10102409} implies the upper bound
\begin{align*}
    \left\langle \exp \left(  \frac{\beta}{8} \left\| \nabla \phi - \lambda \nabla u_{\Lambda, \eta} \right\|_{L^2  \left( \Lambda^+ \right)}^2 \right)\right\rangle_{\mu_\Lambda^{\IV, \beta,  \eta}} & \leq \exp \left( C \beta \left|\Lambda^+ \right| \right)  \frac{ \sum_{\phi : \Lambda^+ \to \Z} \exp \left( - \frac{5\beta}{16} \left\| \nabla \phi \right\|_{L^2 \left( \Lambda^+ \right)}^2 \right)}{\sum_{\phi : \Lambda^+ \to \Z} \exp \left( - \frac{5\beta}8 \left\| \nabla \phi \right\|_{L^2 \left( \Lambda^+ \right)}^2 \right)} \\
    & = \exp \left( C \beta \left|\Lambda^+ \right| \right) \left\langle \exp \left(  \frac{5\beta}{16} \left\| \nabla \phi \right\|_{L^2  \left( \Lambda^+ \right)}^2 \right)\right\rangle_{\mu_\Lambda^{\IV, \frac {5 \beta}4, 0}}.
\end{align*}
The proof of the inequality~\eqref{eq:14132309} is complete.

\end{proof}

\begin{proof}[Proof of the upper bounds of Theorem~\ref{prop3.100910disc} and Theorem~\ref{prop3.1009100bisdisc}]
Lemma~\ref{Thm7DGFF} and Jensen's inequality imply that 
\begin{equation}\label{eq:JensenComparisonBound}
\| \nabla \phi - \lambda \nabla u_{\Lambda_L,\eta} \|_{\underline{L}^2\left(\Lambda^+_L,\mu_{\Lambda_L}^{\IV,\beta,\eta}\right)} \leq C\left(1 + \beta^{-1}\right). 
\end{equation}
The upper bound of Theorem~\ref{prop3.100910disc} now follows from Proposition~\ref{prop.propueta} and the triangle inequality. We also obtain in this way a version of the lower bound of Theorem~\ref{prop3.100910disc} which is, however, suboptimal for small $\beta$. To obtain a lower bound which is temperature independent (as stated in Theorem~\ref{prop3.100910disc}), we provide a different argument in the following section.

For the upper bound of Theorem~\ref{prop3.1009100bisdisc}, one can apply the Poincar\'e inequality and obtain
\begin{equation} \label{eq:11132907}
    \left\| \phi - \lambda u_{\Lambda_L,\eta} \right\|_{\underline{L}^2 \left( \Lambda_L^+ , \mu^{\IV, \beta ,\eta}_{\Lambda_L} \right)} \leq C L \left\| \nabla \phi - \lambda \nabla u_{\Lambda_L,\eta} \right\|_{\underline{L}^2 \left( \Lambda_L^+ , \mu^{\IV,\beta ,\eta}_{\Lambda_L} \right)} \leq C \left( 1 + \beta^{-1} \right) L.
\end{equation}
Again, applying the triangle inequality and appealing to Proposition~\ref{prop.propueta} for an upper bound on $\|u_{\Lambda_L,\eta} \|_{\underline{L}^2 ( \Lambda_L^+)}$ for $d=1,2$ completes the proof. 
\end{proof}

\subsection{Lower bounds on gradient and height fluctuations} \label{subsec4.1}

The lower bounds of Theorem~\ref{prop3.100910disc} and Theorem~\ref{prop3.1009100bisdisc} are also deduced from a comparison of the gradient $\nabla \phi$ of the integer-valued random-field Gaussian free field to the gradient $\nabla u_{\Lambda,\eta}$ of the ground state of the real-valued random-field Gaussian free field. While it is possible to use Lemma~\ref{Thm7DGFF} to this end, the resulting bounds would deteriorate at high temperatures. Instead, we rely on the following lemma which shows that the \emph{thermal expectation} $ \left\langle \nabla \phi \right\rangle_{\mu^{\IV , \beta, \eta}_{\Lambda}}$ is close to $\lambda \nabla u_{\Lambda,\eta}$ uniformly in the temperature parameter (closeness of the thermal expectation suffices for the lower bounds, due to Jensen's inequality, but is insufficient for the upper bounds).

\begin{lemma} \label{lemma6.1}
 Let $d\ge 1$, $\Lambda\subset\Z^d$ be a finite subset of $\Zd$, $\lambda>0$ be the strength of the random field, $\beta > 0$ an inverse temperature, $\eta:\Lambda\to\R$ be an external field. Let $\phi$ be distributed according to the integer-valued random-field Gaussian free field $\mu^{\IV , \beta, \eta}_{\Lambda}$. Then for any map $w : \Lambda^+ \to \Z$ normalized to be $0$ on the boundary $\partial \Lambda$,
\begin{equation} \label{eq:17431209}
    \left| \sum_{e \in E \left(\Lambda^+\right)} \left( \left\langle \nabla \phi(e) \right\rangle_{\mu^{\IV , \beta, \eta}_{\Lambda}} - \lambda \nabla u_{\Lambda, \eta} (e) \right) \nabla w(e)\right| \leq  \frac12 \left\|  \nabla w \right\|_{L^2 \left( \Lambda^+ \right)}^2.
\end{equation}
Consequently, there exists a constant $C$ depending only on the dimension $d$ such that
\begin{equation} \label{eq:10231406}
    \left\| \left\langle \nabla \phi \right\rangle_{\mu^{\IV , \beta, \eta}_{\Lambda}} - \lambda \nabla u_{\Lambda, \eta} \right\|_{\underline{L}^2(\Lambda^+)} \leq C.
\end{equation}
\end{lemma}

\begin{proof}
We only prove the inequality~\eqref{eq:17431209} in the case $\eta = 0$; the general case can be obtained by a notational modification of the argument. By performing the discrete change of variable $\phi \to \phi + w$, we have the identity
\begin{equation} \label{eq:17301209}
    \sum_{\phi : \Lambda \to \Z} \exp \left( - \frac \beta2 \left\| \nabla \phi \right\|_{L^2 \left( \Lambda^+ \right)}^2  \right) \\ = \sum_{\phi : \Lambda \to \Z} \exp \left( - \frac \beta2 \left\| \nabla \phi + \nabla w \right\|_{L^2 \left( \Lambda^+ \right)}^2\right).
\end{equation}
Subtracting the right and left hand sides of the identity~\eqref{eq:17301209}, and expanding the square, we obtain the identity
\begin{equation*}
    \sum_{\phi : \Lambda \to \Z} \exp \left( - \frac \beta2 \left\| \nabla \phi \right\|_{L^2 \left( \Lambda^+ \right)}^2 \right) \left( \exp \left(  - \beta \sum_{e \in E \left(\Lambda^+\right)} \nabla  \phi(e) \nabla w(e) - \frac \beta2 \left\| \nabla w \right\|_{L^2 \left( \Lambda^+ \right)}^2 \right)  - 1 \right) = 0.
\end{equation*}
Using the identity $e^y - 1 \geq y$, we obtain the estimate
\begin{equation} \label{eq:17421209}
    \sum_{\phi : \Lambda \to \Z} \exp \left( - \frac \beta2 \left\| \nabla \phi \right\|_{L^2 \left( \Lambda^+ \right)}^2 \right) \left( - 2\sum_{e \in E \left(\Lambda^+\right)} \nabla  \phi(e) \nabla w(e) -  \left\|  \nabla w \right\|_{L^2 \left( \Lambda^+ \right)}^2 \right) \leq 0.
\end{equation}
Dividing both sides of the inequality~\eqref{eq:17421209} by the partition function $Z^{\IV, \beta, 0}_{\Lambda}$ yields the inequality: For any $w : \Lambda \to \Z$,
\begin{equation} \label{eq:174312099}
    2\sum_{e \in E \left(\Lambda^+\right)} \left\langle \nabla \phi(e) \right\rangle_{\mu^{\IV , \beta, \eta}_{\Lambda}} \nabla w(e) +  \left\|  \nabla w \right\|_{L^2 \left( \Lambda^+ \right)}^2 \geq 0.
\end{equation}
Using the definition of the mapping $u_{\Lambda, \eta}$, the inequality~\eqref{eq:174312099} can be equivalently rewritten as follows: For any $w : \Lambda \to \Z$,
\begin{equation} \label{eq:10301406}
    2\sum_{e \in E \left(\Lambda^+\right)} \left( \left\langle \nabla \phi(e) \right\rangle_{\mu^{\IV , \beta, \eta}_{\Lambda}} - \lambda \nabla u_{\Lambda, \eta} (e) \right) \nabla w(e) +  \left\|  \nabla w \right\|_{L^2 \left( \Lambda^+ \right)}^2 \geq 0.
\end{equation}
We then deduce the inequality~\eqref{eq:17431209} by applying inequality~\eqref{eq:10301406} with either the function $w$ or the function $-w$.

We next prove~\eqref{eq:10231406}. To this end, let us consider the integer-valued function $w = \lfloor  \lambda u_{\Lambda, \eta}  - \left\langle \phi \right\rangle_{\mu^{\IV , \beta, \eta}_{\Lambda}} \rfloor$. Using that the difference between the functions $w$ and $\lambda u_{\Lambda, \eta}  - \left\langle \phi \right\rangle_{\mu^{\IV , \beta, \eta}_{\Lambda}}$ is smaller than $1$ in absolute value and the inequality $(a+b)^2 \leq \frac54 a^2 + 5 b^2$, we obtain the inequality
\begin{equation*}
    \left\|  \nabla w \right\|_{\underline{L}^2 \left( \Lambda^+ \right)}^2 \leq \frac 54 \left\|  \left\langle \nabla \phi \right\rangle_{\mu^{\IV , \beta, \eta}_{\Lambda}} - \lambda \nabla u_{\Lambda, \eta}  \right\|_{\underline{L}^2 \left( \Lambda^+ \right)}^2 + C.
\end{equation*}
A similar argument yields
\begin{align*}
    \lefteqn{2\sum_{e \in E \left(\Lambda^+\right)} \left( \left\langle \nabla \phi(e) \right\rangle_{\mu^{\IV , \beta, \eta}_{\Lambda}} - \lambda \nabla u_{\Lambda, \eta} (e) \right) \nabla w(e)} \qquad & \\ &
    \leq - 2\left\| \left\langle \nabla \phi \right\rangle_{\mu^{\IV , \beta, \eta}_{\Lambda}} - \lambda \nabla u_{\Lambda, \eta} \right\|_{L^2 (\Lambda^+)}^2 + C  \sum_{e \in E \left(\Lambda^+\right)} \left| \left\langle \nabla \phi(e) \right\rangle_{\mu^{\IV , \beta, \eta}_{\Lambda}} - \lambda \nabla u_{\Lambda, \eta} (e) \right| \\
    & \leq - \frac74 \left\| \left\langle \nabla \phi \right\rangle_{\mu^{\IV , \beta, \eta}_{\Lambda}} - \lambda \nabla u_{\Lambda, \eta} \right\|_{L^2 (\Lambda^+)}^2 + C \left|\Lambda^+\right|
\end{align*}
where we used the inequality $ab \leq \frac{1}{4C} a^2 + 16C b^2$ in the second line. A combination of the three previous displays implies the inequality~\eqref{eq:10231406}.
\end{proof}

\begin{remark}
We point out that the inequality~\eqref{eq:17431209} also follows from the non-negativity of the Kullback-Leibler divergence between the distribution of $\phi$ and $\phi+w$, when $\phi$ is sampled from $\mu^{\IV , \beta, \eta}_{\Lambda}$.
\end{remark}

\begin{proof}[Proof of the lower bounds of Theorem~\ref{prop3.100910disc} and Theorem~\ref{prop3.1009100bisdisc}]
We first prove the lower bounds of Theorem~\ref{prop3.100910disc}. To this end, we use the triangle and Jensen's inequalities to write
\begin{align*}
    \left\| \lambda \nabla u_{\Lambda_L , \eta} \right\|_{\underline{L}^2(\Lambda^+_L )} - \left\| \left\langle \nabla \phi \right\rangle_{\mu^{\IV , \beta, \eta}_{\Lambda}} - \lambda \nabla u_{\Lambda_L, \eta} \right\|_{\underline{L}^2(\Lambda^+_L)} & \leq \left\| \left\langle \nabla \phi \right\rangle_{\mu^{\IV , \beta, \eta}_{\Lambda}} \right\|_{\underline{L}^2(\Lambda^+_L)} \\
    & \leq  \left\| \nabla \phi \right\|_{\underline{L}^2(\Lambda^+_L , \mu^{\IV , \beta, \eta}_{\Lambda})}.
\end{align*}
The lower bounds of Theorem~\ref{prop3.100910disc} are then obtained by appealing to the inequality~\eqref{eq:10231406} and to Proposition~\ref{prop.propueta} for a lower bound on $\|\nabla u_{\Lambda_L,\eta} \|_{\underline{L}^2 ( \Lambda_L^+)}$.

Let us now prove the lower bounds of Theorem~\ref{prop3.1009100bisdisc}.
Fix an integer $L \geq 1$. Let $\{X_n\}$ be the simple random walk on $\Zd$ and define $\tau_L:= \min \{n \geq 0\, :\, X_n \in \Zd \setminus \Lambda_L\}$. For $x \in \Zd$, denote by $\P_x$ and $\E_x$ the law and expectation with respect to $\{X_n\}$ started from $X_0 = x$. Define a map $P_L : \Zd \to \R$ by
\begin{equation} \label{eq:09260408}
    P_L(x) := \frac1{2d}\mathbb{E}_x[ \tau_L ].
\end{equation}
The Markov property for the random walk implies that, for any $x \in \Lambda_{L}$, 
\[
2d P_L(x) = \sum_{y \in \Zd}  \mathbb{P}_x[X_1 = y] \cdot \mathbb{E}_y[1 + \tau_L ] = 1 + \sum_{\substack{y \in \Zd \\ x \sim y}} P_L(y);
\]
rearranging this equation shows that $-\Delta P_L = 1$ in $\Lambda_L$. Furthermore, standard estimates on the random walk imply that
\[
c  \inf_{z \in \Zd \setminus \Lambda_{ L}} |x-z|^2 \leq P_L(x) \leq  C  \inf_{z \in \Zd \setminus \Lambda_{ L}} |x-z|^2. 
\]
Therefore, we find that 
\[
\|P_L\|_{\underline{L}^2(\Lambda_L)} \geq cL^2  , \; \; \text{ and } \; \;   \|\nabla P_L\|_{\underline{L}^2(\Lambda_L)} = \left(\frac{1}{|\Lambda_L|} \sum_{x \in \Lambda_L} P_L(x) \Delta P_L(x)\right)^{1/2} \leq C L,
\]
where we use discrete integration by parts in the bound on the norm of $\nabla P_L$.

By Jensen's and the triangle inequalities,  
\begin{equation} \label{eq:1941}
\left\| \phi \right\|_{\underline{L}^2 \left(\Lambda_L^+ ,  \mu^{\IV, \beta ,\eta}_{\Lambda_L}  \right)}  \geq  \left\| \left\langle \phi \right\rangle_{\mu^{\IV, \beta ,\eta}_{\Lambda_L}} \right\|_{\underline{L}^2 \left(\Lambda_L^+  \right)} \geq \| \lambda u_{\Lambda_L,\eta} \|_{\underline{L}^2(\Lambda_L)}  -\left\| \left\langle \phi \right\rangle_{\mu^{\IV, \beta ,\eta}_{\Lambda_L}} - \lambda u_{\Lambda_L,\eta} \right\|_{\underline{L}^2 \left(\Lambda_L^+\right)}. 
\end{equation}
We next estimate the second term in the right-hand side of the previous display. To this end, we apply Poincar\'e inequality and the estimate~\eqref{eq:10231406}. We obtain
\begin{equation} \label{eq:1942}
    \left\| \left\langle \phi \right\rangle_{\mu^{\IV, \beta ,\eta}_{\Lambda_L}} - \lambda u_{\Lambda_L,\eta} \right\|_{\underline{L}^2 \left(\Lambda_L^+  \right)} \leq CL \left\| \left\langle \nabla  \phi \right\rangle_{\mu^{\IV, \beta ,\eta}_{\Lambda_L}} - \lambda \nabla u_{\Lambda_L,\eta} \right\|_{\underline{L}^2 \left(\Lambda_L^+ \right)} \leq CL.
\end{equation}
Thus, 
\begin{equation}\label{eq:ProbRepIVLB}
\mathbb{P}\left[\left\| \phi \right\|_{\underline{L}^2 \left(\Lambda_L^+ ,  \mu^{\IV, \beta ,\eta}_{\Lambda_L}  \right)} > c L^{2 - d/2}\right] \geq \mathbb{P}\left[\| \lambda u_{\Lambda_L,\eta} \|_{\underline{L}^2(\Lambda_L)} > c L^{2-d/2} + CL\right].
\end{equation}
To complete the proof, we wish to show that $\| u_{\Lambda_L,\eta} \|_{\underline{L}^2(\Lambda_L)}$ has large fluctuations. By the Cauchy-Schwarz inequality, the construction of $P_L$ and discrete integration by parts, we may write 
\begin{align*}
\|\lambda  u_{\Lambda_L,\eta} \|_{\underline{L}^2(\Lambda_L)} & \geq \frac{1}{|\Lambda_L|} \sum_{x \in \Lambda_L} \lambda  u_{\Lambda_L,\eta}(x) \\
& \geq \frac{1}{|\Lambda_L|} \sum_{x \in \Lambda_L} \lambda  u_{\Lambda_L,\eta}(x) (- \Delta P_L(x)) \\
& = \frac{1}{|\Lambda_L|} \sum_{x \in \Lambda_L} \lambda \eta(x) P_L(x).
\end{align*}
We will use the Central Limit Theorem to prove that, as $L$ grows, 
\begin{equation}\label{eq:CLT}
\frac{1}{\lambda |\Lambda_L|^{1/2} \| P_L\|_{\underline{L}^2(\Lambda_L) }} \sum_{x \in \Lambda_L} \lambda \eta(x) P_L(x) \rightarrow  N\left(0, 1 \right),
\end{equation}
where $N(0,1)$ is the standard normal distribution. Indeed, the summands in~\eqref{eq:CLT} are independent with mean zero, and $\lambda \eta(x) P_L(x)$ has variance $\lambda^2 P_L(x)^2$.  Since $|P_L(x)| \leq C L^2$, $\sum_{x \in \Lambda_L} P_L(x)^2 > cL^{4 +d} $ and the $(\eta(x))$ are identically distributed, the sum~\eqref{eq:CLT} satisfies the Lindeberg-Feller Central Limit Theorem (see~\cite[Theorem 3.4.10]{Dur2019book}).

To connect the two previous displays, we point out that
\[
\frac{1}{|\Lambda_L|} \ge c\lambda L^{2 - d/2} \cdot \frac{1}{\lambda |\Lambda_L|^{1/2} \| P_L\|_{\underline{L}^2(\Lambda_L) }}.
\]
Thus, combining the previous four displays, we find that  
\[
\liminf_{L \to \infty} \mathbb{P}\left[\left\| \phi \right\|_{\underline{L}^2 \left(\Lambda_L^+ ,  \mu^{\IV, \beta ,\eta}_{\Lambda_L}  \right)} > c L^{2 - d/2}\right] \geq \liminf_{L \to \infty} \mathbb{P}\left[N(0,1)  >\frac{1}{\lambda}  \left( c + CL^{d/2 -1} \right)\right] \geq e^{-C/ \lambda^2}
\]
in dimensions $d=1,2$. Since this bound is uniform in $L$, we deduce the desired lower bound on the expectation of  $\left\| \phi \right\|_{\underline{L}^2 \left(\Lambda_L^+ ,  \mu^{\IV, \beta ,\eta}_{\Lambda_L}\right)}$.
\end{proof}

\begin{remark} \label{remark6.4}
We complete this section by establishing the lower bound stated in Remark~\ref{remark1.1}.
Using that the random field $(\eta(x))_{x \in \Zd}$ is assumed to be i.i.d. with expectation zero and variance one, we see that there exists a constant $\kappa > 0$ depending on the law of the random field such that, for any $x \in \Zd$,
\begin{equation*}
    \P \left[ \left| \eta(x) \right| \geq 1 \right] \geq 2\kappa.
\end{equation*}
We may then assume, without loss of generality, that $\P \left[  \eta(x)  \geq 1 \right] \geq \kappa$.
We next let $M$ be the smallest integer larger than $1$ such that
\begin{equation*}
    \left| \Lambda_M \right| \geq \frac 1\lambda \left| \partial \Lambda_M \right|.
\end{equation*}
We assume that $L \geq 2M$ and consider the random set
\begin{equation*}
    \mathcal{E}_{L, \lambda} := \left\{ x \in \Lambda_{\lceil L/2\rceil} \cap (2M+1) \Zd, \,\forall y \in (x + \Lambda_M), \, \eta(y) \geq 1 \right\}.
\end{equation*}
We next set
\begin{equation*}
    E_{L, \lambda} := \bigcup_{x \in\mathcal{E}_{L, \lambda}} \left( x + \Lambda_M \right).
\end{equation*}
Since the random field is i.i.d., there exists a constant $c > 0$ depending only on the dimension, the law of the random field, and the disorder strength, such that
\begin{equation} \label{eq:23002002}
    \E \left[ \frac{\left| E_{L , \lambda} \right|}{\left| \Lambda_L^+ \right|} \right] \geq c.
\end{equation}
We next introduce the function $w := \indc_{E_{L , \lambda}}$ and observe that, by the definition of the set $E_{L , \lambda}$,
\begin{equation} \label{eq:22282002}
    \sum_{e \in E \left(\Lambda^+_L \right)}  \lambda \nabla u_{\Lambda_L, \eta} (e)  \nabla w(e) = \sum_{x \in \Lambda_L^+} \lambda \eta(x) \indc_{E_{L , \lambda}} \geq \lambda \left| E_{L , \lambda} \right| ~\mbox{and}~ \left\|  \nabla w \right\|_{L^2 \left( \Lambda^+_L \right)}^2 \leq  \left| E_{L , \lambda} \right| \times \left( \frac{\left| \partial \Lambda_M \right|}{\left| \Lambda_M \right|} \right). 
\end{equation}
Combining the bounds of~\eqref{eq:22282002} with~\eqref{eq:17431209}, the definition of the integer $M$ and the triangle inequality, we deduce that
\begin{align*}
    \left| \sum_{e \in E \left(\Lambda^+_L \right)} \left\langle \nabla \phi(e) \right\rangle_{\mu^{\IV , \beta, \eta}_{\Lambda_L}} \nabla w(e) \right|
   & \geq \left| \sum_{e \in E \left(\Lambda^+_L\right)} \lambda \nabla u_{\Lambda_L, \eta} (e) \nabla w(e)\right| - \frac12 \left\|  \nabla w \right\|_{L^2 \left( \Lambda^+_L \right)}^2 \\
    & \geq  \lambda  \left| E_{L , \lambda} \right| - \frac{\lambda}{2} \left| E_{L , \lambda} \right| \\
    & \geq \frac{\lambda}{2}  \left| E_{L , \lambda} \right| .
\end{align*}
Using the Cauchy-Schwarz inequality and noting that $\left| \nabla w(e) \right| \leq 1$ for any edge $e \in E \left( E_{\Lambda_L^+} \right)$  since $w$ is an indicator function, we deduce that
\begin{equation*}
    \left\| \nabla \phi \right\|_{\underline{L}^2(\Lambda^+_L , \mu^{\IV , \beta, \eta}_{\Lambda})} \geq \frac{{c'}}{\left| \Lambda_{L}^+ \right|} \sum_{e \in E \left(\Lambda_L \right)} \left\langle \left| \nabla \phi(e) \right|  \right\rangle_{\mu^{\IV , \beta, \eta}_{\Lambda_L}} \geq \frac{{c'}\lambda   \left| E_{L , \lambda} \right|}{2 \left| \Lambda_{L}^+ \right|},
\end{equation*}
where the constant $c'$ depends only on dimension (due to the normalization counting vertices instead of edges). Taking the expectation on both sides of the previous inequality and using the estimate~\eqref{eq:23002002} completes the proof.
\end{remark}

\subsection{Height fluctuations in dimensions $d \geq 3$}
We conclude the paper by proving Theorem~\ref{thm10162010}. The argument relies on a Peierls-type argument, and we will make use of the two following propositions. The first is a result of Fisher--Fr\"{o}hlich--Spencer~\cite{fisher1984ising} (see also Chalker~\cite{C83}).

\begin{proposition}[Fisher--Fr\"{o}hlich--Spencer~\cite{fisher1984ising}] \label{eq:11182010}
Let $d \geq 3$ and $\left( \eta(x) \right)_{x \in \Zd}$ be a collection of independent standard Gaussian random variables. There exist two constants $c, \lambda' > 0$ such that for any $\lambda \in (0, \lambda']$ and any $v \in \Zd$, the event
\begin{equation} \label{eq:10242010}
    \mathcal{E}_{\lambda, v,N} := \left\{\eta : \Zd \to \R \, : \, \begin{gathered}\text{For all $\Lambda \subseteq \Zd$ with $v \in \Lambda$ and with $\Lambda$ connected, having connected}\\ \text{complement, and } |\partial \Lambda| \geq N, \text{ it holds that $\lambda \left| \sum_{x \in \Lambda} \eta(x) \right| \leq |\partial \Lambda|$}\end{gathered} \right\}
\end{equation}
satisfies
\begin{equation*}
    \P \left(  \mathcal{E}_{\lambda, v,N} \right) \geq  1 - e^{-\frac{c N^{1/3}}{\lambda^2}}.
\end{equation*}
\end{proposition}

We mention that the original proof of Fisher--Fr\"{o}hlich--Spencer~\cite{fisher1984ising} is written for the case $d = 3$ only, and states that 
\[
 \P \left(  \mathcal{E}_{\lambda, v,1} \right) \geq  1 - e^{-\frac{c}{\lambda^2}}.
\]
However, the proof proves the stronger statement that we have claimed above, and the argument applies to  any dimension $d \geq 3$ (as remarked at the end of the proof in~\cite{fisher1984ising}). In fact, the exponent $1/3$ can be improved to a constant approaching $1$ as $d$ grows.

The second result is a classical upper bound on the number of bounded connected subsets of $\Zd$ containing a fixed vertex $v$ and with a fixed boundary size.

\begin{proposition}[Lemma 5.3.5 of \cite{ruelle1999statistical}, \cite{LM98, BB07}]\label{prop:10252010}
There exists a constant $\alpha_d > 0$ depending only on the dimension such that for each integer $N \in \N$ and each $v \in \Zd$, the set
\begin{equation*}
    A_{N,v} :=  \left\{ \Lambda \subseteq \Zd \, : \, v \in \Lambda, \Lambda \mbox{ is connected, finite, } \partial \Lambda  \mbox{ is } \star-\mbox{connected and} \, \left| \partial \Lambda  \right| = N \right\}
\end{equation*}
satisfies
\begin{equation*}
    \left| A_{N,v} \right| \leq e^{\alpha_d N}.
\end{equation*}
\end{proposition}

We are now ready to give the proof of Theorem~\ref{thm10162010}.

\begin{proof}[Proof of Theorem~\ref{thm10162010}]
We need only prove the estimate~\eqref{eq:10142010} as the bound~\eqref{eq:10152010} on the $k$-th moment can be deduced from~\eqref{eq:10142010} by (suitably) integrating over the $t$ variable. Fix $d \geq 3$, a side length $L \geq 2$, a vertex $v \in \Lambda_L$, and an integer $t \geq 1$. Let $\lambda'$ and $\alpha_d$ be the constants appearing in the statements of Proposition~\ref{eq:11182010} and Proposition~\ref{prop:10252010}. We also define $s_0 = t^{\frac{d-2}{2(d-1)}}$. Let $\beta_0 \in (1 , \infty)$, and $c \in (0 , 1)$ be constants depending only on the dimension whose value will be selected later in the argument. Set $ \lambda_0 := (1/4 \wedge c) \lambda'$ and fix $\lambda \in (0,\lambda_0)$ and $\beta \in (\beta_0,\infty)$.

By an application of Proposition~\ref{eq:11182010}, we see that it is sufficient for the proof of Theorem~\ref{thm10162010} to prove the inclusion
\begin{equation} \label{eq:13382010}
  \mathcal{E}_{4\lambda,v,  t^{1/2}} \cap  \mathcal{E}_{\lambda/(cs_0),v,1} \subseteq \left\{ \eta : \Zd \to \R \, : \, \mu^{\IV , \beta, \eta}_{\Lambda_L} \left( |\phi(v)| < t \right) \geq 1-  e^{- c \beta t^{1/2}} \right\} .
\end{equation}
Let us fix a realization of the random field $\eta \in \mathcal{E}_{4\lambda,v,  t^{1/2}} \cap  \mathcal{E}_{\lambda/(cs_0),v,1}$, and let $\phi : \Lambda_L \to \Z$ be a random surface distributed according to the random-field integer-valued Gaussian free field $\mu^{\IV , \beta, \eta}_{\Lambda_L}$. We denote by $D_+^0 \left( \phi \right)$ the connected component of $v$ in the set $\left\{ x \in \Lambda_L \, : \, \phi(x) \geq 1  \right\}$, and set $D_+^0 \left( \phi \right) = \emptyset$ if $\phi(v) \leq 0$.
We then define the set $D_+(\phi)$ to be the union of the set $D_+^0(\phi)$ and of all the finite connected components of $\Zd \setminus D_+^0(\phi)$. This procedure ensures that the sets $D_+(\phi)$ and $\Zd \setminus D_+(\phi)$ are connected and thus that the boundary $\partial D_+(\phi)$ is $\star$-connected (see~\cite[Lemma 2]{T13}), and will allow us to apply Proposition~\ref{eq:11182010}.
Symmetrically, we let $D_-^0 \left( \phi \right)$ be the connected component of $v$ in the set $\left\{ x \in \Lambda_L \, : \, \phi(x) \leq -1  \right\}$, set $D_-^0 \left( \phi \right) = \emptyset$ if $\phi(v) \geq 0$, and let $D_-(\phi)$ to be the union of the set $D_-^0(\phi)$ and of all the finite connected components of $\Zd \setminus D_-^0(\phi)$. We note that either $D_-(\phi)$ or $D_+(\phi)$ is empty, and both are empty if and only if $\phi(v) = 0$.

For convenience, we will assume that $D_+(\phi)$ is nonempty, as the other case follows from symmetry. We say that $\phi$ is $s$-steep if there exists a cutset of edges $\mathcal{S}$ that separates $v$ from $\partial D_+(\phi)$ such that $\nabla \phi(e) \geq  s$ for each  $e \in \mathcal{S}$. We first claim that, for some constant $c_d $ depending on dimension only,
\begin{equation}\label{eq:phi v large alternative}
   \text{$|\phi(v)| \geq t$ implies that either $|\partial D_+(\phi)| > t^{1/2}$ or $\phi$ is $c_d s_0$-steep.}
\end{equation} Indeed, assume that $|\partial D_+(\phi)| \leq t^{1/2}$ and $\phi$ is not $c_d s_0$-steep and let us show that $\phi(v) \leq t$. Let $c_d>0$ be sufficiently small for the following arguments. By standard isoperimetric estimates on $\mathbb{Z}^d$ (see~\cite[Lemma 3.1]{magazinov2020concentration}, following~\cite{bollobas1991edge}), there are at most $t^{\frac{d}{2(d-1)}}/c_{d}$ edges with at least one endpoint in $D_{+}(\phi)$. Furthermore, standard duality between paths and cutsets tells us that, whenever $\phi$ is not $c_d s_0$-steep, there exists a path from $\partial D_+(\phi)$ to $v$ all of whose edges $e$ satisfy $|\nabla\phi(e)| < c_d s_0$ and have at least one endpoint in $\partial D_+(\phi)$. Let $\gamma_{wv}$ be such a path, whose starting point is some $w\in\partial D_+(\phi)$. Then
\[
\phi(v)= \phi(w) + \sum_{e \in \gamma_{wv}} \nabla \phi(e) \leq \phi(w) + c_d s_0  \cdot |\gamma_{wv}| \leq t, 
\]
where we used $\phi(w) \leq 0$ since $w\in\partial D_+(\phi)$ and the fact that all edges of $\gamma_{wv}$ have at least one endpoint in $D_+(\phi)$. This establishes~\eqref{eq:phi v large alternative}.

Thus, the inclusion~\eqref{eq:13382010} follows from showing that for any $\eta \in \mathcal{E}_{4\lambda,v, t^{1/2}} \cap  \mathcal{E}_{\lambda/(cs_0),v,1}$,
\begin{equation}\label{eq:two events to control}
\mu^{\IV , \beta, \eta}_{\Lambda_L} \left( |\partial D_+ \left( \phi \right)| > t^{1/2} \right) +  \mu^{\IV , \beta, \eta}_{\Lambda_L} \left( \phi \text{ is } c_ds_0\text{-steep} \right) \leq e^{- c\beta t^{1/2}}.
\end{equation}

We start the proof of~\eqref{eq:two events to control} by estimating the first term on its left-hand side. To this end we decompose the event $\{ |\partial D_+ \left( \phi \right)| > t^{1/2}\}$ according to the identity
\begin{equation} \label{eq:11042110}
   \mu^{\IV , \beta, \eta}_{\Lambda_L} \left( |\partial D_+ \left( \phi \right)| > t^{1/2} \right)  = \sum_{\substack{ D \subseteq \Lambda_L \\ |\partial D| > t^{1/2}}} \mu^{\IV , \beta, \eta}_{\Lambda_L} \left(  D_+ \left( \phi \right) = D \right),
\end{equation}
where the sum in the right-hand side is computed over all the bounded connected subsets $D$ containing the vertex $v$ such that $\partial D$ is $\star$-connected and has more than $t^{1/2}$ vertices. Let us fix such a set $D \subseteq \Lambda_L$. We have the identity
\begin{equation} \label{eq:10472110}
      \mu^{\IV , \beta, \eta}_{\Lambda_L} \left(D_+ \left( \phi \right) = D \right) = \frac{\sum_{\substack{\phi: \Lambda_L \to \Z \\ D_+\left( \phi \right) = D}} \exp \left( - \frac{\beta}{2} \left\| \nabla \phi\right\|_{L^2 \left( \Lambda_L^+ \right)}^2 + \beta \lambda \sum_{x \in \Lambda_L} \eta(x) \phi(x)  \right)}{\sum_{\phi: \Lambda_L \to \Z} \exp \left( - \frac{\beta}{2} \left\| \nabla \phi\right\|_{L^2 \left( \Lambda_L^+ \right)}^2 + \beta \lambda \sum_{x \in \Lambda_L} \eta(x) \phi(x)  \right)}.
\end{equation}
For each integer-valued $\phi : \Lambda_L \to \Z$  (with zero boundary values) such that $D_+ \left( \phi \right) = D$, we define $\tilde \phi$ according to the formula $\tilde\phi := \phi - \indc_D$. Using that, by the identity $D_+ \left( \phi \right) = D$, the map $\phi$ satisfies $\phi \geq 1$ on $D$ and $\phi \leq 0$ on $\partial D$, we have the inequality, for any edge $e \in E \left( \Lambda_L^+\right)$,
\begin{equation} \label{eq:18532010}
    \left( \nabla \tilde \phi (e) \right)^2 = \left( \nabla \phi (e) \right)^2 - \left( 2 \left| \nabla \phi(e) \right| - 1 \right) \left| \nabla \indc_D(e) \right| \leq \left( \nabla  \phi (e) \right)^2 - \left| \nabla \indc_D(e) \right|.
\end{equation}
Summing the inequality~\eqref{eq:18532010} over the edges $e \in E \left( \Lambda_L^+ \right)$ and using that the cardinality of the support of the mapping $\nabla \indc_D$ is at least $\left| \partial D \right|$ (since an edge belongs to the support of this map if and only if it has exactly one endpoint in $D$ and one endpoint in $\partial D$), we obtain
\begin{equation} \label{eq:10402110}
    \left\| \nabla \tilde \phi \right\|_{L^2 \left( \Lambda_L^+ \right)}^2 \leq \left\| \nabla \phi \right\|_{L^2 \left( \Lambda_L^+ \right)}^2 - \left| \partial D \right|.
\end{equation}
Combining the inequality~\eqref{eq:10402110} with the assumption $\eta \in \mathcal{E}_{4\lambda,v, t^{1/2}}$ (and noting that all sets we are interested in here have boundary of at least $t^{1/2}$ vertices), we obtain
\begin{align*}
     \frac 12\left\| \nabla \tilde \phi\right\|_{L^2 \left( \Lambda_L^+ \right)}^2 - \lambda \sum_{x \in \Lambda_L} \eta(x) \tilde \phi(x) & = \frac 12 \left\| \nabla \tilde \phi\right\|_{L^2 \left( \Lambda_L^+ \right)}^2 - \lambda \sum_{x \in \Lambda_L} \eta(x)  \phi(x) + \lambda \sum_{x \in D} \eta(x)   \\
     & \leq \frac 12 \left\| \nabla  \phi\right\|_{L^2 \left( \Lambda_L^+ \right)}^2 - \lambda \sum_{x \in \Lambda_L} \eta(x)  \phi(x) - \frac{\left| \partial D \right|}{2} + \lambda \sum_{x \in D} \eta(x) \\
     & \leq  \frac 12 \left\| \nabla  \phi\right\|_{L^2 \left( \Lambda_L^+ \right)}^2 - \lambda \sum_{x \in \Lambda_L} \eta(x)  \phi(x) - \frac{\left| \partial D \right|}{4}.
\end{align*}
Using the previous computation, we may write
\begin{align} \label{eq:10492110}
    \lefteqn{\sum_{\substack{\phi: \Lambda_L \to \Z \\ D_+\left( \phi \right) = D}} \exp \left( - \frac{\beta}{2} \left\| \nabla \phi\right\|_{\underline{L}^2 \left( \Lambda_L^+ \right)}^2 + \beta \lambda \sum_{x \in \Lambda_L} \eta(x) \phi(x)  \right)} \qquad & \\ & \leq \exp \left( - \frac{ \beta\left| \partial D \right|}{4} \right)  \sum_{\substack{\phi: \Lambda_L \to \Z \\ D_+\left( \phi \right) = D}} \exp \left( - \frac{\beta}{2} \left\| \nabla \tilde \phi\right\|_{\underline{L}^2 \left( \Lambda_L^+ \right)}^2 + \beta \lambda \sum_{x \in \Lambda_L} \eta(x) \tilde \phi(x)  \right) \notag \\
    & \leq \exp \left( - \frac{\beta \left| \partial D \right|}{4} \right) \sum_{\phi: \Lambda_L \to \Z} \exp \left( - \frac{\beta}{2} \left\| \nabla  \phi\right\|_{\underline{L}^2 \left( \Lambda_L^+ \right)}^2 + \beta \lambda \sum_{x \in \Lambda_L} \eta(x)  \phi(x)  \right). \notag
\end{align}
A combination of the identity~\eqref{eq:10472110} and the inequality~\eqref{eq:10492110} implies
\begin{equation}
\label{eq:11032110}
     \mu^{\IV , \beta, \eta}_{\Lambda_L} \left(D_+ \left( \phi \right) = D \right)\leq \exp \left(-\frac{\beta \left|\partial D \right|}{4} \right).
\end{equation}
Putting the estimate~\eqref{eq:11032110} back into~\eqref{eq:11042110}, using Proposition~\ref{prop:10252010}, the assumption $\beta \geq \beta_0$, selecting $\beta_0 $ large enough and $c$ small enough, and summing over sets with boundary of size at least $t^{1/2}$,  we deduce that
\begin{equation}\label{eq:estimate of first term}
    \mu^{\IV , \beta, \eta}_{\Lambda_L} \left( |\partial D_+ \left( \phi \right)| > t^{1/2} \right) \leq \sum_{ \substack{ D \subseteq \Lambda_L \\ |\partial D| > t^{1/2}}} e^{-\frac{\beta \left| \partial D \right|}{4}} \leq \sum_{N = \lceil t^{1/2} \rceil }^{L^d} \left| A_{N,v} \right| e^{-\frac{\beta N}{4}}  \leq \sum_{N = \lceil t^{1/2} \rceil}^{L^d} e^{\alpha_d N } e^{-\frac{\beta N}{4}}  \leq \tfrac12 e^{- c\beta t^{1/2}}.
\end{equation}

Next, we wish to analyze the second term on the left-hand side of~\eqref{eq:two events to control}, i.e., the probability that $\phi$ is $c_d s_0$-steep. The logic is very similar to the proof above, except that the set $D_+(\phi)$ will be replaced by a set where the gradient of $\phi$ is large. More specifically, for every $c_d s_0$-steep $\phi$, let $\bar{\mathcal{S}}$ be the outermost cutset of edges separating $v$ from $\partial D_+(\phi)$ on which $\nabla \phi(e) \geq c_d s_0$ for every edge $e$; this is well defined, e.g., using an exploration procedure from the boundary of $D_+(\phi)$. We then define $\bar{I}$ to be the set of vertices in the interior of $\bar{\mathcal{S}}$ -- i.e., the vertices $u$ such that every path from $u$ to $\partial D_+(\phi)$ must include at least one edge in $\bar{\mathcal{S}}$. If $\phi$ is not $c_d s_0$-steep, set $\bar{I}$ to be empty. Thus,
\begin{equation} \label{eq:DecompI}
 \mu^{\IV , \beta, \eta}_{\Lambda_L} \left( \phi \text{ is } c_d s_0\text{-steep} \right) =  \mu^{\IV , \beta, \eta}_{\Lambda_L} \left( \bar{I} \neq \emptyset \right) = \sum_{ I \subseteq \Lambda_L}  \mu^{\IV , \beta, \eta}_{\Lambda_L} \left(  \bar{I} = I \right),
\end{equation}
where the sum is, again, over the connected subsets of $\Lambda_L$ containing $v$ whose complement is also connected.  Repeating the procedure used above, we define $\hat{\phi} = \phi - \lceil c_d s_0 \rceil \cdot \indc_{I}$, and observe that
\begin{equation} \label{eq:SubtractingLargeBound}
    \left( \nabla \hat \phi (e) \right)^2 = \left( \nabla \phi (e) \right)^2 - \lceil c_d s_0 \rceil \left( 2 \nabla \phi(e)  -  \lceil c_d s_0 \rceil \right)  |\nabla \indc_{I}|  \leq \left( \nabla  \phi (e) \right)^2 - \lceil c_d s_0 \rceil^2 |\nabla \indc_{I}|,
\end{equation}
where we observe that every edge in the support of $|\nabla \indc_{I}|$ must be an edge of the cutset $\bar{\mathcal{S}}$, and hence $\nabla \phi(e)$ will be at least $\lceil c_d s_0\rceil$, granting the final inequality. Summing the inequality~\eqref{eq:SubtractingLargeBound} over all edges, using the assumption $\eta \in \mathcal{E}_{\lambda/(cs_0), v,1}$ and choosing the constant $c$ small enough (e.g., smaller than $c_d/4$), we have that
\begin{align*}
     \frac 12\left\| \nabla \hat \phi\right\|_{L^2 \left( \Lambda_L^+ \right)}^2 - \lambda \sum_{x \in \Lambda_L} \eta(x) \hat \phi(x) & = \frac 12 \left\| \nabla \hat \phi\right\|_{L^2 \left( \Lambda_L^+ \right)}^2 - \lambda \sum_{x \in \Lambda_L} \eta(x)  \phi(x) + \lambda \lceil c_d s_0\rceil \sum_{x \in I} \eta(x)   \\
     & \leq \frac 12 \left\| \nabla  \phi\right\|_{L^2 \left( \Lambda_L^+ \right)}^2 - \lambda \sum_{x \in \Lambda_L} \eta(x)  \phi(x) - \lceil c_d s_0 \rceil^2 \left( \frac{  \left|\partial I\right|}{2} - \frac{\lambda}{\lceil c_d s_0 \rceil} \sum_{x \in I} \eta(x) \right) \\
     & \leq  \frac 12 \left\| \nabla  \phi\right\|_{L^2 \left( \Lambda_L^+ \right)}^2 - \lambda \sum_{x \in \Lambda_L} \eta(x)  \phi(x) - \frac{c_d^2 s_0^2}{4}  \cdot |\partial I|.
\end{align*}
Repeating the earlier logic, we deduce that
\[
  \mu^{\IV , \beta, \eta}_{\Lambda_L} \left(  \bar{I} = I \right) \leq \exp \left(-\frac{\beta c_d^2 s_0^2  \cdot |\partial I|}{4} \right).
\]
Applying Proposition~\ref{prop:10252010} and selecting $\beta_0$ large enough and $c$ small enough yields
\begin{equation}\label{eq:estimate of second term}
 \mu^{\IV , \beta, \eta}_{\Lambda_L} \left( \phi \text{ is } c_d s_0\text{-steep}\right) \leq \tfrac12 e^{-  c \beta t^{1/2}},
\end{equation}
where we used the definition of $s_0$ and the fact that $d\ge 3$. Plugging~\eqref{eq:estimate of first term} and~\eqref{eq:estimate of second term} into~\eqref{eq:two events to control} completes the proof.
\end{proof}

\section{Discussion and open questions}\label{sec:discussion and open questions}

In this section we provide further discussion and highlight several research directions.

\subsection{The integer-valued random-field Gaussian free field and the random-phase sine-Gordon model.}
Theorem~\ref{prop3.100910disc} determines that the gradient of the integer-valued random-field Gaussian free field delocalizes in dimensions $d=1,2$ and localizes in dimensions $d\ge 3$, at all temperatures (including zero temperature) and all positive disorder strengths $\lambda$. Our results for the height fluctuations are less complete: Theorem~\ref{prop3.1009100bisdisc} proves delocalization in dimensions $d=1,2$, again at all temperatures and positive disorder strengths, while Theorem~\ref{thm10162010} establishes localization in dimensions $d\ge3$ at low temperature and weak disorder (unlike the real-valued case, which delocalizes in dimensions $d=3,4$; see Theorem~\ref{prop3.1009100bis}). The height fluctuations in dimensions $d\ge 3$ in the high temperature or strong disorder regimes remain unclear. We emphasize that delocalization for some $d\ge 3$, temperature and disorder strength would constitute a \emph{roughening transition} from the localized behavior proved in Theorem~\ref{thm10162010}, and would thus be of significant interest. Unlike the non-disordered case, where roughening transitions are familiar only in two dimensions~\cite{FrSp, frohlich1981kosterlitz, L20, V06}, we conjecture that such a transition indeed takes place for the three-dimensional integer-valued random-field Gaussian free field, with delocalization occurring at least in the low-temperature and strong disorder regime. We further conjecture that no such transition occurs in dimensions $d\ge 5$, where we expect the surface to remain localized at all temperatures and disorder strengths (as in the real-valued case; see Theorem~\ref{prop3.1009100bis}). These conjectures are supported by a connection of the model with the random-phase sine-Gordon model, on which we elaborate next. We tend to think that a roughening transition also occurs in the intermediate four-dimensional case.

The \emph{sine-Gordon model} is a model of real-valued surfaces $\psi$ whose Hamiltonian on a domain $\Lambda$ with given boundary conditions takes the form
\begin{equation}\label{eq:sine-Gordon Hamiltonian}
\sum_{e\in E(\Lambda^+)} (\nabla \psi (e))^2 -  z\sum_{v\in \Lambda} \cos(2\pi(\psi(v) - r(v))),
\end{equation}
where the $(r(v))_{v \in \Lambda}$ are given elements of the torus $\R/\Z$ and $z\ge 0$ is a given activity parameter. In the limit $z\to\infty$, configurations $\psi$ are restricted to satisfying $\psi_v\in\Z + r(v)$ at every vertex, and the effective Hamiltonian on this restricted configuration space consists only of the first sum in~\eqref{eq:sine-Gordon Hamiltonian}. The case where the $(r(v))$ are (quenched) random is known as the \emph{random-phase sine-Gordon model}.

The random-phase sine-Gordon model has received much attention in the physics literature (see, e.g.,~\cite{CO82, HWA94, LDS07, RLDS12,TD90}) and the following behavior was predicted for the heights $\psi$ in the $z\to\infty$ limit when the $(r(v))$ are uniform and independent: In two dimensions, on a box $\Lambda_L$ with zero boundary conditions, the heights are predicted to delocalize with $\log L$ variance at high temperature (rough phase), but with $\log^2 L$ variance at low temperature (super-rough phase); thus the fluctuations of the model are expected to \emph{decrease} as the temperature rises! It is further predicted that the heights delocalize with logarithmic variance in three dimensions and are localized when $d\ge 5$ (see e.g.~\cite{GLD95, N90SC, OS95, VF84}). These predictions appear to be open in the mathematical literature, apart from the following recent result of Garban and Sep\'ulveda~\cite{GS20} on the sine-Gordon model: At sufficiently high temperature, for any deterministic choice of $(r(v))$ and any activity $z\in [0,\infty]$, the heights $\psi$ on a two-dimensional box $\Lambda_L$ with zero boundary conditions delocalize with variance at least $\log L$.

Let us now describe the connection between the random-field integer-valued Gaussian free field and the random-phase sine-Gordon model. To this end, observe that the Hamiltonian of the random-field integer-valued Gaussian free field can be written, up to the addition of a constant factor depending only on $d, \lambda$ and the quenched disorder $\eta$, as
\begin{equation}\label{eq:Hamiltonian RFIVGFF}
\frac{1}{2}\sum_{e\in E(\Lambda^+)} (\nabla \phi(e) - \lambda \nabla u_{\Lambda, \eta}(e))^2
\end{equation}
where we recall that $u_{\Lambda, \eta}$, the ground state of the real-valued random-field Gaussian free field, is defined in~\eqref{def.ueta}. Let $\phi$ be sampled from the integer-valued random-field Gaussian free field and define $\psi := \phi - \lambda u_{\Lambda, \eta}$. Observe that $\psi(v)\in\Z + r(v)$ at every vertex where $r(v) := - \proj \lambda u_{\Lambda,\eta}(v)$ and where $\proj$ denotes the canonical projection from $\R$ to $\R/\Z$. It follows that $\psi$ is distributed as the random-phase sine-Gordon model in the $z\to\infty$ limit with this choice of $r(v)$ (noting that $r$ is a function of the quenched disorder $\eta$).

The distribution of the vector $r$ in the above connection is not uniform and independent, though for large $\lambda$ we expect it to be `somewhat close' to such a distribution, at least when the $(\eta(x))$ are independent standard Gaussian random variables. It is thus suggested that, for large $\lambda$, the distribution of $\psi$ is close to that of the random-phase sine-Gordon model with uniform and independent $r$. This supports the above conjectures regarding the height fluctuations of the integer-valued random-field Gaussian free field in the strong disorder regime.

Further support to the possibility of a roughening transition in dimensions $d=3,4$ is lent by the following observation, which does not rely on the predictions for the behavior of the random-phase sine-Gordon model. Suppose the Hamiltonian~\eqref{eq:Hamiltonian RFIVGFF} is modified by replacing the function $u_{\Lambda,\eta}$ with the function $v_{\Lambda, \eta}:=u_{\Lambda, \eta} + \frac{1}{\lambda} p$ where $p:\Lambda^+\to\R$ is another quenched disorder, consisting of independent random variables which are uniform on $[-1/2, 1/2]$, independently of $\eta$. We claim that the function $\varphi:\Lambda^+\to\Z$ sampled from this modified Hamiltonian is delocalized in dimensions $d=3,4$ at all temperatures and disorder strengths. However, we would like to think of $\frac{1}{\lambda} \nabla p$ as a `small perturbation' of $\nabla u_{\Lambda, \eta}$ (noting that~\eqref{eq:Hamiltonian RFIVGFF} depends on $u_{\Lambda, \eta}$ only through its gradient) and this is more reasonable at strong disorder, when $\lambda$ is large, since $\nabla u_{\Lambda,\eta}$ is localized in dimensions $d=3,4$ (Proposition~\ref{prop.propueta}); thus, consideration of the modified Hamiltonian lends support to delocalization of the random-field integer-valued Gaussian free field in dimensions $d=3,4$ in the strong disorder regime. The main observation leading to the delocalization claim is that the random phase $r(v) := -\lambda \proj v_{\Lambda,\eta}(v) = -p(v)-\lambda \proj u_{\Lambda,\eta}(v)$ is distributed uniformly on the torus $\R/\Z$, independently between vertices, and that this remains its distribution even after conditioning on $\eta$ (i.e., for every fixed realization of $u_{\Lambda, \eta}$, the addition of $p$ to the projection uniformizes the random phase $r$). It follows, as in the above discussion, that the integer-valued surface $\varphi$ which follows the modified Hamiltonian is a sum of two \emph{independent} contributions: The integer part of $\lambda u_{\Lambda,\eta}$ and a sample of the random-phase sine-Gordon model with uniform phases $r$. Thus, regardless of the fluctuations of the random-phase sine-Gordon model, the fluctuations of $\varphi$ are at least as large as the fluctuations of the integer part of $\lambda u_{\Lambda,\eta}$ and thus $\varphi$ is delocalized in dimensions $d=3,4$.

\subsection{The real-valued random-field Gaussian free field and its relation with the membrane model}\label{sec:random-field Gaussian free field}

The simplest among the random-field $\nabla \phi$-models is the random-field Gaussian free field (when the interaction potential is $V = \frac{1}{2}x^2$) with the random field $\eta$ taken to be independent standard Gaussians. The Gaussian nature of this case makes it amenable to exact calculations and, as we will now discuss, relates it to the more familiar \emph{membrane model} (see~\cite{Sa2003}, and~\cite{Sch21} for background and recent results). In the physics literature~\cite[Section 5.1]{For91} such a model (in the continuum) is termed the \emph{random-rod model}. Precisely, for a finite $\Lambda\subset\Z^d$ and independent standard Gaussians $\eta:\Lambda\to\R$, the Hamiltonian $H_\Lambda^{\eta,\mathrm{GFF}}$ of the random-field Gaussian free field associates to each $\phi:\Lambda^+\to\R$ the energy
\begin{equation} \label{eq:101071608bis}
    H_\Lambda^{\eta,\mathrm{GFF}} \left( \phi \right) := \sum_{e\in E(\Lambda^+)} (\nabla\phi(e))^2 - \lambda \sum_{x \in \Lambda} \eta(x) \phi(x).
\end{equation}
For simplicity, we study the model with zero boundary conditions, i.e., $\phi\equiv 0$ on $\partial\Lambda$ and $\lambda = 1$. Recall the function $u_{\Lambda, \eta}$ defined by~\eqref{def.ueta} and note that it is precisely the ground state of random-field Gaussian free field (i.e., the minimizer of the Hamiltonian~\eqref{eq:101071608bis}). Moreover, the quadratic nature of the Hamiltonian implies that when $\varphi$ is sampled with the Hamiltonian~\eqref{eq:101071608bis} at any positive temperature then the difference $\varphi - u_{\Lambda, \eta}$ is distributed as a Gaussian free field (with zero external field) at that temperature, independently of $\eta$. Thus the study of the random-field Gaussian free field decouples to the separate studies of the ground state $u_{\Lambda, \eta}$ and of the thermal fluctuations given by the Gaussian free field. It turns out that the fluctuations of $u_{\Lambda, \eta}$ dominate the behavior, so we focus solely on it in the following discussion.

The ground state $u_{\Lambda, \eta}$ has a joint Gaussian distribution, with covariance matrix given by
\begin{equation} \label{eq:16111910}
    \forall x,y \in \Lambda, \hspace{3mm} \E \left[ u_{\Lambda, \eta} (x)  u_{\Lambda, \eta} (y)\right] = \sum_{z \in \Lambda} G_{\Lambda} (x, z) G_{\Lambda} (z, y) =: v(x,y),
\end{equation}
and it is simple to check that the right-hand side solves the discrete biharmonic equation
\begin{equation} \label{eq:16121910bis}
    \left\{ \begin{aligned}
    \Delta^2 v(\cdot , y) = \delta_y &~\mbox{in}~\Lambda, \\
    v(\cdot , y) = 0 &~\mbox{on}~ \partial \Lambda, \\
     \Delta  v(\cdot , y) = 0 &~\mbox{on}~ \partial \Lambda.
    \end{aligned} \right.
\end{equation}
This covariance function coincides with that of the membrane model, and we conclude that $u_{\Lambda, \eta}$ is equal to it in distribution (to be more precise, equation~\eqref{eq:16121910bis} uses one of the two standard normalizations for the membrane model. In the other normalization, the requirement $\Delta  v(\cdot , y) = 0$ on $\partial \Lambda$ is replaced by the requirement that $v(\cdot,y)=0$ on a second boundary layer. See~\cite[Section 2]{K09} for a discussion of these two possibilities). This observation allows to directly apply the established results for the membrane model to the study of the random-field Gaussian free field, including precise fluctuation results, behavior of the maximum and entropic repulsion~\cite{AKW16, BCK17, CD08, K07, K09}. Of special interest is the scaling limit of the membrane model: The limit, a continuum membrane model, was identified in~\cite{cipriani2019scaling} for the other standard boundary condition. For the boundary condition~\eqref{eq:16121910bis}, by taking suitable limits of the lattice Green's function, the following scaling limit result is expected in dimensions $d\ge 1$: For smooth functions $f$ compactly supported in the continuum box~$\left( -1 , 1\right)^d$, 
\begin{equation}\label{eq:scaling limit random field GFF}
    L^{-\frac d2 -2} \sum_{x \in \Lambda_L} u_{\Lambda_L, \eta}(x) f \left( \frac{x}{L} \right) \overset{\mathrm{(dist)}}{\underset{L \to \infty}{\longrightarrow}} N \left( 0, \int_{\left[ -1 , 1 \right]^d} f(y) w_f (y) \, dy \right),
\end{equation}
where the mapping $w_f$ is the solution of the continuous biharmonic equation
\begin{equation*}
    \left\{ \begin{aligned}
    \Delta^2 w_f &= f ~\mbox{in}~\left[ -1 , 1 \right]^d, \\
     w_f &= 0 ~\mbox{on}~ \partial \left[ -1 , 1 \right]^d, \\
     \Delta w_f &= 0 ~\mbox{on}~ \partial \left[ -1 , 1 \right]^d. \\
    \end{aligned} \right.
\end{equation*}

\subsection{Thermodynamic and scaling limits of the random-field $\nabla \phi$-model.} Theorem~\ref{prop3.1009100bis} establishes quantitative estimates on the height of the random-field $\nabla \phi$-model. A natural direction for further research is the study of the infinite-volume limits of the random surface and its gradient. As mentioned, Cotar and K\"ulske~\cite{CK12} proved the existence of translation-covariant \emph{gradient} Gibbs measures, and their uniqueness for a given tilt~\cite{CK15}, in dimensions $d\ge3$. The first author~\cite{Da21} proved the convergence along the thermodynamic limit for the random surface in dimensions $d\ge 5$ and for its gradient in dimensions $d\ge 4$; convergence of the gradient is also expected in dimension $d=3$, but remains unproven (the existence and uniqueness of infinite-volume translation-covariant gradient Gibbs measure with a specified tilt was established in~\cite{CK12, CK15} in any dimension $d \geq 3$).

Beyond infinite-volume Gibbs measures, it is also natural to study the scaling limit of the random-field $\nabla \phi$-model. We expect universality of the scaling limit, so that the convergence to the continuum membrane model discussed in~\eqref{eq:scaling limit random field GFF} for the random-field Gaussian free field should continue to hold for the uniformly convex random-field random surface model~\eqref{eq:defmuLeta}.

\medskip

\subsection{Relation between between the ground state of the $\nabla \phi$ model and its expectation}\label{sec:thermal expectation for grad phi model}

If we let $\phi$ be a random-field Gaussian free field with quenched disorder $\eta$, then we have the identity, for any side length $L \geq 2$ and any vertex $x \in \Lambda_L$, \begin{equation*}
    \left\langle \phi(x) \right\rangle_{\mu_{\Lambda_L}^\eta} = u_{\Lambda_L, \eta}(x),
\end{equation*}
that is, the thermal expectation of the surface is equal to the ground-state of the Hamiltonian associated with the random-field Gaussian free field. We mention the following question: What is the relation, if any, between the ground state of the $\nabla \phi$-model and its thermal expectation?

\medskip

\subsection{The random-field $\nabla \phi$-model with general potentials $V$}\label{sec:general potentials discussion} To what extent are the fluctuations of the random-field $\nabla \phi$-model affected by the specific choice of the potential $V$? Theorem~\ref{prop3.100910} and Theorem~\ref{prop3.1009100bis} show that the order of magnitude in $L$ of the gradient and height fluctuations does not depend on $V$, as long as the second derivative of $V$ is uniformly bounded from zero and infinity (assumption~\eqref{eq:V ellipticity}). How generic is this phenomenon? In the \emph{non-disordered} case, universality of the order of magnitude has been shown for a wide class of potentials, including the family $V(x) = |x|^p$ with $p>1$ (see~\cite{magazinov2020concentration,MP15} and references within). We observe here that, in the disordered case, such wide universality fails in one dimension. The situation in higher dimensions is not clear.

Let $d=1$ and consider the model~\eqref{eq:real-valued model with temperature} with potential $V(x) = |x|^p$ when the disorder $\eta$ consists of, say, independent standard Gaussian random variables. We take $p>1$, as the model is not well defined when $p\le 1$ (its partition function on $\Lambda_L$ is infinite with positive probability). When $p=2$, the ground state $u_{\Lambda_L,\eta}$ satisfies $\E \left[u_{\Lambda_L,\eta}(x)^2 \right] \leq C L^3$ for all $x\in\Lambda_L$ by Proposition~\ref{prop.propueta}. Now, for general $p>1$, observe that, on the one hand, any random function $u_\eta:\Lambda_L^+\to\R$ with $u_\eta(x)=0$ on $\partial\Lambda_L$ which  satisfies
\begin{equation}\label{eq:disordered ground state bound}
    \E \left[u_{\eta}(x)^2 \right] \leq C L^3
\end{equation}
for all $x\in\Lambda_L$ also satisfies the energy estimate
\begin{equation}
\begin{split}
   \E\left[H_{\Lambda_L}^\eta(u_\eta)\right] &= \E\left[\sum_{e\in E(\Lambda_L^+)} |\nabla u_\eta(e)|^p - \lambda \sum_{x \in \Lambda_L} \eta(x) u_\eta(x)\right]\\
   &\ge -\lambda\, \E\left[\sum_{x \in \Lambda_L} \eta(x) u_\eta(x)\right]\ge -\lambda \sum_{x\in\Lambda_L}\sqrt{\E\left[\eta(x)^2\right] \E\left[u_\eta(x)^2\right]}\ge -C \lambda L^{5/2}.
\end{split}
\end{equation}
On the other hand, for $a>0$, the function $v_\eta:\Lambda_L^+\to\R$ defined by $v_\eta(x):=0$ on $\partial\Lambda_L$ and $v_\eta(x):=\sign(\hat\eta)L^a$ for $x\in\Lambda_L$, where $\hat\eta:=\frac{1}{|\Lambda_L|}\sum_{x\in\Lambda_L} \eta(x)$, satisfies
\begin{equation}
   H_{\Lambda_L}^\eta(v_\eta) = \sum_{e\in E(\Lambda_L^+)} |\nabla v_\eta(e)|^p - \lambda \sum_{x \in \Lambda_L} \eta(x) v_\eta(x) = 2L^{ap} - \lambda |\Lambda_L|\cdot|\hat\eta|\cdot L^{a}.
\end{equation}
Since $|\hat\eta|$ is of order $L^{-1/2}$ with high probability, we conclude that $H_{\Lambda_L}^\eta(v_\eta) \le -c \lambda L^{a+1/2}$ with high probability when $p<1 + \frac{1}{2a}$ and $L$ is sufficiently large. Thus, taking any $a>2$, we see that, at least when $1<p<1+\frac{1}{2a}$, the ground state of the disordered model cannot satisfy the bound~\eqref{eq:disordered ground state bound} for large $L$; it thus exhibits different height fluctuations than the case $p=2$.

We also point out that the fluctuations of the gradient of the ground state in dimensions $d\ge 3$ are of order one for the class of potentials $V(x)=|x|^p$, $p>1$. Indeed, the ground state in the domain $\Lambda$ with zero boundary conditions minimizes (using~\eqref{def.ueta},  discrete integration by parts~\eqref{eq:discrete integration by parts} and H\"older's inequality with $\frac{1}{p}+\frac{1}{q}=1$)
\begin{equation}
\begin{split}
    H_\Lambda^\eta \left( \phi \right) &= \sum_{e\in E(\Lambda^+)} |\nabla\phi(e)|^p - \lambda \sum_{x \in \Lambda} \eta(x) \phi(x) = \sum_{e\in E(\Lambda^+)} |\nabla\phi(e)|^p - \lambda\nabla u_{\Lambda, \eta}(e) \nabla\phi(e)\\
    &\ge \|\nabla\phi\|_{L^p(\Lambda^+)}^p - \lambda\|\nabla u_{\Lambda,\eta}\|_{L^q(\Lambda^+)} \|\nabla\phi\|_{L^p(\Lambda^+)},
\end{split}
\end{equation}
where we extended the norm notation~\eqref{eq:11151509} to $L^p$ in the standard way. Since the energy of the zero function is zero, we conclude that the ground state $\phi_{\Lambda,\eta}$ satisfies
\begin{equation}
    \|\nabla\phi_{\Lambda,\eta}\|_{L^p(\Lambda^+)}^p\le \lambda^q \|\nabla u_{\Lambda,\eta}\|_{L^q(\Lambda^+)}^q.
\end{equation}
Using this inequality with the proof of Proposition~\ref{prop.propueta}, we see that $\sup_L \E\left[\|\nabla\phi_{\Lambda_L,\eta}\|_{\underline{L}^p(\Lambda_L^+)}^p\right]<\infty$ in dimensions $d\ge 3$ (extending again the norm notation~\eqref{eq:11151509}), for a class of disorder distributions which includes the case that $\eta$ consists of independent standard Gaussian random variables.

\medskip

\subsection{Dynamical random-field random surfaces}
One may naturally form a dynamics on random-field random surfaces by allowing the random field to evolve in time. When the random field consists of independent standard Gaussians, it is natural to let it evolve via independent Ornstein-Uhlenbeck processes. Note that this gives, in particular, a natural dynamics for the membrane model using its representation as the ground state of the random-field Gaussian free field (see Section~\ref{sec:random-field Gaussian free field}).

\subsection{The random-field membrane model.} \label{RFMM} In this paper, we considered the effect that a random field can have on the fluctuations of random surfaces of the $\nabla\phi$ model. The same type of disorder may also be applied to other random surface models. Here, we briefly discuss its effect on the membrane model. The membrane model (in the absence of disorder) is known to fluctuate more strongly than the $\nabla\phi$ model, delocalizing in dimensions $1\le d\le 4$ and localizing in dimensions $d\ge 5$. We will now see that the added disorder causes the surface to delocalize in all dimensions $d\le 8$.

In this setting, given $L \geq 0$ and an external field $\eta : \Lambda_L \to \R$, the disordered Hamiltonian of a finite-volume $\phi : \Lambda_L^+ \to \R$ normalized to $0$ on $\partial \Lambda_L$ is given by the formula
\begin{equation} \label{def.hetaDelta}
    H^{\eta, \Delta}_{\Lambda_L}(\phi) := \frac 12 \left\| \Delta \phi\right\|_{L^2 \left( \Lambda_L \right)}^2 - \sum_{x \in \Lambda_L} \eta(x) \phi(x).
\end{equation}
The random-field membrane model is the probability distribution
\begin{equation} \label{eq:11351710}
    \mu^{\eta, \Delta}_{\Lambda_L}(d \phi) := \frac{1}{Z^{\eta, \Delta}_{\Lambda_L}} \exp\left( - H^{\eta, \Delta}_{\Lambda_L}(\phi) \right),
\end{equation}
where the partition function $Z^{\eta, \Delta}_{\Lambda_L}$ is the constant which makes the measure~\eqref{eq:11351710} a probability distribution. Let us first consider the ground state of the model, that is, the interface $v_{\Lambda, \eta} : \Lambda_L \to \R$ which minimizes the variational problem
\begin{equation*}
    \inf_{\substack{ w : \Zd \to \R \\ w \equiv 0 \,\mathrm{on}\, \partial \Lambda_L}} ~  \frac 12\left\| \Delta w \right\|_{L^2 \left( \Lambda_L \right)}^2 - \sum_{x \in \Lambda_L} \eta(x) w(x).
\end{equation*}
Note that it can can be equivalently defined as the solution of the biharmonic equation
\begin{equation*}
    \left\{ \begin{aligned}
    \Delta^2 v_{\Lambda_L , \eta} = \eta &~\mbox{in}~\Lambda_L, \\
   v_{\Lambda_L , \eta} = 0 &~\mbox{on}~ \partial \Lambda_L, \\
     \Delta  v_{\Lambda_L , \eta} = 0 &~\mbox{on}~ \partial \Lambda_L.
    \end{aligned} \right.
\end{equation*}
In that case, the interface $v_{\Lambda_L, \eta}$ is a linear functional of $\eta$, and is given by the explicit formula
\begin{equation*}
    v_{\Lambda_L, \eta}(x) = \sum_{y , z \in \Lambda_L} G_{\Lambda_L}(x , y)  G_{\Lambda_L}(y , z) \eta(z)
\end{equation*}
(analogously to the discussion in Section~\ref{sec:random-field Gaussian free field}, the ground state $v_{\Lambda_L, \eta}$ has the distribution of a non-disordered random surface model with Hamiltonian proportional to $\left\| \Delta^2 \phi\right\|_{L^2 \left( \Lambda_L \right)}^2$. Such a construction may also be applied to higher powers of the Laplacian, and their combination with the gradient operator).
Using the upper and lower bounds on the Green's function stated in~\eqref{eq:10281908} and~\eqref{eq:10291908}, we obtain that the ground state satisfies the following estimates
 \begin{alignat*}{3}
    &1 \leq d \leq 7:\qquad&&&c L^{8-d} \leq~ &\E \left[  v_{\Lambda_L, \eta}(0)^2 \right] \leq C L^{8-d},\\
    &d=8:&&&c \ln L \leq~ &\E \left[ v_{\Lambda_L, \eta}(0)^2 \right] \leq C \ln L,\\
    &d\ge9:&&& c \leq~ &\E \left[ v_{\Lambda_L, \eta}(0)^2 \right] \leq C.
\end{alignat*}
It is thus delocalized in dimensions $d \leq 8$ and localized in dimensions $d \geq 9.$

To study the probability distribution~\eqref{eq:11351710}, let us observe that if $\phi : \Lambda_L^+ \to \R$ is a random surface distributed according to~\eqref{eq:11351710}, then $\psi := \phi - v_{\Lambda_L, \eta}$ is a membrane model with external field set to $0$. The field $\psi$ has thus a multivariate Gaussian distribution, its mean vector is equal to $0$ and its covariance matrix is given by, for any $x , y \in \Lambda_L$,
\begin{equation*}
    \mathrm{Cov} \left[ \psi(x) , \psi(y) \right] = \sum_{ z \in \Lambda_L} G_{\Lambda_L}(x , z)  G_{\Lambda_L}(y , z).
\end{equation*}
Two consequences can be deduced from this observation: First, the expectation of the random variable $\psi$ is equal to $0$, and second the fluctuations of $\psi$ at the center of the box can be explicitly quantified using the upper and lower bounds on the Green's function stated in~\eqref{eq:10281908} and~\eqref{eq:10291908}. We have, for any realization of the random field $\eta$,
\begin{alignat*}{3}
    &1 \leq d \leq 3:\qquad&&&c L^{3-d} \leq~ &\left\langle \psi(0)^2 \right\rangle_{\mu^{\eta, \Delta}_{\Lambda_L}} \leq C L^{3-d},\\
    &d=4:&&&c \ln L \leq~ &\left\langle \psi(0)^2 \right\rangle_{\mu^{\eta, \Delta}_{\Lambda_L}} \leq C \ln L,\\
    &d\ge5:&&& c \leq~ & \left\langle \psi(0)^2 \right\rangle_{\mu^{\eta, \Delta}_{\Lambda_L}}  \leq C.
\end{alignat*}
A combination of the two previous sets of estimates shows that the random-field membrane model satisfies the inequalities
 \begin{alignat*}{3}
    &1 \leq d \leq 7:\qquad&&&c L^{8-d} \leq~ &\E \left[  \left\langle \phi(0)^2 \right\rangle_{\mu^{\eta, \Delta}_{\Lambda_L}} \right] \leq C L^{8-d},\\
    &d=8:&&&c \ln L \leq~ &\E \left[ \left\langle \phi(0)^2 \right\rangle_{\mu^{\eta, \Delta}_{\Lambda_L}} \right] \leq C \ln L,\\
    &d\ge9:&&& c \leq~ &\E \left[ \left\langle \phi(0)^2 \right\rangle_{\mu^{\eta, \Delta}_{\Lambda_L}} \right] \leq C.
\end{alignat*}
The model is delocalized in dimensions $d \leq 8$ and localized in dimensions $d \geq 9$.

We finally mention that it would be possible to consider models with higher-order symmetries by, for instance, replacing the Laplacian in the definition~\eqref{def.hetaDelta} by one of its iteration. We expect that this modification would increase the critical dimension above which the system is localized and below which it is delocalized, and refer to~\cite[Section 9]{DHP20++} for a related discussion about the effect of higher-order symmetries on disordered spin systems with compact spin space.

\medskip

\subsection{Random-field Gaussian free field with $1$-dependent external field.} Theorem~\ref{prop3.1009100bis} states that, if the components of the random field $\eta$ are independent, then the random-field $\nabla \phi$-model is delocalized in dimensions $d \leq 4$ and localized in dimensions $d \geq 5$. One can thus raise the question whether the independence assumption on the random field can be relaxed without changing the critical dimension for localization. In this section, we show that the independence assumption cannot generally be relaxed to $1$-dependence (i.e., independence at distance $2$). Specifically, we present a model of a random-field random surface in which the external random field is $1$-dependent, which localizes in dimensions $d \geq 3$ and thus exhibits a different qualitative behaviour.

Let us fix an integer $L \geq 1$ and let $\zeta : \Lambda_L \to \R$ a discrete Gaussian free field in the box $\Lambda_L$ with Dirichlet boundary condition. We then define the external field $\eta$ by the formula, for any vertex $x \in \Zd$,
\begin{equation} \label{eq:18481610}
    \eta(x) := -\Delta \zeta(x).
\end{equation}
The law of the external field $\eta$ is Gaussian, let us verify that it is $1$-dependent. An explicit computation shows, for any points $x , z  \in \Lambda_L$,
\begin{equation} \label{eq:1610187}
    \E \left[ \eta(x) \zeta(z)  \right]  = \E \left[ -\Delta \zeta(x)  \zeta(z)  \right]  = -\Delta G_{\Lambda_L} \left( x , z \right)  = \indc_{\{ x = z\} }.
\end{equation}
Using the identity~\eqref{eq:1610187} and an explicit computation, we further deduce, for any $x , y \in \Lambda_L$
\begin{equation} \label{eq:18431610}
     \E \left[ \eta(x) \eta(y)  \right]  =  2d \indc_{\{x = y\}} - \sum_{z \sim y} \indc_{\{ x = z\}}. \notag
\end{equation}
If the vertices $x$ and $y$ are distinct and non-adjacent, then all the terms in the right-hand side are equal to $0$ and thus $\E \left[ \eta(x) \eta(y)  \right] =0$. Since the field $\eta$ is Gaussian, we deduce that the random variables $ \eta(x)$ and $\eta(y)$ are independent.

Let us now consider the random-field Gaussian free field $\phi$ with the external field $\eta$ given by~\eqref{eq:18481610}. Since the random surface $\phi - \zeta$ is a Gaussian free field with external field set to $0$, independently of $\zeta$, we have,
\begin{equation} \label{eq:19281610}
    \left\langle \left( \phi(0) - \zeta(0) \right)^2 \right\rangle_{\mu^{\eta, \mathrm{GFF}}_{\Lambda_L}} \approx \left\{ \begin{aligned}
     L &~\mbox{if}~ d=1, \\
     \ln L  &~\mbox{if}~ d=2, \\
    1  &~\mbox{if}~ d \geq 3,
    \end{aligned} \right.
\end{equation}
where $a\approx b$ is used here in the sense $c\cdot a\le b\le C\cdot a$.
From the inequality~\eqref{eq:19281610} and the fact that $\zeta$ is a Gaussian free field with Dirichlet boundary condition in the box $\Lambda_L$, we obtain
\begin{equation*}
    \E \left[ \left\langle\phi(0)^2 \right\rangle_{\mu^{\eta, \mathrm{GFF}}_{\Lambda_L}} \right] \approx \left\{ \begin{aligned}
     L &~\mbox{if}~ d=1, \\
     \ln L  &~\mbox{if}~ d=2, \\
    1  &~\mbox{if}~ d \geq 3.
    \end{aligned} \right.
\end{equation*}
The random surface is thus localized in dimensions $3$ and higher.

\appendix

\section{Nash--Aronson estimates for the Dirichlet problem} \label{appA}

In this section, we prove the Nash--Aronson estimate in finite volume stated in Proposition~\ref{prop.NashAronson}. The proof builds upon the infinite-volume result of~\cite[Appendix B]{GOS} stated below (which itself builds upon the infinite-volume and continuous estimate of Aronson~\cite{Ar}). We first introduce the infinite-volume heat kernel and state the Nash--Aronson estimate for discrete, time-dependent and uniformly elliptic environment of Giacomin--Olla--Spohn~\cite{GOS}.

\begin{definition} \label{def.defA1}
Let $s_0 \in \R$. For each continuous, time-dependent, uniformly elliptic environment $\a : [s_0 , \infty) \times E(\Zd) \to [c_- , c_+]$, each initial time $s \in [s_0 , \infty )$, and each point $y \in \Lambda_L$, we introduce the infinite-volume heat kernel $P_{\a, \infty}$ to be the solution of the parabolic equation
\begin{equation*}
    \left\{ \begin{aligned}
    \partial_t P_{\a, \infty} (t , x ; s  ,y) - \nabla \cdot \a \nabla P_{\a, \infty} (t , x ; s  ,y) &= 0 &&(t , x) \in (s , \infty) \times \Zd, \\
    P_{\a, \infty}(s , x ; s  ,y) &= \indc_{\{ x = y \}} &&x \in \Zd.
    \end{aligned} \right.
\end{equation*}
\end{definition}
The next proposition establishes lower and upper bounds on the map $P_{\a, \infty}.$

\begin{proposition}[Nash--Aronson estimates, Propositions B.3 and B.4 of~\cite{GOS}] \label{prop.NashAronsoninfintevol}
There exist constants $C,c$ depending on the dimension $d$ and the ellipticity parameters $c_-, c_+$ such that, for any pair of times $s , t \in (s_0 , \infty)$ with $t \geq s$ and any pair of points $x , y \in \Zd$,
\begin{equation} \label{eq:upboundNasinfvol}
     P_{\a, \infty} \left(t , x ; s  ,y \right) \leq  \frac{C}{1 \vee (t-s)^{\frac d2}} \exp \left( - \frac{c|x - y|}{ 1 \vee \sqrt{t- s}} \right).
\end{equation}
Under the additional assumption $|x - y| \leq \sqrt{t - s} $, one has the lower bound
\begin{equation} \label{eq:lowboundNasinfvol}
     P_{\a, \infty} \left(t , x ; s  ,y \right) \geq \frac{c}{1 \vee (t-s)^{\frac d2}}.
\end{equation}
\end{proposition}

\begin{remark}
The article of Giacomin--Olla--Spohn~\cite{GOS} only establishes the lower bound of the Nash--Aronson estimate in the on-diagonal case (i.e., under the assumption $|x - y| \leq \sqrt{t - s}$). While it would be possible to obtain off-diagonal lower bounds (in the case $|x - y| \geq \sqrt{t - s}$) by adapting the techniques of~\cite{Ar, delmotte1999parabolic} (written in either the continuous setting or the discrete setting with a static environment), they are not necessary in the article~\cite{GOS} or in the proof of Theorem~\ref{prop3.1009100bis}.
\end{remark}

\begin{remark}
Under the assumption $|x - y| \leq \sqrt{t - s}$, the ratio $|x - y| / (1 \vee \sqrt{t- s})$ is smaller than $1$ and the right-hand sides of~\eqref{eq:upboundNasinfvol} and~\eqref{eq:lowboundNasinfvol} are of comparable sizes. The estimates are thus sharp up to multiplicative constants. 
\end{remark}

The proof of Proposition~\ref{prop.NashAronsoninfintevol} given below relies on analytic arguments. We mention that a more probabilistic approach, relying on the introduction of the random-walk whose generator is the operator $- \nabla \cdot \a \nabla$ (following~\cite[Section 3.2]{GOS}) and on stopping time arguments, would yield the same result.

\begin{proof}[Proof of Proposition~\ref{prop.NashAronson}]

First let us note that by the change of variable $(t - s) \to (t-s)/c_-$, it is sufficient to prove the result when  $c_- = 1$. Let us fix an integer $L\ge 0$, let $ c_+ \in [1 , \infty)$ be an ellipticity constant, and let $s_0\in\R$. Let $\a : [s_0 , \infty) \times E(\Lambda_L^+) \to [1 , c_+]$ be a continuous time-dependent (uniformly elliptic) environment. For any $s \in [s_0 , \infty)$ and any vertex $y \in \Lambda_L$, we denote by $P_\a(\cdot , \cdot ; s, x)$ the solution of the parabolic equation~\eqref{eq:16041308}.
We extend the environment $\a$ to the space $[s_0 , \infty) \times E(\Zd)$ by setting $\a(t , e) = c_+$ for any pair $(t , e) \in [s_0 , \infty) \times \left( E(\Zd) \setminus E\left(\Lambda_L^+ \right) \right)$, and let $P_{\a , \infty}$ be the infinite volume heat kernel associated with the extended environment $\a$ as defined in Definition~\ref{def.defA1}. We prove the upper and lower bounds of Proposition~\ref{prop.NashAronson} separately. We will make use of the notation
\begin{equation*}
    \left\| P_\a \left( t , \cdot ; s , y \right) \right\|^2_{L^2 \left( \Lambda_L \right)} := \sum_{x \in \Lambda_L} P_\a \left( t , x ; s , y \right)^2 \hspace{5mm} \mbox{and} \hspace{5mm} \left\| P_\a \left( t , x ; s , \cdot \right) \right\|^2_{L^2 \left( \Lambda_L \right)} := \sum_{y \in \Lambda_L} P_\a \left( t , x ; s , y \right)^2
\end{equation*}
as well as, for any directed edge $e = (x , z) \in \vec{E} \left(\Lambda_L^+ \right)$,
$$\nabla P_\a \left( t , e ; s , y \right) = P_\a \left( t , x ; s , y \right) - P_\a \left( t , z ; s , y \right)$$
and, following the conventions of Section~\ref{Fctandevts},
\begin{equation*}
    \left\| \nabla P_\a \left( t , \cdot ; s , y \right) \right\|^2_{L^2 \left( \Lambda_L^+ \right)} = \sum_{e \in E \left( \Lambda_L^+ \right)} \left(\nabla P_\a \left( t , e ; s , y \right)\right)^2.
\end{equation*}

\smallskip

\textbf{Proof of the upper bound.}
We first note that the estimate~\eqref{Nash.supprop3.3} is equivalent to the two following inequalities: There exist constants $c, C$ depending on $d , c_+ $ such that
\begin{equation} \label{eq:13252909}
    \left\{ \begin{aligned}
    P_{\a}\left( t , x ; s , y   \right) &\leq \frac{C}{1 \vee (t-s)^{\frac d2}} \exp \left( - \frac{c|x - y|}{ 1 \vee \sqrt{t- s}} \right) &~\mbox{if}~ (t - s) \leq L^2, \\
    P_{\a}\left( t , x ; s , y   \right) &\leq \frac{C}{1 \vee (t - s)^\frac d2} \exp \left( - \frac{c(t - s)}{L^2} \right) &~\mbox{if}~ (t - s) \geq L^2.
    \end{aligned} \right.
\end{equation}
We first treat the case $(t - s) \leq L^2$. By the maximum principle for the parabolic operator $\partial_t - \nabla \cdot \a \nabla$, one has the estimate, for any $t, s \in [ s_0 , \infty)$ with $t \geq s$ and any $x , y \in \Lambda_L$,
\begin{equation} \label{eq:17202908}
    P_\a \left( t , x ; s , y   \right) \leq P_{\a , \infty} \left( t , x ; s , y \right).
\end{equation}
Combining the inequality~\eqref{eq:17202908} with Proposition~\ref{prop.NashAronsoninfintevol} yields the upper bound, for any $s , t \in (s_0 , \infty)$ with $t \geq s$ and any $x , y \in \Zd$,
\begin{equation} \label{eq:17342809}
    P_\a \left( t , x ; s , y   \right) \leq \frac{C}{1 \vee (t-s)^{\frac d2}} \exp \left( - \frac{c|x - y|}{ 1 \vee \sqrt{t- s}} \right).
\end{equation}
This is~\eqref{eq:13252909} in the case $t \leq L^2$. We now focus on the case $t \geq L^2$. To this end, we denote by $C_{\mathrm{Poinc}}$ the constant which appears in the Poincar\'e inequality, that is, the smallest constant which satisfies $\left\| u \right\|_{L^2 \left( \Lambda_\ell \right)}^2 \leq C_\mathrm{Poinc} \ell^2 \left\| \nabla u \right\|_{L^2 \left( \Lambda_\ell^+\right)}^2$ for any side length $\ell \in \N$ and any function $u : \Lambda_\ell^+ \to \R$ normalized to be $0$ on the boundary $\partial \Lambda_\ell$.

We first prove that the $L^1(\Lambda_L)$-norm of the heat kernel $P_\a$ decays exponentially fast in the ratio $t/L^2$. Specifically, we prove that there exists a constant $\bar{c}_1 < \infty$ depending only on $d$ and $c_+$ such that
\begin{equation} \label{eq:174419121}
    \sum_{x \in \Lambda_L} P_\a \left( t , x ; s , y \right) \leq e^{ -  \frac{\bar{c}_1 t}{ L^2}}.
\end{equation}
To prove the inequality~\eqref{eq:174419121}, we first note that, by the inequality~\eqref{eq:17342809}, there exists a constant $C_{1/2} < \infty$ depending only on $d$ and $c_+$ such that, for any $t , s \in (0 , \infty)$ satisfying $t-s \geq C_{1/2} L^2$ and any $y \in \Lambda_L$,
\begin{equation} \label{12property}
    \sum_{x \in \Lambda_L} P_\a \left( t , x ; s , y \right) \leq \frac{1}{2}.
\end{equation}
Using the convolution property for the heat kernel $P_\a$, we have the identity, for any $t \geq C_{1/2} L^2$ and any $x \in \Lambda_L$,
\begin{equation*}
    P_\a (t , x ; 0,0) = \sum_{y \in \Lambda_L}  P_\a \left(t , x ; t - C_{1/2} L^2 , y\right) P_\a \left(t - C_{1/2} L^2 , y ; 0 , 0 \right).
\end{equation*}
Summing the previous inequality over the vertices $x \in \Lambda_L$ and using the inequality~\eqref{12property}, we deduce that
\begin{equation*}
\sum_{x \in \Lambda_L} P_\a \left( t , x ; 0 , 0 \right)  \leq \frac{1}{2}  \sum_{y \in \Lambda_L} P_\a \left(t - C_{1/2} L^2 , y ; 0 , 0 \right).
\end{equation*}
We conclude the proof of~\eqref{eq:174419121} by iterating the previous inequality.

We then let $c_1 := \min \left( 1/C_\mathrm{Poinc}, 2\bar{c}_1 \right)$, and note that this constant depends only on the parameter~$d$.

Using that the heat kernel $P_\a$ solves the parabolic equation~\eqref{eq:16041308} and the discrete integration by parts~\eqref{eq:discrete integration by parts}, we have
\begin{align} \label{eq:14472909}
    \partial_t \left( e^{\frac{c_1(t-s)}{L^2}} \left\| P_\a \left( t , \cdot ; s , y \right) \right\|^2_{L^2 \left( \Lambda_L \right)} \right) & =  \frac{c_1 e^{\frac{c_1(t-s)}{L^2}}}{L^2}  \left\| P_\a \left( t , \cdot ; s , y \right) \right\|^2_{L^2 \left( \Lambda_L \right)} \\ & \quad - 2 e^{\frac{c_1(t-s)}{ L^2}}\sum_{e \in E \left(\Lambda_L^+ \right)} \a(t , e) \left( \nabla P_\a \left( t , e ; s , y \right) \right)^2. \notag
\end{align}
Using the lower bound $\a \geq 1$ on the environment, and the Poincar\'e inequality in the box $\Lambda_L$ (which can be applied since the mapping $x \mapsto P_\a(t , x ; s , y )$ is equal to $0$ on $\partial \Lambda_L$), we obtain
\begin{equation} \label{eq:180628099}
    \sum_{e \in E \left(\Lambda_L^+ \right)} \a(t , e) \left( \nabla P_\a \left( t , e ; s , y \right) \right)^2
     \geq \frac{1}{C_{\mathrm{Poinc}} L^2} \left\| P_\a \left( t , \cdot ; s , y \right) \right\|_{L^2 \left( \Lambda_L \right)}^2.
\end{equation}
A combination of~\eqref{eq:14472909} and~\eqref{eq:180628099} with the definition of the constant $c_1$ above yields
\begin{equation} \label{eq:15202909}
    \partial_t \left( e^{ \frac{c_1 (t-s)}{L^2}} \left\| P_\a \left( t , \cdot ; s , y \right) \right\|^2_{L^2 \left( \Lambda_L \right)} \right)  \leq - e^{\frac{c_1 (t-s)}{L^2}} \left\| \nabla P_\a \left( t , \cdot ; s , y \right) \right\|^2_{L^2 \left( \Lambda_L^+ \right)}.
\end{equation}
By the Nash inequality (see~\cite{nash1958continuity}) and the non-negativity of the map $P_\a$, there exists a constant $C_{\mathrm{Nash}}$ depending only on the dimension $d$ such that
\begin{equation} \label{eq:15152909}
     \left\| P_\a \left( t , \cdot ; s , y \right) \right\|_{L^2 \left( \Lambda_L \right)}^{1 + \frac 2d} \leq C_{\mathrm{Nash}} \left( \sum_{x \in \Lambda_L} P_\a \left( t , x ; s , y \right) \right)^\frac 2d  \left\| \nabla P_\a \left( t , \cdot ; s , y \right) \right\|_{L^2 \left( \Lambda_L^+ \right)}.
\end{equation}
Applying the inequality~\eqref{eq:174419121}, we may simplify the inequality~\eqref{eq:15152909} and write
\begin{equation} \label{eq:15212909}
    \left\| P_\a \left( t , \cdot ; s , y \right) \right\|_{L^2 \left( \Lambda_L \right)}^{1 + \frac 2d} \leq C_{\mathrm{Nash}}  e^{-\frac{c_1 (t-s)}{d}}\left\| \nabla P_\a \left( t , \cdot ; s , y \right) \right\|_{L^2 \left( \Lambda_L^+ \right)}.
\end{equation}
Combining~\eqref{eq:15202909} and~\eqref{eq:15212909} yields
\begin{align} \label{eq:16182909}
    \partial_t \left( e^{\frac{c_1(t-s)}{L^2}} \left\| P_\a \left( t , \cdot ; s , y \right) \right\|^2_{L^2 \left( \Lambda_L \right)} \right) & \leq - \frac{1}{C_{\mathrm{Nash}}^2} e^{\frac{c_1(d+2)(t-s)}{dL^2}} \left\| P_\a \left( t , \cdot ; s , y \right) \right\|^{\frac{2(d+2)}{d}}_{L^2 \left( \Lambda_L \right)} \\
    & \leq - \frac{1}{C_{\mathrm{Nash}}^2} \left( e^{\frac{c_1(t-s)}{L^2} } \left\| P_\a \left( t , \cdot ; s , y \right) \right\|^{2}_{L^2 \left( \Lambda_L \right)} \right)^{\frac{d+2}{d}} .\notag
\end{align}
Integrating the differential inequality~\eqref{eq:16182909} and using the identity $\left\| P_\a \left( s , \cdot ; s , y \right) \right\|^2_{L^2 \left( \Lambda_L \right)} = 1$, we obtain that there exists a constant $C$ depending on $d$ and $c_-$ such that, for any $t \geq s$,
\begin{equation} \label{eq:16302909}
    e^{\frac{c_1(t-s)}{L^2}} \left\| P_\a \left( t , \cdot ; s , y \right) \right\|^2_{L^2 \left( \Lambda_L \right)} \leq \frac{C}{1 \vee (t - s)^\frac{d}{2}}.
\end{equation}
We now show that the estimate~\eqref{eq:16302909} implies the inequality~\eqref{eq:13252909} in the case $(t -s) \geq L^2$. Using the convolution property for the heat kernel $P_\a$, we have the identity, for any $t \in [s + L^2 , \infty)$ and any $x, y \in \Lambda_L$,
\begin{equation*}
    P_\a (t , x ; s , y) = \sum_{z \in \Lambda_L}  P_\a \left(t , x ; \frac{t + s}2 , z\right) P_\a \left(\frac{t + s}2 , z ; s , y \right).
\end{equation*}
The Cauchy-Schwarz inequality then yields
\begin{equation} \label{eq:12562909}
    P_\a (t , x ; s , y) \leq \left\| P_\a \left( t , \cdot ; \frac{t + s}2 , y \right) \right\|_{L^2 \left( \Lambda_L \right)} \left\| P_\a \left( \frac{t + s}2 , x ; s , \cdot \right) \right\|_{L^2 \left( \Lambda_L \right)}.
\end{equation}
The first term of the right-hand side of~\eqref{eq:12562909} can be estimated by the upper bound~\eqref{eq:16302909} applied with the initial time $\frac{s + t}{2}$ instead of $s$. We obtain
\begin{equation} \label{eq:29091811}
     \left\| P_\a \left( t , \cdot ; \frac{s + t}{2} , y \right) \right\|^2_{L^2 \left( \Lambda_L \right)} \leq \frac{C e^{-\frac{c_1(t-s)}{2 L^2}} }{1 \vee (t - s)^\frac{d}{2}}.
\end{equation}
To estimate the second term in the right-hand side, we use the estimate~\eqref{eq:17342809}, the assumption $t - s \geq L^2$, the observations that, for any pair of points $x , z \in \Lambda_L$, $|x - z| \leq C L$, and that the cardinality of the box $\Lambda_L$ is equal to $(2L+1)^d$. We obtain
\begin{align} \label{eq:13272909}
    \left\| P_\a \left( \frac{t + s}2 , x ; s , \cdot \right) \right\|_{L^2 \left( \Lambda_L \right)}^2  & \leq \sum_{z \in \Lambda_L} \left(\frac{C}{1 \vee (t-s)^{\frac d2}} \exp \left( - \frac{c|x - y|}{ 1 \vee \sqrt{t- s}}
    \right)\right)^2 \\
    & \leq \sum_{z \in \Lambda_L} \frac{C}{1 \vee (t-s)^{d}} \notag \\
    & \leq \frac{CL^d}{1 \vee (t-s)^{d}} \notag \\
    & \leq \frac{C}{1 \vee (t-s)^{\frac d2}}. \notag
\end{align}
A combination of~\eqref{eq:12562909},~\eqref{eq:29091811} and~\eqref{eq:13272909} completes the proof of~\eqref{eq:13252909} in the case $t \geq L^2$.

\medskip

\textbf{Proof of the lower bound.}
We first claim that there exists a constant $c_0 \in (0 , 1)$ depending on the parameters $d , c_+$ such that, for any $t , s \in [ s_0 , \infty)$ satisfying $\sqrt{t - s} \leq c_0 L$ and any $y \in \Lambda_{L/2}$,
\begin{equation} \label{eq:10183009}
\sup_{(t' , x) \in [s , t] \times \partial \Lambda_{L}} P_{\a, \infty} \left(t' , x ; s , y  \right) \leq \frac12 \inf_{x \in y + \Lambda_{\sqrt{t-s}}} P_{\a, \infty} \left(t , x ; s , y  \right).
\end{equation}
The proof of this inequality relies on Proposition~\ref{prop.NashAronsoninfintevol}. First by the lower bound~\eqref{eq:lowboundNasinfvol}, we have the estimate, for any $t, s \in [s_0 , \infty)$ such that $t - s \leq L^2$,
\begin{equation} \label{eq:15103009}
    \inf_{x \in y + \Lambda_{\sqrt{t-s}}} P_{\a, \infty} \left(t , x ; s , y  \right) \geq \frac{c}{1 \vee (t - s)^\frac d2} \geq \frac{c}{L^d}.
\end{equation}
Let us fix a constant $c_1 \in (0 , 1)$. Using that for any point $x \in \partial \Lambda_L$ and any point $y \in \Lambda_{L/2}$, we have $|x - y| \geq L/2$ together with the upper bound~\eqref{eq:upboundNasinfvol}, we obtain the estimate, for any $L \geq c_1^{-1}$ and any $t, s \in [s_0 , \infty)$ satisfying $\sqrt{t - s} \leq c_1 L$,
\begin{align} \label{eq:15113009}
    \sup_{(t' , x) \in [s , t] \times \partial \Lambda_L} P_{\a, \infty} \left(t' , x ; s , y  \right) & \leq \sup_{(t' , x) \in [s , t] \times \partial \Lambda_L} \frac{C}{1 \vee (t'-s)^{\frac d2}} \exp \left( - \frac{c|x - y|}{ 1 \vee \sqrt{t'- s}} \right) \\
    & \leq \sup_{t' \in [s , t]} \frac{C}{1 \vee (t'-s)^{\frac d2}} \exp \left( - \frac{cL}{2 \left( 1 \vee \sqrt{t'- s} \right)} \right) \notag \\
    & \leq \sup_{t' \in [0, c_1^2 L^2]} \frac{C}{ t'^{d/2}} \exp \left( - \frac{cL}{2\sqrt{t'}} \right) \notag \\
    & \leq \sup_{t' \in [0, c_1^2]} \frac{C}{L^d t'^{d/2}} \exp \left( - \frac{c}{2\sqrt{t'}} \right). \notag
\end{align}
Using that the mapping $t' \mapsto t'^{-d/2} \exp \left( - c/ (2 \sqrt{t'}) \right)$ tends to $0$ as $t'$ tends to $0$, we may select a constant $c_0 \in (0 , 1]$ such that
\begin{equation} \label{eq:15123009}
    \sup_{t' \in [0, c_0^2]} \frac{C}{ t'^{d/2}} \exp \left( - \frac{c}{2  \sqrt{t'}} \right) \leq \frac{c}{2},
\end{equation}
where the constants $c ,C$ are the ones which appear in the right-hand sides of~\eqref{eq:15103009} and~\eqref{eq:15113009}. Let us note that, since the constants $c , C$ depend only on the dimension $d$ and the ellipticity constant $c_+$, the constant $c_0$ may be chosen so that it depends only on $d , c_+$. Multiplying both sides of the inequality~\eqref{eq:15123009} by $L^{-d}$ yields, for any $t, s \in [s_0 , \infty)$ such that $\sqrt{t - s} \leq c_0 L$,
\begin{align*}
    \sup_{(t' , x) \in [s , t] \times \partial \Lambda_L} P_{\a, \infty} \left(t' , x ; s , y  \right) & \leq \sup_{t' \in [0, c_0]} \frac{C}{2 L^d t'^{d/2}} \exp \left( - \frac{1}{2C  \sqrt{t'}} \right) \\
    & \leq \frac{c}{2L^d} \\
    & \leq \frac{1}{2} \inf_{x \in y + \Lambda_{\sqrt{t-s}}} P_{\a, \infty} \left(t , x ; s , y  \right).
\end{align*}
The proof of~\eqref{eq:10183009} is complete.

We now deduce the lower bound~\eqref{Nash.infprop3.3} from the inequality~\eqref{eq:10183009}. To this end, let us fix a time $t \in (s , s + c_0 L^2)$, set $\ep := \frac12 \inf_{x \in y + \Lambda_{\sqrt{t-s}}} P_{\a, \infty} \left(t , x ; s , y  \right)$, and define the map $P_{\a , \infty}^\ep := P_{\a , \infty} - \ep$. Let us note that the mapping $P_{\a , \infty}^\ep$ solves the parabolic equation
\begin{equation*}
    \partial_t P_{\a , \infty}^\ep (\cdot , \cdot ; s  ,y) - \nabla \cdot \a \nabla P_{\a , \infty}^\ep (\cdot , \cdot ; s ,y) = 0 \hspace{3mm} \mbox{in}~ (s , t) \times \Lambda_L .
\end{equation*}
By the definition of the parameter $\ep$ and the inequality~\eqref{eq:10183009}, the map $P_{\a , \infty}^\ep $ satisfies the boundary estimates
\begin{equation*}
    \left\{ \begin{aligned}
    P_{\a , \infty}^\ep (t' , x ; s  ,y) & \leq 0 \leq P_\a (t' , x ; s  ,y) &&(t' , x) \in [s , t] \times \partial \Lambda_L, \\
    P_{\a , \infty}^\ep(s , x ; s  ,y) & = \indc_{\{ x = y \}} - \ep \leq \indc_{\{ x = y \}} = P_\a(s , x ; s  ,y)  &&x \in \Lambda_L^+.
    \end{aligned} \right.
\end{equation*}
Applying the maximum principle for the parabolic operator $\partial_t - \nabla \cdot \a \nabla$, we obtain the inequality, for any $(t', x) \in [s , t] \times \Lambda_L$,
\begin{equation} \label{eq:11393009}
    P_{\a , \infty}^\ep (t' , x ; s , y) \leq P_\a (t' , x ; s , y).
\end{equation}
Applying the estimate~\eqref{eq:11393009} at time $t' = t$ and using the definition of the parameter $\ep$ yields, for any vertex $x$ satisfying $|x - y| \leq \sqrt{t - s}$,
\begin{equation} \label{eq:11433009}
    \frac 12 P_{\a, \infty} \left( t , x ; s , y \right) \leq P_\a^\ep \left( t , x ; y , s \right) \leq P_\a \left( t , x ; y , s \right).
\end{equation}
Combining the estimate~\eqref{eq:11433009} with the lower bound of Proposition~\ref{prop.NashAronsoninfintevol} implies
\begin{equation*}
    P_\a \left( t , x ; y , s \right) \geq \frac{c}{1 \vee (t-s)^\frac d2}.
\end{equation*}
The proof of the lower bound~\eqref{Nash.infprop3.3} is complete.
\end{proof}

\small
\bibliographystyle{abbrv}
\bibliography{RFRS}

\newcommand{\noop}[1]{} \def\cprime{$'$}
\begin{thebibliography}{10}

\bibitem{AKW16}
S.~Adams, A.~Kister, and H.~Weber.
\newblock Sample path large deviations for laplacian models in $(1+ 1)
  $-dimensions.
\newblock {\em Electron. J. Probab.}, 21, 2016.

\bibitem{AHP20}
M.~Aizenman, M.~Harel, and R.~Peled.
\newblock Exponential decay of correlations in the {$2D$} random field {I}sing
  model.
\newblock {\em J. Stat. Phys.}, 180(1-6):304--331, 2020.

\bibitem{AP19}
M.~Aizenman and R.~Peled.
\newblock A power-law upper bound on the correlations in the {$2D$} random
  field {I}sing model.
\newblock {\em Comm. Math. Phys.}, 372(3):865--892, 2019.

\bibitem{AW1989}
M.~Aizenman and J.~Wehr.
\newblock Rounding of first-order phase transitions in systems with quenched
  disorder.
\newblock {\em Physical {R}eview {L}etters}, 62(21):2503, 1989.

\bibitem{AW89}
M.~Aizenman and J.~Wehr.
\newblock Rounding effects of quenched randomness on first-order phase
  transitions.
\newblock {\em Comm. Math. Phys.}, 130(3):489--528, 1990.

\bibitem{Ar}
D.~G. Aronson.
\newblock Bounds for the fundamental solution of a parabolic equation.
\newblock {\em Bull. Amer. Math. Soc.}, 73:890--896, 1967.

\bibitem{BB07}
P.~N. Balister and B.~Bollob\'{a}s.
\newblock Counting regions with bounded surface area.
\newblock {\em Comm. Math. Phys.}, 273(2):305--315, 2007.

\bibitem{bollobas1991edge}
B.~Bollob{\'a}s and I.~Leader.
\newblock Edge-isoperimetric inequalities in the grid.
\newblock {\em Combinatorica}, 11(4):299--314, 1991.

\bibitem{BCK17}
E.~Bolthausen, A.~Cipriani, and N.~Kurt.
\newblock Exponential decay of covariances for the supercritical membrane
  model.
\newblock {\em Comm. Math. Phys.}, 353(3):1217--1240, 2017.

\bibitem{BGM13}
S.~Boucheron, G.~Lugosi, and P.~Massart.
\newblock {\em Concentration inequalities: A nonasymptotic theory of
  independence}.
\newblock Oxford University Press, 2013.

\bibitem{bovier1992stability}
A.~Bovier and C.~K{\"u}lske.
\newblock Stability of hierarchical interfaces in a random field model.
\newblock {\em J. Stat. Phys.}, 69(1):79--110, 1992.

\bibitem{bovier1993hierarchical}
A.~Bovier and C.~K{\"u}lske.
\newblock Hierarchical interfaces in random media {II}: {T}he {G}ibbs measures.
\newblock {\em Journal of statistical physics}, 73(1):253--266, 1993.

\bibitem{BK94}
A.~Bovier and C.~K\"{u}lske.
\newblock A rigorous renormalization group method for interfaces in random
  media.
\newblock {\em Rev. Math. Phys.}, 6(3):413--496, 1994.

\bibitem{BK96}
A.~Bovier and C.~K{\"u}lske.
\newblock There are no nice interfaces in (2+1)-dimensional {SOS} models in
  random media.
\newblock {\em J. {S}tat. {P}hys.}, 83(3-4):751--759, 1996.

\bibitem{bovier1991stability}
A.~Bovier and P.~Picco.
\newblock Stability of interfaces in a random environment. a rigorous
  renormalization group analysis of a hierarchical model.
\newblock {\em J. Stat. Phys.}, 62(1):177--199, 1991.

\bibitem{BL76}
H.~J. Brascamp and E.~H. Lieb.
\newblock On extensions of the {B}runn-{M}inkowski and {P}r\'{e}kopa-{L}eindler
  theorems, including inequalities for log concave functions, and with an
  application to the diffusion equation.
\newblock {\em J. Functional Analysis}, 22(4):366--389, 1976.

\bibitem{BL75}
H.~J. Brascamp and E.~H. Lieb.
\newblock Some inequalities for {G}aussian measures and the long-range order of
  the one-dimensional plasma.
\newblock In {\em Inequalities}, pages 403--416. Springer, 2002.

\bibitem{BLL75}
H.~J. Brascamp, E.~H. Lieb, and J.~L. Lebowitz.
\newblock The statistical mechanics of anharmonic lattices.
\newblock {\em Bull. Inst. Internat. Statist.}, 46(1):393--404 (1976), 1975.

\bibitem{BK88}
J.~Bricmont and A.~Kupiainen.
\newblock Phase transition in the {$3d$} random field {I}sing model.
\newblock {\em Comm. Math. Phys.}, 116(4):539--572, 1988.

\bibitem{CD08}
F.~Caravenna and J.-D. Deuschel.
\newblock Pinning and wetting transition for (1+ 1)-dimensional fields with
  {L}aplacian interaction.
\newblock {\em Ann. Prob.}, 36(6):2388--2433, 2008.

\bibitem{CO82}
J.~L. Cardy and S.~Ostlund.
\newblock Random symmetry-breaking fields and the {XY} model.
\newblock {\em Physical Review B}, 25(11):6899, 1982.

\bibitem{C83}
J.~Chalker.
\newblock On the lower critical dimensionality of the {I}sing model in a random
  field.
\newblock {\em Journal of Physics C: Solid State Physics}, 16(34):6615, 1983.

\bibitem{C18}
S.~Chatterjee.
\newblock On the decay of correlations in the random field {I}sing model.
\newblock {\em Comm. Math. Phys.}, 362(1):253--267, 2018.

\bibitem{cipriani2019scaling}
A.~Cipriani, B.~Dan, and R.~S. Hazra.
\newblock The scaling limit of the membrane model.
\newblock {\em Ann. Probab.}, 47(6):3963--4001, 2019.

\bibitem{CK12}
C.~Cotar and C.~K{\"u}lske.
\newblock Existence of random gradient states.
\newblock {\em Ann. Appl. Probab.}, 22(4):1650--1692, 2012.

\bibitem{CK15}
C.~Cotar and C.~K{\"u}lske.
\newblock Uniqueness of gradient {G}ibbs measures with disorder.
\newblock {\em Probab. Theory Related Fields}, 162(3-4):587--635, 2015.

\bibitem{Da21}
P.~Dario.
\newblock Convergence of the thermodynamic limit for random-field random
  surfaces.
\newblock {\em arXiv preprint arXiv:2105.03940}, 2021.

\bibitem{DHP20++}
P.~Dario, M.~Harel, and R.~Peled.
\newblock Quantitative disorder effects in low dimensional spin systems.
\newblock {\em arXiv preprint arXiv:2101.01711}, 2021.

\bibitem{De99}
T.~Delmotte.
\newblock Parabolic {H}arnack inequality and estimates of {M}arkov chains on
  graphs.
\newblock {\em Rev. Mat. Iberoamericana}, 15(1):181--232, 1999.

\bibitem{delmotte1999parabolic}
T.~Delmotte.
\newblock Parabolic harnack inequality and estimates of markov chains on
  graphs.
\newblock {\em Revista matem{\'a}tica iberoamericana}, 15(1):181--232, 1999.

\bibitem{DGI00}
J.-D. Deuschel, G.~Giacomin, and D.~Ioffe.
\newblock Large deviations and concentration properties for {$\nabla\phi$}
  interface models.
\newblock {\em Probab. Theory Related Fields}, 117(1):49--111, 2000.

\bibitem{DW2020}
J.~Ding and M.~Wirth.
\newblock Correlation length of two-dimensional random field {I}sing model via
  greedy lattice animal.
\newblock {\em arXiv preprint arXiv:2011.08768}, 2020.

\bibitem{ding20192exponential}
J.~Ding and J.~Xia.
\newblock Exponential decay of correlations in the two-dimensional random field
  {I}sing model.
\newblock {\em Invent. {M}ath.}, pages 1--47, 2021.

\bibitem{Dur2019book}
R.~Durrett.
\newblock {\em Probability---theory and examples}, volume~49 of {\em Cambridge
  Series in Statistical and Probabilistic Mathematics}.
\newblock Cambridge University Press, Cambridge, 2019.
\newblock Fifth edition.

\bibitem{evans2010partial}
L.~C. Evans.
\newblock {\em Partial differential equations}, volume~19.
\newblock American Mathematical Soc., 2010.

\bibitem{fisher1984ising}
D.~S. Fisher, J.~Fr\"{o}hlich, and T.~Spencer.
\newblock The {I}sing model in a random magnetic field.
\newblock {\em J. Statist. Phys.}, 34(5-6):863--870, 1984.

\bibitem{For91}
G.~Forgacs, R.~Lipowsky, and T.~M. Nieuwenhuizen.
\newblock The behavior of interfaces in ordered and disordered systems.
\newblock {\em Phase transitions and critical phenomena}, 14:135--363, 1991.

\bibitem{FrSp}
J.~Fr\"{o}hlich and T.~Spencer.
\newblock Kosterlitz-{T}houless transition in the two-dimensional plane rotator
  and {C}oulomb gas.
\newblock {\em Phys. Rev. Lett.}, 46(15):1006--1009, 1981.

\bibitem{frohlich1981kosterlitz}
J.~Fr{\"o}hlich and T.~Spencer.
\newblock The {K}osterlitz-{T}houless transition in two-dimensional abelian
  spin systems and the {C}oulomb gas.
\newblock {\em Comm. Math. Phys.}, 81(4):527--602, 1981.

\bibitem{F05}
T.~Funaki.
\newblock Stochastic interface models.
\newblock In {\em Lectures on {P}robability {T}heory and {S}tatistics}, volume
  1869 of {\em Lecture Notes in Math.}, pages 103--274. Springer, Berlin, 2005.

\bibitem{FS}
T.~Funaki and H.~Spohn.
\newblock Motion by mean curvature from the {G}inzburg-{L}andau {$\nabla \phi$}
  interface model.
\newblock {\em Comm. Math. Phys.}, 185(1):1--36, 1997.

\bibitem{GS20}
C.~Garban and A.~Sep{\'u}lveda.
\newblock Statistical reconstruction of the {G}aussian free field and {KT}
  transition.
\newblock {\em arXiv preprint arXiv:2002.12284}, 2020.

\bibitem{GOS}
G.~Giacomin, S.~Olla, and H.~Spohn.
\newblock Equilibrium fluctuations for {$\nabla\phi$} interface model.
\newblock {\em Ann. Probab.}, 29(3):1138--1172, 2001.

\bibitem{GLD95}
T.~Giamarchi and P.~Le~Doussal.
\newblock Elastic theory of flux lattices in the presence of weak disorder.
\newblock {\em Physical Review B}, 52(2):1242, 1995.

\bibitem{HWA94}
T.~Hwa and D.~S. Fisher.
\newblock Vortex glass phase and universal susceptibility variations in planar
  arrays of flux lines.
\newblock {\em Physical review letters}, 72(15):2466, 1994.

\bibitem{IM75}
Y.~Imry and S.-K. Ma.
\newblock Random-field instability of the ordered state of continuous symmetry.
\newblock {\em Physical Review Letters}, 35(21):1399, 1975.

\bibitem{kharash2017fr}
V.~Kharash and R.~Peled.
\newblock The {F}r\"ohlich-{S}pencer proof of the
  {B}erezinskii-{K}osterlitz-{T}houless transition.
\newblock {\em arXiv preprint arXiv:1711.04720}, 2017.

\bibitem{KO06}
C.~K\"{u}lske and E.~Orlandi.
\newblock A simple fluctuation lower bound for a disordered massless random
  continuous spin model in {$d=2$}.
\newblock {\em Electron. Comm. Probab.}, 11:200--205, 2006.

\bibitem{KO08}
C.~K\"{u}lske and E.~Orlandi.
\newblock Continuous interfaces with disorder: even strong pinning is too weak
  in two dimensions.
\newblock {\em Stochastic Process. Appl.}, 118(11):1973--1981, 2008.

\bibitem{K07}
N.~Kurt.
\newblock Entropic repulsion for a class of {G}aussian interface models in high
  dimensions.
\newblock {\em Stoch. Process. Appl.}, 117(1), 2007.

\bibitem{K09}
N.~Kurt.
\newblock Maximum and entropic repulsion for a {G}aussian membrane model in the
  critical dimension.
\newblock {\em Ann. Prob.}, 37(2):687--725, 2009.

\bibitem{L20}
P.~Lammers.
\newblock Height function delocalisation on cubic planar graphs.
\newblock {\em Probability Theory and Related Fields}, 182(1-2):531--550, 2022.

\bibitem{LDS07}
P.~Le~Doussal and G.~Schehr.
\newblock Disordered free fermions and the {C}ardy-{O}stlund fixed line at low
  temperature.
\newblock {\em Physical Review B}, 75(18):184401, 2007.

\bibitem{LM98}
J.~L. Lebowitz and A.~E. Mazel.
\newblock Improved {P}eierls argument for high-dimensional {I}sing models.
\newblock {\em J. Statist. Phys.}, 90(3-4):1051--1059, 1998.

\bibitem{magazinov2020concentration}
A.~Magazinov and R.~Peled.
\newblock Concentration inequalities for log-concave distributions with
  applications to random surface fluctuations.
\newblock {\em Ann. Probab.}, 50(2):735--770, 2022.

\bibitem{MP15}
P.~Mi{\l}o\'{s} and R.~Peled.
\newblock Delocalization of two-dimensional random surfaces with hard-core
  constraints.
\newblock {\em Comm. Math. Phys.}, 340(1):1--46, 2015.

\bibitem{nash1958continuity}
J.~Nash.
\newblock Continuity of solutions of parabolic and elliptic equations.
\newblock {\em American Journal of Mathematics}, 80(4):931--954, 1958.

\bibitem{N90SC}
T.~Nattermann.
\newblock Scaling approach to pinning: Charge density waves and giant flux
  creep in superconductors.
\newblock {\em Physical review letters}, 64(20):2454, 1990.

\bibitem{OS95}
H.~Orland and Y.~Shapir.
\newblock A disorder-dependent variational method without replicas: Application
  to the random phase sine-{G}ordon model.
\newblock {\em EPL (Europhysics Letters)}, 30(4):203, 1995.

\bibitem{RLDS12}
Z.~Ristivojevic, P.~Le~Doussal, and K.~J. Wiese.
\newblock Super-rough phase of the random-phase sine-gordon model: Two-loop
  results.
\newblock {\em Physical Review B}, 86(5):054201, 2012.

\bibitem{ruelle1999statistical}
D.~Ruelle.
\newblock {\em Statistical mechanics: Rigorous results}.
\newblock World Scientific, 1999.

\bibitem{Sa2003}
H.~Sakagawa.
\newblock Entropic repulsion for a {G}aussian lattice field with certain finite
  range interaction.
\newblock {\em J. Math. Phys.}, 44(7):2939--2951, 2003.

\bibitem{Sch21}
F.~Schweiger.
\newblock On the membrane model and the discrete {B}ilaplacian.
\newblock {\em PhD thesis, Rheinische Friedrich-Wilhelms-Universit{\"a}t Bonn},
  2021.

\bibitem{Sh}
S.~Sheffield.
\newblock Random surfaces.
\newblock {\em Ast\'{e}risque}, (304):vi+175, 2005.

\bibitem{T13}
{\'A}.~Tim{\'a}r.
\newblock Boundary-connectivity via graph theory.
\newblock {\em Proceedings of the American Mathematical Society},
  141(2):475--480, 2013.

\bibitem{TD90}
J.~Toner and D.~DiVincenzo.
\newblock Super-roughening: A new phase transition on the surfaces of crystals
  with quenched bulk disorder.
\newblock {\em Physical Review B}, 41(1):632, 1990.

\bibitem{VK08}
A.~C. van Enter and C.~K{\"u}lske.
\newblock Nonexistence of random gradient {G}ibbs measures in continuous
  interface models in {$d=2$}.
\newblock {\em Ann. Appl. Probab.}, 18(1):109--119, 2008.

\bibitem{V06}
Y.~Velenik.
\newblock Localization and delocalization of random interfaces.
\newblock {\em Probab. Surv.}, 3:112--169, 2006.

\bibitem{VF84}
J.~Villain and J.~F. Fernandez.
\newblock Harmonic system in a random field.
\newblock {\em Zeitschrift f{\"u}r Physik B Condensed Matter}, 54(2):139--150,
  1984.

\end{thebibliography}

\end{document}